\documentclass[12pt, a4paper]{article}

\usepackage{latexbase}
\usepackage{latexscience}
\usepackage{custom_commands}

\begin{document}

\title{Algorithmic complexity of $\beta$-expansions and application to A/D conversion}
\author{Valentin Abadie \and Helmut Boelcskei}
\date{}
\maketitle


We establish diverse relationships between the algorithmic (Kolmogorov) complexity of the prefixes of any binary expansion and $\beta$-expansions. These relationships allow to develop intuitions on the complexity behavior of $\beta$-expansions, and raise problems related to compressibility of binary sequences generated in the context of A/D conversion relying on $\beta$-expansions. Our last contribution is to solve these problems.

\tableofcontents

\begingroup
\let\clearpage\relax

\section{Introduction}

\noindent Measuring the complexity of real numbers is of major importance in computer science. Consider a non-computable real number $s$, i.e., a real number which cannot be stored on a computer \cite[Section 10]{copeland2004computable}. One can store only an approximation of $s$, for instance by considering a finite length bitstring $y$ representing a prefix $x$ of the binary expansion $\xf$ of $s$. For a fixed approximation error $\varepsilon>0$, the required length of $y$, as a function of $\varepsilon$, depends on the \textit{algorithmic complexity} of the prefix $x$ of the binary expansion achieving error $\varepsilon$. The \textit{algorithmic complexity} of a binary sequence $x$, often referred to as \textit{Kolmogorov complexity}, is the length of the smallest binary sequence $y$, for which there exists an algorithm, such that when presented with $y$ as input, delivers $x$ as output \cite[Definition 10.1]{balcazar2012structural}\cite[Section 14.2]{cover1999elements}\cite[Definition 2.1.2]{li2008introduction}. The algorithmic complexity of the binary expansions of real numbers has been widely studied, but the algorithmic complexity of expansions in bases other than $2$ remains poorly understood. Several papers have established an equivalence between the algorithmic complexity of the expansions in different bases $q \in \N$ \cite[Theorem 6.1]{calude2005randomness}\cite[Theorem 5.1]{hertling1998randomness}\cite[Theorem 3]{staiger2002kolmogorov}. Here, we study the algorithmic complexity of expansions in noninteger bases, which display a much more sophisticated behavior. This type of expansions are often referred to as $\beta$-expansions \cite[(1)]{parry1960beta}\cite[Section 4]{renyi1957representations}. $\beta$-expansions display some redundancy properties that are used to design robust A/D converters \cite{daubechies2002beta}\cite{daubechies2006imperfect}\cite{daubechies2010GoldenRatioEncoder}. However, we discover that the algorithmic complexity of $\beta$-expansions can be much larger than the algorithmic complexity of binary expansions of the same number. This engenders problems of compressibility of the sequences generated by the aforementioned A/D converters. Fortunately, we find a fast algorithm to fix this issue, converting any $\beta$-expansion of potentially large algorithmic complexity into another $\beta$-expansion which algorithmic complexity is equal to the algorithmic complexity of the binary expansions.

\endgroup

\begingroup
\let\clearpage\relax

\subsection{Notation}

$\N$ denotes the set of natural numbers, $\N_0 := \N \cup \{0\}$, and $\N_2 := \N \backslash \{1\}$. $\Q$ stands for the set of rational numbers, and $\R$ for the set of real numbers. For a set $A$, $\card A$ denotes its cardinality. If $\card X = \infty$, we may write $\card X = \aleph_0$ if $X$ is countable, and $\card X = 2^{\aleph_0}$ if $X$ is uncountable. 
Let $f : X \to Y$ be a map from a set $X$ into a set $Y$. For a given subset $A \subset Y$, we let $f(A) := \{ f(x) : x \in A\}$ and $f^{-1}(A) := \{ x \in X : f(x) \in A\}$.

$\{0,1\}^\N$ denotes the set of infinite binary sequences, $\{0,1\}^\ast$ denotes the set of binary sequences of finite, but otherwise arbitrary length, and $\{0,1\}^n$ denotes the set of binary sequences of length $n \in \N$. 
$\len x$ refers to the length of $x$, i.e. for $x \in \{0,1\}^n$, we have $\len x = n$, and by convention, we set $\len x = \infty$ for $x \in \{0,1\}^\N$.
$\epsilon$ denotes the empty sequence. 
For a sequence $x \in \{0,1\}^\ast \cup \{0,1\}^\N$, $x_n$ denotes the $n$-th element of $x$, and $x_{1:n}$ is defined to be the $n$-prefix $x_1 \ldots x_n$ of $x$. 
For $x \in \{0,1\}^\ast$ and $y \in \{0,1\}^\ast \cup \{0,1\}^\N$, we write $x \sqsubset y$ if $x$ is a prefix of $y$. 
For $x,y \in \{0,1\}^\ast \cup \{0, 1\}^\N$, we define the lexicographical ordering as $x <_L y$ to express that there exists $j \in \N$ such that $x_i = y_i$ for all $i < j$, and $y_j < x_j$.
For a subset $A \subseteq \{0,1\}^\ast \cup \{0, 1\}^\N$, we define $\lexmax(A)$ to be the lexicographically largest sequence in $A$, and $\lexmin(A)$ to be the lexicographically smallest sequence in $A$.
For $x \in \{0,1\}^\ast$ and $y \in \{0,1\}^\ast \cup \{0,1\}^\N$, $x y$ denotes the concatenation of $x$ and $y$. 
For $x^{(1)},\ldots, x^{(n)} \in \{0,1\}^\ast$, $\coprod_{i=1}^n x^{(i)}$ stands for the concatenation $x^{(1)} x^{(2)} \ldots \hspace{0.2mm} x^{(n)}$. 
For an infinite family $(x^{(i)})_{i\in \N} \in \{0,1\}^\ast$, $\coprod_{i=1}^\infty x^{(i)} \in \{0,1\}^\N$ denotes the concatenation of the elements in $(x^{(i)})_{i \in \N}$. Finally, for $x \in \{0,1\}^\ast$, we set $x^0 = \epsilon$, $x^n := \coprod_{i=1}^n x$, for $n \in \N$, and $x^\infty := \coprod_{i=1}^\infty x$, to denote the $n$-fold and infinite repetitions, respectively. 
For instance, $0^\infty$ denotes the infinite sequence containing only zeroes, and $1^5 = 11111$. 
We may combine repetitions and concatenations in a single expression, as for example $xy^5 = xyyyyy$, for $x,y \in \{0,1\}^\ast$. So that there is no confusion, we may use parentheses. For instance, $(10)^5 = 1010101010$, while $10^5 = 100000$. 
Following \cite[Ch. 1.4]{li2008introduction}, we map $\{0,1\}^\ast$ one-to-one onto $\N_0$ by indexing each string through the quasi-lexicographical ordering \cite[Section 1.1]{calude2002information}, i.e.,
\begin{equation}\label{eq:quasilexicographic enumeration}
	(\epsilon,0), (0,1), (1,2), (00,3), (01,4), (10,5), (11,6), \ldots,
\end{equation}
where each of the above pairs contains first an element $x \in \{0,1\}^\ast$ and second the index of $x$ in the quasi-lexicographical ordering of $\{0,1\}^\ast$. For $n \in \N$, we denote by $\bin(n) \in \{0,1\}^\ast$ the binary sequence associated this way. For $x \in \{0,1\}^\ast$, we define $\bar x := 1^{|x|}0x$, which is often referred to as prefixing $x$, i.e. $\{\bar x : x \in \{0,1\}^\ast\}$ forms a prefix-free set. For $x,y \in \{0,1\}^\ast$, we set $\langle x,y \rangle := \bar x y$. Note that $|\langle x,y \rangle| = 2 |x| + |y| + 1$. We generalise $\langle \cdot \rangle$ as follows. For $n \in \N$, $n \geq 3$, and $x_1, x_2, \ldots, x_n \in \{0,1\}^\ast$, we define $\langle x_1, \ldots, x_n \rangle := \langle \langle x_1, \ldots, x_{n-1} \rangle , x_n \rangle$. For a finite set $A = \{a_1, \ldots, a_k \} \subseteq \{0,1\}^\ast$, we define $\langle A \rangle := \langle a_1, \ldots, a_k \rangle$. 

In order to make a clear distinction between infinite and finite sequences, we adopt boldface fonts for infinite sequences $\xf \in \{0,1\}^\N$, and normal fonts for finite sequences $x \in \{0,1\}^\ast$, for $n \in \N$.

\endgroup

\begingroup
\let\clearpage\relax

\subsection{Binary expansions}
\label{sec:binary expansion}
  
Representing real numbers in terms of sequences of bits is of major interest in computer science and electrical engineering. The most common way to effect such a conversion is the binary expansion. Indeed, every real number $s \in [0,1]$ can be written as a sum of negative powers of $2$. For example, $s = 0.75$ can be represented as
  \begin{equation}\label{eq:greedy binary expansion example}
      s = 0.75 =  \frac{1}{2} + \frac{1}{4} = 2^{-1} + 2^{-2} = 1 \times 2^{-1} + 1 \times 2^{-2} + 0 \times 2^{-3} + 0 \times 2^{-4} + \ldots
  \end{equation}
Here, the sequence $110^\infty$ is called a binary expansion of $s=0.75$. We conspicuously say ``a binary expansion'' rather than ``the binary expansion'' as the expansion is not unique. Indeed, with 
  \begin{equation}
      \frac{1}{4} = \frac{\frac{1}{8}}{\frac{1}{2}} = \frac{\frac{1}{8}}{1 - \frac{1}{2}} = \sum_{i=3}^\infty 2^{-i},
  \end{equation}
  one has equivalently
  \begin{equation}\label{eq:lazy binary expansion example}
      s = 0.75 = \frac{1}{2} + \frac{1}{4} = 2^{-1} + \sum_{i=3}^\infty 2^{-i} = 1\times 2^{-1} +  0\times 2^{-2} + \sum_{i=3}^\infty 1\times 2^{-i},
  \end{equation}
  so that $\xf = 101^\infty$ constitutes an alternative binary representation of $s$. Generally speaking, for every given $s \in [0,1]$, there exist either one or two binary expansions. Following
  \cite[Section 1]{staiger2002kolmogorov}, we refer to it as 2-ambiguity of the binary expansion. The real numbers with precisely two binary expansions form a subset of the rational numbers called the dyadic numbers \cite[Section 1.5, Problem 44]{royden1968real}. Dyadic numbers are defined to be the real numbers that have a binary expansion ending with $0^\infty$. Dyadic numbers do not include, for instance, $s = 1/3$, which indeed displays only one binary expansion. The definition of dyadic numbers immediately implies that they have two expansions: one that ends with $0^\infty$, and another one that ends with $1^\infty$. This can be seen as follows. Consider a dyadic number $s \neq 0$. By definition, $s$ has a binary expansion that ends with $0^\infty$. Then, there exists a binary sequence $x \in \{0,1\}^\ast$, of length $n = |x|$, such that $\xf = x 1 0^\infty$ is a binary expansion of $s$, i.e.,
  \begin{equation}\label{eq:greedy expansion ending with infinitetail of zeros}
	s = x_1 \times 2^{-1} + \ldots + x_n \times 2^{-n} + 1 \times 2^{-(n+1)} + 0 \times 2^{-(n+2)} + 0 \times 2^{-(n+3)} + \ldots
  \end{equation}
As $1$ satisfies the following algebraic equation
\begin{equation}\label{eq:algebraic relation between 1 and powers of one half}
	1 = \sum_{i=1}^\infty 2^{-i},
\end{equation}
one can replace the last digit $1$ of the binary expansion of $s$ according to
\begin{align}
	s &= x_1 \times 2^{-1} + \ldots + x_n \times 2^{-n} + \left(\sum_{i=1}^\infty 2^{-i}\right) \times 2^{-(n+1)}\\
	&= x_1 \times 2^{-1} + \ldots + x_n \times 2^{-n} + \sum_{i=n+2}^\infty 2^{-i}\\
	&= x_1 \times 2^{-1} + \ldots + x_n \times 2^{-n} + 0 \times 2^{-(n+1)} + 1 \times 2^{-(n+2)} +  1 \times 2^{-(n+3)} + \ldots \label{eq:lazy expansion ending with infinitetail of ones}
  \end{align}
This establishes that $\xf' = x 0  1^\infty$ is also a binary expansion of $s$. The lexicographically larger of these two expansions  $\xf = x  1  0^\infty$ is referred to as the \newname{greedy binary expansion of $s$}, and the lexicographically smaller $\xf' = x  0  1^\infty$ is called the \newname{lazy binary expansion}. When the binary expansion of $s$ is unique, the greedy and the lazy binary expansion coincide. Both the greedy and the lazy binary expansion have an associated algorithm generating them. We detail the algorithm generating the first $n$ bits of the greedy binary expansion of $s$ as follows \cite[Section 2]{dajani2002greedy}.
\begin{algorithm}[ht]
	\caption{Greedy algorithm for binary expansion}\label{alg:greedy algorithm binary expansion}
	\begin{algorithmic}[1]
		\Require $s \in [0,1]$, $n \in \N$
		\State $r \gets s$
		\State Initialize $x$ as the empty string $\epsilon$
		\For {$i=1, \ldots,n$}\label{alg1:line 3}
			\If{ $r < 0.5$,} $b \gets 0$ \label{alg1:line 4}
			\Else{ $b \gets 1$}
			\EndIf
			\State $x \gets xb$
			\State $r \gets 2 \, r - b$
		\EndFor
		\State \Return $x$
	\end{algorithmic}
\end{algorithm}

\noindent By way of example, we consider the case $s = 0.75$, $n = 4$ to illustrate the algorithm.
  \begin{enumerate}[label=(\arabic*)]
      \item $r \gets 0.75$, $x$ is initialized as the empty string $\epsilon$, the \textbf{for} loop (line \ref{alg1:line 3}) is entered, with $i=1$.
      \item $r =0.75 \geq 0.5$, so $b \gets 1$, $x \gets \epsilon b = 1$, $r \gets 2 r -  b = 2 \times 0.75 - 1 = 0.5$, the algorithm continues the \textbf{for} loop with $i = 2$.
      \item $r = 0.5 \geq 0.5$, so $b \gets 1$, $x \gets 1 b = 11$, $r \gets 2 r - b = 2\times 0.5 - 1 = 0$, the algorithm continues the \textbf{for} loop with $i = 3$.
      \item $r = 0 < 0.5$, so $b \gets 0$, $x \gets 11 b = 110$, $r \gets 2 r - b = 2\times 0 - 0 = 0$,  the algorithm continues the \textbf{for} loop with $i = 4$.
      \item $r = 0 < 0.5$, so $b \gets 0$, $x \gets 110 b = 1100$, $r \gets 2 r - b = 2\times 0 - 0 = 0$, the algorithm exits the \textbf{for} loop, as $i = 4 = n$.
      \item The algorithm returns $x = 1100$.
  \end{enumerate}
Algorithm \ref{alg:greedy algorithm binary expansion} is therefore seen to, indeed, generate the first $n=4$ bits of the greedy binary expansion of $s=0.75$. The variable $r$ in Algorithm \ref{alg:greedy algorithm binary expansion} is seen to be the approximation error of the real number $s$ by the successive finite prefixes of its greedy expansion. To generate the lazy binary expansion, instead one only needs to change the condition $r < 0.5$ in line \ref{alg1:line 4} of Algorithm \ref{alg:greedy algorithm binary expansion} to $r \leq 0.5$ \cite[Section 2]{dajani2002greedy}.
  
\subsection{\texorpdfstring{$\beta$-expansions}{}}

\label{sec:beta expansions}
  
Binary expansions are not the only way to represent real numbers as infinite sequences of bits. One can replace the base 2 in $s = \sum_{i=1}^\infty \xf_i 2^{-i}$ by $\beta \in (1,2]$ to get so-called $\beta$-expansions, i.e., representations of real numbers as sums of negative powers of $\beta$, originally introduced in \cite[Section 4, Example 4]{renyi1957representations}. The set of real numbers that can be represented by $\beta$-expansions is larger: while every real number in $[0,1]$ can be represented by a binary expansion, it turns out that for $\beta \in (1,2]$, every number in $I_\beta := [0, (\beta-1)^{-1}]$ can be represented by a $\beta$-expansion \cite[Section 1]{dajani2002greedy}. Note that $I_2 = [0,1]$, so this is consistent with the binary case. 

$\beta$-expansions exhibit a much richer structure than binary expansions, specifically,  the 2-ambiguity can turn into an $\infty$-ambiguity. An example illustrating this aspect is the $G$-expansion of $s=1$, where $G$ is the golden ratio, satisfying $G^2 = G + 1$. Indeed, it follows that
  \begin{equation}
      1 = \frac{1}{G} + \frac{1}{G^2} = 1 \times \frac{1}{G} + 1 \times \frac{1}{G^2} + 0 \times \frac{1}{G^3} + 0 \times \frac{1}{G^4} + \cdots
  \end{equation}
 and hence $\xf = 1100\ldots$ is a $G$-expansion of $s=1$. An alternative representation of $s=1$ in base $\beta=G$ is obtained by noting that $\frac{1}{G^2} = \frac{1}{G^3} + \frac{1}{G^4}$, and hence
  \begin{equation}
      1 = \frac{1}{G} + \frac{1}{G^3} + \frac{1}{G^4} = 1 \times \frac{1}{G} + 0 \times \frac{1}{G^2} + 1 \times \frac{1}{G^3} + 1 \times \frac{1}{G^4} + 0 \times \frac{1}{G^5} + \cdots
  \end{equation}
  so that $\xf'=101100\ldots$ is also a $G$-expansion of $s=1$. By iterating this heuristic, one can show that for every $n \geq 0$, the sequence $\xf^{(n)} := (10)^{n}110^\infty$ is a $G$-expansion of $s=1$. There are hence infinitely many different $G$-expansions of $s=1$, all obtained by exploiting the algebraic equation
  \begin{equation}
	1 = \frac{1}{G} + \frac{1}{G^2},
  \end{equation}
  which is of the same flavor as the algebraic equation (\ref{eq:algebraic relation between 1 and powers of one half}).

  We can generalize the idea underlying the example $G = \frac{1+\sqrt{5}}{2}$ to all pairs $(s,\beta)$ satisfying
  \begin{enumerate}[label = (\alph*)]
	\item $\beta$ satisfies an equation of the form 
	\begin{equation}
		1 = \sum_{i=1}^N \beta^{-n_i},
	\end{equation}
	where $N, n_1, \ldots, n_N \in \N$ and the $n_i$ are pairwise distinct.
	\item there is a $\beta$-expansion of $s$ ending with $0^\infty$.
  \end{enumerate}
  Let $s, \beta$ satisfying (a) and (b), and let $\xf$ be a $\beta$-expansion of $s$ ending with $0^\infty$. Then, there exists $j \in \N$ and $x \in \{0,1\}^{j-1}$ such that $\xf = x10^\infty$. Also, note that as there exists $N,n_1, \ldots,n_N \in \N$ such that $1 = \sum_{i=1}^N \beta^{-n_i}$, there exists $y \in \{0,1\}^{n_N}$, satisfying $y_{n_i} = 1$ for all $i \in \{1, \ldots, N\}$, and $y_k = 0$ if $k \notin \{n_1, \ldots, n_N\}$, such that $1 = \sum_{i=1}^{n_N} y_i \beta^{-i}$. Therefore,
  \begin{align}
	s &= \sum_{i=1}^\infty \xf_i \beta^{-i} = \sum_{i=1}^{j-1} x_i \beta^{-i} + \beta^{-j} + \sum_{i=1}^j 0 \beta^{-i}\\
		&=\sum_{i=1}^{j-1} x_i \beta^{-i} + \beta^{-j} = \sum_{i=1}^{j-1} x_i \beta^{-i} + \beta^{-j} \sum_{i=1}^{n_N} y_i \beta^{-i}\\
		&=\sum_{i=1}^{j-1} x_i \beta^{-i} + \sum_{i=1}^{n_N} y_i \beta^{-i-j}.
  \end{align}
  We get that $\xf' = x_{1:j-1}y0^\infty$ is also a $\beta$-expansion of $s$. By iterating the heuristic, we get that $x_{1:j-1}y_{1:n_M-1}^{n-1}y0^\infty$ is a $\beta$-expansion of $s$ for every $n \in \NN$. Hence, $s$ has infinitely many $\beta$-expansions.

  With different proof techniques, the $\infty$-ambiguity property has been established far beyond the scope of the proof above. It was notably proven that for all $\beta \in (1,2)$ and Lesbesgue-almost all $s \in I_\beta$, $s$ has $2^{\aleph_0}$ different $\beta$-expansions \cite[Theorem 1]{sidorov2003almost}, and that for all $\beta \in (1,G)$, every $s \in \mathring I_\beta$ has $2^{\aleph_0}$ different $\beta$-expansions \cite[Theorem 3]{erdos1990characterization}.
  
  One can now define the notions of greedy and lazy expansions for arbitrary $\beta \in (1,2]$, in the same manner as in the binary case. Specifically, motivated by the insight in \cite[IV-A]{daubechies2006imperfect}, we implicitely define the greedy $\beta$-expansion of $s \in I_\beta$ as the lexicographically largest $\beta$-expansion, and the lazy $\beta$-expansion as the lexicographically smallest $\beta$-expansion. In contrast to the binary expansion case, there are, in general, infinitely many $\beta$-expansion between these two extreme cases. Both the greedy and the lazy $\beta$-expansion have an associated algorithm generating them. Algorithm 
  \ref{alg:greedy algorithm beta expansion} below generalizes Algorithm \ref{alg:greedy algorithm binary expansion} to deliver the first $n \in \N$ bits of the greedy $\beta$-expansion of $s$ for $\beta \in (1,2]$ \cite[Section 1]{dajani2002greedy}.

  \begin{algorithm}[ht]
	\caption{Greedy algorithm for $\beta$-expansion}\label{alg:greedy algorithm beta expansion}
	\begin{algorithmic}[1]
		\Require $\beta \in (1,2]$, $s \in I_\beta$, $n\in\N$
		\State $r \gets s$, 
		\State Initialize $x$ as the empty string $\epsilon$
		\For {$i = 1,\ldots,n$}
			\If{ $r < \beta^{-1}$ } $b \gets 0$  \label{alg2:line 4}
			\Else{ $b \gets 1$}
			\EndIf
			\State $x \gets xb$
			\State $r \gets \beta \, r -  b$
		\EndFor
		\State \Return $x$
	\end{algorithmic}
	\end{algorithm}

	To generate the lazy $\beta$-expansion, instead one only needs to change the condition $r < \beta^{-1}$ in line \ref{alg2:line 4} of Algorithm \ref{alg:greedy algorithm beta expansion} to $r \leq \beta^{-1}(\beta - 1)^{-1}$ \cite[Section 2]{dajani2002greedy}. As any other $\beta$-expansion is lexicographically between these two extreme cases, we can come up with an algorithm that covers them all. This algorithm is referred to as the random $\beta$-expansion algorithm \cite[Section 1]{dajani2003random}, and is summarized in Algorithm \ref{alg:random algorithm beta expansion}.


	\begin{algorithm}[ht]
		\caption{Random $\beta$-expansion algorithm}\label{alg:random algorithm beta expansion}
		\begin{algorithmic}[1]
			\Require $\beta \in (1,2]$, $s \in I_\beta$, $n\in\N$
			\State $r \gets s$, 
			\State Initialize $x$ as the empty string $\epsilon$
			\For {$i = 1,\ldots,n$}
				\If{ $r < \beta^{-1}$ } $b \gets 0$  \label{alg3:line 4}
				\ElsIf{$r > \beta^{-1}(\beta-1)^{-1}$}  $b \gets 1$
				\Else{ $b \gets 0$ or $1$},  \label{alg3:line 6}
				\EndIf
				\State $x \gets xb$
				\State $r \gets \beta \, r -  b$
			\EndFor
			\State \Return $x$
		\end{algorithmic}
	\end{algorithm}
	Note that in the case $\beta=2$, line \ref{alg3:line 6} can occur only if $r = 0.5$, which is the source of the 2-ambiguity if binary expansions. This algorithm gives a remarkable illustration on how multiplicity of $\beta$-expansions can be turned into a tractable nondeterministic algorithm.

  This redundancy of $\beta$-expansions of a given real number can be exploited in practical applications, e.g. in A/D-conversion. Specifically, it was shown in \cite{daubechies2006imperfect} that A/D-conversion based on $\beta$-expansions can yield arbitrary precision even in the presence of imperfect quantizers. To illustrate this, consider the problem of representing $s \in [0,1]$ as a binary sequence, i.e., we turn the analog quantity $s$ into a binary sequence $x$, for $\beta = 2$, say in a greedy manner. If one builds a physical device to accomplish this through Algorithm \ref{alg:greedy algorithm binary expansion}, the operation $r < 0.5$ is by virtue of requiring infinite precision impossible to realize in practice. Any physical device that has to realize the thresholding operation $r<0.5$ will fail from time to time. These failure can happen when $r \in [0.5 - \varepsilon, 0.5 + \varepsilon]$, for some small $\varepsilon > 0$, which correspond to the precision limit of the physical device. However, the random $\beta$-expansion delivered by Algorithm \ref{alg:random algorithm beta expansion} can be used to overcome this issue. Suppose in more general terms that we dispose of a device $T$ that can be tuned to compare an input $r \geq 0$ to a threshold $t>0$ fixed by the user, within precision $\varepsilon$, i.e., the device returns $0$ if $r < t - \varepsilon$, $1$ if $r > t + \varepsilon$, and either $0$ or $1$ if $r \in [t - \varepsilon, t + \varepsilon]$. We denote by $T(r) \in \{0,1\}$ the output of the device. Then, one can conceive an algorithm that uses this physical device that attempts to generate a $\beta$-expansion, as follows.

  \begin{algorithm}[ht]
	\caption{Physical algorithm for $\beta$-expansion}\label{alg:physical algorithm beta expansion}
	\begin{algorithmic}[1]
		\Require $\beta \in (1,2]$, $s \in I_\beta$, $n\in\N$
		\State $r \gets s$, 
		\State Initialize $x$ as the empty string $\epsilon$
		\For {$i = 1,\ldots,n$}
			\State $b \gets T(r)$  \label{alg4:line 4}
			\State $x \gets xb$
			\State $r \gets \beta \, r -  b$
		\EndFor
		\State \Return $x$
	\end{algorithmic}
	\end{algorithm}
  
  One can choose $\beta \in (1,2)$ and $t > 0$ so that $[t - \varepsilon, t + \varepsilon] = [\beta^{-1}, \beta^{-1}(\beta-1)^{-1}]$. Then, line \ref{alg4:line 4} in Algorithm \ref{alg:physical algorithm beta expansion} is precisely equivalent to the lines \ref{alg3:line 4}-\ref{alg3:line 6} of Algorithm \ref{alg:random algorithm beta expansion}. It follows that Algorithm \ref{alg:physical algorithm beta expansion} indeed delivers a valid $\beta$-expansion, provided that $\beta$ and $t$ satisfy the above mentionned constraints. Note that $\beta$ cannot be chosen equal to $2$, as in this case $[\beta^{-1}, \beta^{-1}(\beta-1)^{-1}] = \{0.5\}\neq [t - \varepsilon, t + \varepsilon]$ for any $t,\varepsilon > 0$, thus this technique cannot be achieved in the binary expansion framework.
  
  We have seen that $\beta$-expansions lead to robust binary representations. However, these binary representations can be very complex and very long compared to the binary representation. Indeed, line \ref{alg4:line 4} in Algorithm \ref{alg:physical algorithm beta expansion} appears to generate randomness, which points to that fact that the generated $\beta$-expansion is structureless. In this case $\beta$-expansions to represent real numbers would be of poor interest, as the binary expansions would be much more efficient in terms of storage. Even the apparently simple greedy $\beta$-expansion empirically generates sequences that do not present obvious regularities, as depicted on Table \ref{tab:different beta expansions of}.

  \begin{table}[H]
	\centering
		\begin{tabular}{c|l}
			$\beta$ & Greedy $\beta$-expansions of $s = 0.75$\\
			\hline
			2 & 11000000000000000000000000000000000000000000000000\\
			1.01 & 00000000000000000000000000001000000000000000000000\\
			1.2 & 01000000000000010000000000000000000100000000000000\\
			1.5 & 10000010010010100000000010000001000010000001001001\\
			1.8 & 10100010101000000110101000011000011010011000010000\\
			1.99 & 10111110001001001001010001100011010000100000111010
		\end{tabular}
		\caption{First 50 bits of the greedy $\beta$-expansion of $s=0.75$ generated by Algorithm \ref{alg:greedy algorithm beta expansion}, for different values of $\beta$.}\label{tab:different beta expansions of}
	\end{table}
  
  On the above table, while the greedy $\beta$-expansion of $0.75$ appears to be very structured for $\beta=2$, it is not so obvious for any other value of $\beta$. To ask whether one can quantify the differences of algorithmic complexity, also known as algorithmic complexity, between $\beta$-expansions and binary expansions, for different values of $\beta \in (1,2]$.

\subsection{Algorithmic complexity}

\label{subsec:algorithmic complexity}

We proceed to building up the definitions related to algorithmic complexity. We will remain on the surface of the topic, but we refer the interested reader to \cite{li2008introduction} for an in-depth study. An algorithm is a finite sequence of nonambiguous elementary mathematical operations. We consider a special class of algorithms, called \newname{effective algorithms}, that are defined precisely by the formalism of Turing machines \cite{turing1936computable}. Effective algorithms are essentially algorithms that can be run on digital computers, assuming infinite memory. As a counter-example, Algorithm \ref{alg:greedy algorithm binary expansion} is not an effective algorithm, as line \ref{alg1:line 4} uses the comparison of the real variable ``$r$'' to the threshold $0.5$, which in full generality requires infinite precision, the latter being impossible to realize in practice, and a fortiori not possible on digital computers. Effective algorithms can manipulate finite sequences of nonnegative integers and perform operations such as shifts, copying and pasting which can then be used to define more complex operations on integers and rationals, such as addition, multiplication, divison and comparison. We can equip an effective algorithm $E$ with an infinite binary sequence $\xf \in \{0,1\}^\N$. The effective algorithm can then use this sequence $\xf \in \{0,1\}^\N$ as a help to perform its computations. We write $E^\xf$ to denote that $E$ is equipped with $\xf$. In practice, such an infinite binary $\xf \in \{0,1\}^\N$ is the mathematical model of a dataset.

Effective algorithms lead to the definition of a special class of functions. Let $X,Y$ be two countable sets, such as $\{0,1\}^\ast$, $\N$, $\Q$, or a finite product of these sets. A function $\varphi : X \times \{0,1\}^\N \to Y$ is called \newname{computable} if there exists an effective algorithm $E_\varphi$ that, when presented with input $x \in A$ and equipped with $\xf \in \{0,1\}^\N$, delivers $\varphi(x, \xf)$ as an output. A function $\psi : X \to Y$ are said to be \newname{computable relatively to} $\xf \in \{0,1\}^\N$ if there exists a computable function $\varphi : X \times \{0,1\}^\N \to Y$ such that $\psi(x) = \varphi(x, \xf)$, for all $x \in X$. In particular, we say that $\psi : X \to Y$ is \newname{computable} if $\psi$ is computable relatively to $0^\infty$.

A universal computable function is a computable function $U : \{0,1\}^\ast \times \{0,1\}^\N \to \{0,1\}^\ast$ such that for every given computable function $\varphi: \{0,1\}^\ast \times \{0,1\}^\N \to \{0,1\}^\ast$, there exists a $p \in \{0,1\}^\ast$, called program for $\varphi$, so that $\varphi(\cdot, \cdot) = U(\langle p, \cdot \rangle, \cdot)$. The existence of universal computable functions was established in \cite{turing1936computable}. The concept of universal computable function is the base to the definition of algorithmic complexity, that we now introduce. 
\begin{definition}
	Let $U$ be a universal computable function, and $x \in \{0,1\}^\ast$. The algorithmic complexity of $x$ with respect to $U$ relatively to $\xf \in \{0,1\}^\N$ is defined as the length of the shortest (finite) binary sequence $x^\ast$, called canonial sequence for $x$, which satisfies $U(x^\ast, \xf) = x$. It is denoted $K_U[x|\xf] = |x^\ast|$.
\end{definition}
Note that the above definition of algorithmic complexity is given with respect to some universal computable function, and relatively to some sequence $\xf \in \{0,1\}^\N$. Following \cite[Definition 2.1.2]{li2008introduction}, we fix a universal computable function $\Uref$ throughout the paper, and we define the \newname{algorithmic complexity of $x \in \{0,1\}^\ast$ relatively to $\xf \in \{0,1\}^\N$} as
\begin{equation}\label{eq:definition of relative algorithmic complexity}
	K[x|\xf] := K_{\Uref}[x|\xf], \ \text{for all} \ x \in \{0,1\}^\ast, \ \xf \in \{0,1\}^\N.
\end{equation}
Moreover, we can define the absolute measure of \newname{algorithmic complexity of $x \in \{0,1\}^\ast$} by
\begin{equation}
	K[x] := K[x | 0^\infty].
\end{equation}
Algorithmic complexity reflects the structure (or the absence of structure) of a binary sequence $x \in \{0,1\}^\ast$. High algorithmic complexity denotes the absence of structure, while low algorithmic complexity is an indicator of strong interdependencies between the bits of $x$. Moreover, two sequences that are structurally related display similar algorithmic complexities. More precisely, the following very simple standard lemma expresses that computable functions can only reduce algorithmic complexity. This lemma can be found as \cite[Exercise 2.1.6.a]{li2008introduction}, and will turn out useful for the rest of the paper.
\begin{lemma}\label{lem:computable functions decrease complexity}
	Let $\varphi : \{0,1\}^\ast \to \{0,1\}^\ast$ be a computable function. Then,
	\begin{equation}\label{eq:main:lem:computable functions decrease complexity}
		K[\varphi(x)] \leq K[x] + \bigOx[1]{|x|}, \ \forall x \in \{0,1\}^\ast.
	\end{equation}
\end{lemma}
\begin{proof}
	Recall that $\Uref$ is the universal computable function fixed throughout the paper, and let $U := \Uref(\cdot, 0^\infty)$. Let $\varphi : \{0,1\}^\ast \to \{0,1\}^\ast$ be a computable function. $\varphi \circ U$ is also a computable function (the composition of two computable functions is again a computable function \cite[Exercise 2.1.2]{lovasz2018notes}). Further, as $U$ is a universal computable function, there exists a program $p\in \{0,1\}^\ast$ such that $U(\langle p, x \rangle) = \varphi \circ U(x)$, for all $x \in \{0,1\}^\ast$. Let now $x \in \{0,1\}^\ast$, and $x^\ast \in \{0,1\}^\ast$ be such that $U(x^\ast) = x$ and $|x^\ast| = K[x]$. Then, $U(\langle p, x^\ast \rangle) = \varphi \circ U(x^\ast) = \varphi(x)$. By definition, it immediately follows that $K[\varphi(x)] \leq |\langle p, x^\ast\rangle| = 2|p| + |x^\ast| + 1 = K[x] + 2|p| + 1$.
\end{proof}

In this paper, we manipulate infinite binary sequences. Hopefully, algorithmic complexity naturally extends from finite to infinite sequences. For an infinite sequence $\xf \in \{0,1\}^\N$, we study the behavior of $K[\xf_{1:n}]$ in function of $n \in \N$. We also extend this concept to relative algorithmic complexity. Namely, for two infinite sequences $\xf, \yf \in \{0,1\}^\N$, the algorithmic complexity of $\xf$ relatively to $\yf$ is understood through the study of the behavior of $K[\xf_{1:n}|\yf]$ in function of $n \in \N$.

The difference of the algorithmic complexity between prefixes of the greedy and lazy binary expansions is well understood. For $s \in [0,1]$, the greedy binary expansion $\xf$ of $s$ and the lazy binary expansion $\xf'$ of $s$ satisfy
\begin{equation}\label{eq:comparison between the kolmogorov complexity of the greedy and lazy expansions}
	|K[\xf_{1:n}] - K[\xf'_{1:n}]| \leq c, \ \text{for all} \ n \in \N,
\end{equation}
where $c > 0$ is independant of $n$ and $s$. This follows directly from the fact that one can generate the lazy binary expansion from the greedy binary expansion by following the procedure specified through (\ref{eq:greedy expansion ending with infinitetail of zeros})-(\ref{eq:lazy expansion ending with infinitetail of ones}). In particular, the operations (\ref{eq:greedy expansion ending with infinitetail of zeros})-(\ref{eq:lazy expansion ending with infinitetail of ones}) can be carried out by an effective algorithm. 

To the best of our knowledge, there is no literature at all relating the algorithmic complexity of the $\beta$-expansions to the algorithmic complexity of the binary expansions, even though their rich structure suggests interesting relationships. In particular, as discussed in the previous section, for a given $s\in [0,1]$, its $\beta$-expansions could have a different algorithmic complexity than its binary expansions. We first make important remarks that will shape our study.
	\begin{enumerate}[label=(\alph*)]
		\item As mentionned before, the set of real numbers that can be represented by $\beta$-expansions is larger that the set of real numbers that can be represented by a binary expansion: while every real number in $[0,1]$ can be represented by a binary expansion, every number in $I_\beta := [0, (\beta-1)^{-1}]$ can be represented by a $\beta$-expansion. In the sequel, we will consider only the real numbers in $[0,1]$. Therefore, the sequences that can be used as $\beta$-expansions are only those sequences $\xf \in \{0,1\}^\N$ that satisfy
		\begin{equation}
			\sum_{i=1}^\infty \xf_i \beta^{-i} \leq 1.
		\end{equation}
		We denote by $\Omega_\beta$ the set of such sequences.
		\item $\beta$-expansions do not approximate real numbers with the same rate as binary expansions. This has a consequence on how to interprete their relative algorithmic complexities. Indeed, for typical $s \in [0,1]$, the first $n$ bits of the binary expansion approximate $s$ with a rate of $\sim 2^{-n}$, while the first $n$ bits of $\beta$-expansions approximate $s$ with a rate of $\propto \beta^{-n} = 2^{-n\log_\beta(2)}$. We need to compare the algorithmic complexities of binary expansions and $\beta$-expansions which approximate $s$ with the same order of magnitude. Therefore, if $\xf$ is a $\beta$-expansion of $s$ and $\yf$ is a binary expansion of $s$ we need to compare $K[\yf_{1:n}]$ to $K[\xf_{1:\nlogbeta}]$, where $n(\beta) := \lceil n \log_\beta(2) \rceil$, for $n \in \N$ and $\beta \in (1,2]$.
		\item By (\ref{eq:comparison between the kolmogorov complexity of the greedy and lazy expansions}), the greedy and lazy binary expansions have similar algorithmic complexity. Therefore, without loss of generality, we will focus on comparing the algorithmic complexity of the $\beta$-expansions of $s$ to the algorithmic complexity of the greedy binary expansion of $s$.
		\item The base $\beta \in (1,2)$ itself might contain a lot of information, which might be reflected in terms of algorithmic complexity. To avoid artifacts due to the information content of $\beta$, we evaluate the algorithmic complexity relative to $\beta$, that is defined based on the algorithmic complexity relative to infinite sequences. For $\beta \in (1,2)$, let $\yf_\beta$ be the greedy binary expansion of $\beta -1$. Then, we define 
		\begin{equation}\label{eq:algorithmic complexity relatively to beta}
			K[x|\beta] := K[x | \yf_\beta], \ \text{for all} \ x \in \{0,1\}^\ast.
		\end{equation}
	\end{enumerate}
	In what follows, we will hence consider $\xf$ a $\beta$-expansion of $s$, and the greedy binary expansion $\yf$ of $s$, and compare $K[\yf_{1:n}|\beta]$ to $K[\xf_{1:\nlogbeta}|\beta]$. This comparison is made formal through the evaluation of the following quantity.
	\begin{definition}
		For $\xf \in \Omega_\beta$, we define
		\begin{equation}
			\underline\Delta_\beta(\xf) = \liminf_{n \to \infty}\frac{K[\xf_{1:\nlogbeta}|\beta] - K[\yf_{1:n}|\beta]}{n},
		\end{equation}
		and
		\begin{equation}
			\bar\Delta_\beta(\xf) = \limsup_{n \to \infty}\frac{K[\xf_{1:\nlogbeta}|\beta] - K[\yf_{1:n}|\beta]}{n},
		\end{equation}
		where $\yf$ is the greedy binary expansion of $s = \sum_{i=1}^\infty \xf_i \beta^{-i}$.
	\end{definition}
	$\bar\Delta_\beta(\xf)$ (resp. $\underline\Delta_\beta(\xf)$) evaluates how more complex is the $\beta$-expansion $\xf$ of $s$ as compared to the greedy binary expansion of $s$, in a worst-case (resp. best-case) fashion. Our work will be first dedicated to derive lower bounds on $\underline\Delta_\beta(\xf)$ and upper bounds on $\bar\Delta_\beta(\xf)$. Then, we will focus on certain specific values of $\beta$, for which we will derive a fine-grained distribution of algorithmic complexity. Finally, we will show that our findings suggest that the robustness displayed by A/D conversion algorithms based on $\beta$-expansions is at the cost of generating more complex sequences, but we will give a concrete solution on how to fix this problem.
	
	\subsection{Main results}
	Our first fundamental result shows that $\beta$-expansions are at least as complex as the greedy binary expansion.
	\begin{theorem}\label{thm:binary expansions are at least as compressible as as beta expansions}
		Let $\beta \in (1,2)$. Then,
		\begin{equation}
			0 \leq \underline\Delta_\beta(\xf) \leq \bar\Delta_\beta(\xf),
		\end{equation}
		for all $\xf \in \Omega_\beta$.
	\end{theorem}
	Our second result is a trivial upper bound, which is mere consequence of the definition of algorithmic complexity: the identity function $\operatorname{id} : \{0,1\}^\ast \to \{0,1\}^\ast$ is computable, so by Definition of $\Uref$, there exists a word $p \in \{0,1\}^\ast$ so that $\Uref(\langle x, p \rangle, \yf_\beta) = \operatorname{id}(x) = x$ for every $x \in \{0,1\}^\ast$ (recall that $\yf_\beta$ is the greedy binary expansion of $\beta - 1$). Further, the Definition of relative algorithmic complexity implies that $K[x|\beta] \leq |\langle x,p\rangle| = |x| + 2|p|$ for all $x \in \{0,1\}^\ast$. As a consequence, we get that
	\begin{equation}
		\underline\Delta_\beta(\xf) \leq \bar\Delta_\beta(\xf) \leq  \limsup_{n \to \infty} \frac{K[\xf_{1:\nlogbeta}|\beta]}{n} \leq \limsup_{n \to \infty} \frac{|\xf_{1:\nlogbeta}| + 2|p|}{n} = \log_\beta(2).
	\end{equation}
	This results in the following Lemma.
	\begin{lemma}\label{lem:trivial upper bound on the relative compressibility of beta expansions}
		Let $\beta \in (1,2)$. Then,
		\begin{equation}
			0 \leq \underline\Delta_\beta(\xf) \leq \bar\Delta_\beta(\xf) \leq \log_\beta(2),
		\end{equation}
		for all $\xf \in \Omega_\beta$.
	\end{lemma}
	We will identify classes of bases $\beta \in (1,2)$ for which the upper bound can be improved. We will cover two results of different nature: a result that holds for almost all $\beta \in (1,2)$, and a result for $\beta$ satisfying certain algebraic properties. The latter result will yield interesting corollaries. We now state the result on the improvement of the bound for almost all $\beta \in (1,2)$.
	\begin{theorem} \label{thm:upper bound on the relative compressibility of beta expansions for almost all beta in 1 sqrt 2}
		For almost all $\beta \in (1,2)$,
		\begin{equation}
			0 \leq \underline\Delta_\beta(\xf) \leq \bar\Delta_\beta(\xf) \leq \log_\beta\left(\frac{2}{\beta}\right),
		\end{equation}
		for all $\xf \in \Omega_\beta$.
	\end{theorem}
	We now move to the result that holds for $\beta \in (1,2)$ satisfying some specific algebraic relationships. We first recall some facts about algebraic numbers. A real number $\beta$ is called algebraic if there exists a polynomial $P = a_d X^d + \ldots + a_1X + a_0 \in \Z[X]$ such that $P(\beta)=0$. Note that there may be several such polynomials, we denote by $P_\beta$ such a polynomial uniquely defined by
	\begin{enumerate}[label = (\alph*)]
		\item $P_\beta$ has minimal degree $d$, i.e. for every polynomial $Q$ of degree less than $P$, $Q(\beta) \neq 0$.
		\item the leading coefficient $a_d$ of $P_\beta$ is positive.
		\item  the leading coefficient $a_d$ of $P_\beta$ is minimal, i.e. for every other polynomial $Q$ of degree $d$ and positive leading coefficient $b_d > 0$ such that $Q(\beta) = 0$, $b_d > a_d$.
	\end{enumerate}
	$P_\beta$ is called the \newname{minimal polynomial} of $\beta$. We denote by $L_\beta$ the leading coefficient $a_d$ of $P_\beta$ and $T_\beta$ the tail coefficient $a_0$ of $P_\beta$. The roots of $P_\beta$, not including $\beta$, are called the \newname{Galois conjugates of $\beta$} (by root we mean any complex number $z \in \C$ such that $P_\beta(z) = 0$). We denote by $G_\beta$ the set of Galois conjugates of $\beta$. Following \cite[Lemma 1.51]{garsia1962ArithmeticPropertiesBernoulli}, we the set $G^+_\beta = \{z \in G_\beta : |z| > 1 \}$ of Galois conjugates that lay outside of the unit circle. $G^+_\beta$ will have a drastic impact of the algorithmic complexity of $\beta$-expansions.
	The question of algorithmic complexity of $\beta$-expansions with $\beta$ algebraic is slightly more complicated to handle, because of the following fact. When $\beta$ is algebraic, there might be multiplicities of finite $\beta$-expansions, i.e., for a given $n \in \N$, there might be different $x,y \in \{0,1\}^n$ such that $\sum_{i=1}^n x_i \beta^{-i} = \sum_{i=1}^n y_i \beta^{-i}$. The existence of at least two such $x,y \in \{0,1\}^n$ only happens when $\beta$ is algebraic. Indeed,
	\begin{align}
		&\exists x,y \in \{0,1\}^n \ s.t. \ \sum_{i=1}^n x_i \beta^{-i} = \sum_{i=1}^n y_i \beta^{-i}\\
		\iff & \exists x,y \in \{0,1\}^n \ s.t. \ \beta^{-n}\sum_{i=1}^n (x_i - y_i) \beta^{n-i} = 0\\
		\iff & \exists x,y \in \{0,1\}^n \ s.t. \ \sum_{i=0}^{n-1} (x_{n-i-1} - y_{n-i-1}) \beta^{i} = 0\\
		\iff & \exists P = \sum_{i=0}^{n-1} (x_{n-i-1} - y_{n-i-1}) X^{i} \in \Z[X] \ s.t. \ P(\beta) = 0.
	\end{align}
	These multiplicities naturally define an equivalence relationship on $\{0,1\}^\ast$, by $x \sim_\beta y \iff \len x = \len y =: n \ \text{and} \ \sum_{i=1}^n x_i \beta^{-i} = \sum_{i=1}^n y_i \beta^{-i}$. We denote by $[x]_\beta$ the equivalence class associated to $\sim_\beta$. To every $x \in \{0,1\}^\ast$, we define $M_\beta(x)$ to be the lexicographically maximal element of $[x]_\beta$. Moreover, given an infinite sequence $\xf$, we define $M_\beta \xf$ as being the sequence such that $(M_\beta \xf)_{1:n} := M_\beta(x_{1:n})$. The introduction of $M_\beta\xf$ allows to derive the main result on the algorithmic complexity of the $\beta$-expansions when $\beta$ is algebraic.
	\begin{theorem}\label{thm:upper bound on the relative compressibility of beta expansions for algebraic beta}
		Let $\beta \in (1,2)$ be an algebraic number. Then,
		\begin{equation}
			0 \leq \underline\Delta_\beta(M_\beta\xf) \leq \bar\Delta_\beta(M_\beta\xf) \leq \log_\beta\left(L_\beta \prod_{z \in G^\ast_\beta} \left|z \right|\right),
		\end{equation}
		for all $\xf \in \Omega_\beta$.
	\end{theorem}
	The above theorem allows to derive several corollaries. First, we will establish a class of algebraic numbers for which there are no multiplicities, and hence for which $M_\beta\xf$ can be replaced by $\xf$ in the above result. This class is a generalization of the class of \newname{Garsia numbers} introduced in \cite[Section 1.7]{garsia1962ArithmeticPropertiesBernoulli}. Second, we study another family of algebraic numbers, called \newname{Pisot numbers} introduced in \cite{pisotRepartitionModulo1}, for which the above equation yields $\underline\Delta_\beta(M_\beta\xf) = \bar\Delta_\beta(M_\beta\xf) = 0$. In the next section, we exploit this property to modify slightly the A/D converter of Daubechies in order to control the algorithmic complexity of the $\beta$-expansion it delivers.

	We proceed to introduce a class of algebraic numbers, inspired and generalized from the Garsia numbers. First, we restrict the class of algebraic numbers as follows. A real number $\beta$ is said to be an \newname{algebraic integer} if $\beta$ is an algebraic number, and the leading coefficient $L_\beta$ of its minimal polynomial $P_\beta$ satisfies $L_\beta = 1$. Further, we define a subset $\mathscr C$ of algebraic integers as follows. An algebraic integer $\beta$ belongs to $\mathscr C$ if and only if the tail coefficient $T_\beta$ of its minimal polynomial $P_\beta$ satisfies $T_\beta \geq 2$. The class $\mathscr C$ yields a very important property: for $\beta \in \mathscr C$, there are no multiplicities, i.e., $M_\beta \xf = \xf$ for all $\xf \in \{0,1\}^\N$. This results in the following corollary.
	\begin{corollary}\label{thm:upper bound on the relative compressibility of beta expansions for garsia beta}
		Let $\beta \in (1,2)$ be an algebraic integer in $\mathscr C$. Then,
		\begin{equation}
			0 \leq \underline\Delta_\beta(\xf) \leq \bar\Delta_\beta(\xf) \leq \sum_{z \in G^+_\beta} \log_\beta\left( \left|z \right|\right),
		\end{equation}
		for all $\xf \in \Omega_\beta$.
	\end{corollary}
	A simple example of such a number is $\sqrt{2}$. Indeed, the minimal polynomial of $\sqrt{2}$ is $X^2 - 2$, hence $\sqrt{2}$ belongs to $\mathscr C$. Moreover, the only Galois conjugate of $\sqrt{2}$ is $-\sqrt{2}$, which yields
	\begin{equation}
		0 \leq \underline\Delta_{\sqrt{2}}(\xf) \leq \bar\Delta_{\sqrt{2}}(\xf) \leq  \log_{\sqrt{2}}\left( \sqrt{2}\right) = 1,
	\end{equation}
	for all $\xf \in \Omega_\beta$. More generally, the $m$-th root of any integer $k \geq 2$ belongs to $\mathscr C$, since its minimal polynomial is $P_{\sqrt[m]{k}} = X^m - k$. The Galois conjugates of $\sqrt[m]{k}$ are $\sqrt[m]{k} \exp\left(i \frac{2\pi \ell}{m}\right)$, $\ell \in 1, \ldots, m - 1$. This yields
	\begin{align}
		0 \leq \underline\Delta_{\sqrt[m]{k}}(\xf) \leq \bar\Delta_{\sqrt[m]{k}}(\xf) &\leq \sum_{\ell \in \{1,\ldots,m-1\}} \log_{\sqrt[m]{k}}\left( \left|\sqrt[m]{k} \exp\left(i \frac{2\pi \ell}{m}\right)\right|\right)\\
		& \leq \sum_{\ell \in \{1,\ldots,m-1\}} \log_{\sqrt[m]{k}}\left( \sqrt[m]{k}\right) = m-1,
	\end{align}
	for all $\xf \in \Omega_\beta$. On can appreciate the tradeoff between redundancy and complexity: larger $m$ is, the closer $\sqrt[m]{k}$ is to $1$, i.e., the more redundant the $\sqrt[m]{k}$-expansions are, the higher the upper bound $m-1$ on their relative compressibility is.

	We close this section with a note on Pisot numbers, that are defined as algebraic integers having all their Galois conjugates located strictly inside the unit disk of $\C$. Pisot numbers include notably the golden ratio $G = \frac{1 +\sqrt{5}}{2}$, since its only Galois conjugate is the other root of the polynomial $X^2 - X - 1$, which is $\frac{1 - \sqrt{5}}{2} \simeq -0.61$, that is strictly inside of the unit disk. Note that if $\beta$ is a Pisot number, $G^+_\beta = \emptyset$. This yields the following immediate corollary.
	\begin{corollary}\label{cor:upper bound on the relative compressibility of beta expansions for pisot beta}
		Let $\beta \in (1,2)$ be a Pisot number. Then,
		\begin{equation}
			\underline\Delta_\beta(M_\beta\xf) = \bar\Delta_\beta(M_\beta\xf) = 0,
		\end{equation}
		for all $\xf \in \Omega_\beta$.
	\end{corollary}
	Interestingly, this further has a consequence on the algorithmic complexity of the greedy $\beta$-expansion.
	\begin{corollary}\label{cor:upper bound on the relative compressibility of greedy beta expansions for pisot beta}
		Let $\beta \in (1,2)$ be a Pisot number, $s\in [0,1]$ and $\xf$ be the greedy $\beta$-expansion of $s$. Then,
		\begin{equation}
			\underline\Delta_\beta(\xf) = \bar\Delta_\beta(\xf) = 0.
		\end{equation}
	\end{corollary}
	In the next section, we exploit these considerations on Pisot numbers to control the complexity of the sequences delivered by the A/D conversion algorithm.
	

	\subsection{Application: a denoising algorithm for A/D conversion}

	The algorithm of A/D conversion presented before uses the redundant property of $\beta$-expansions to generate arbitrarily precise representations of the input, even in the presence of imperfect quantizers. However, as raised previously in this paper, the generated $\beta$-expansions could be more complex than the binary expansion of the same input, following a redundancy versus complexity tradeoff. The successive theorems presented above indeed suggest that the $\beta$-expansions generated could indeed be more complex, hence harder to store on a computer. However, we also have seen that when $\beta$ is a Pisot number, we can identify a class of $\beta$-expansions that display the same complexity as the binary expansion. It turns out that Pisot $\beta$ is a usual choice in A/D conversion, notably the golden ratio $\beta = \frac{1+\sqrt{5}}{2}$ \cite{daubechies2010GoldenRatioEncoder}. Here, we first expose a statement that strengthens the idea that the $\beta$-expansions generated by the A/D conversion algorithm can be more complex than the binary expansions, and then we fix this problem by establishing an algorithm that converts any $\beta$-expansion into another $\beta$-expansion of minimal complexity, in linear time.

	In the previous section, we have established that for all $\beta \in (1,2)$, 
	\begin{equation}
		\underline{\Delta}_\beta(\xf) = 0, \ \forall \xf \in \Omega_\beta.
	\end{equation}
	The proof of this statement actually yields an even stronger result, namely 
	\begin{equation}
		K[\yf_{1:n}|\beta] \leq K[\xf_{1:\nlogbeta}|\beta] + \bigOx[1]{n},
	\end{equation}
	for all $n \in \N$, and all $\xf \in \Omega_\beta$, where $\yf$ is the greedy binary expansion of $s := \sum_{i=1}^\infty \xf_i \beta^{-i}$. The question is now to estimate the distribution of $K[\xf_{1:\nlogbeta}|\beta] - K[\yf_{1:n}|\beta]$, in order to understand if for typical $\xf \in \Omega_\beta$, $K[\xf_{1:\nlogbeta}|\beta]$ tends to be close or far from $K[\yf_{1:n}|\beta]$. This results in the study of the following set
	\begin{align}
		\KK_\beta[s,n,k] := & \left\{ \xf_{1:\nlogbeta}: \xf \ \text{is a} \ \beta\text{-expansion of} \ s,\right. \nonumber\\
		& \hspace{0.86in}\left. K[\xf_{1:\nlogbeta}|\beta] \leq K[\yf_{1:n}|\beta] + k\right\},
	\end{align}
	for $s \in [0,1]$, $n,k \in \N$. For $s \in [0,1]$, $n,k \in \N$, the set $\KK_\beta[s,n,k]$ contains the prefixes of length $\nlogbeta$ of those $\beta$-expansions of $s$ that display an algorithmic complexity that is higher than the algorithmic complexity of the corresponding greedy binary expansion by at most $k$ bits. We establish the following result.
	\begin{theorem}\label{thm:distribution of algorithmic complexity of beta expansions when beta is pisot}
		Let $\beta \in (1,2)$ be a Pisot number, $s \in [0,1]$ and $n,k \in \N$. Then, there exists $M \in \N$ such that
		\begin{equation}
			\# \KK_\beta[s,n,k] \leq 2^{k+1} n^M.
		\end{equation}
	\end{theorem}
	It was established in \cite[Theorem 1]{feng2010GrowthRatebetaexpansions} that if $\beta$ is a Pisot number, then for almost all $s \in [0,1]$, there exists a constant $\gamma > 0$ such that the set 
	\begin{equation}
		\Sigma_\beta[s,n] = \left\{ \xf_{1:\nlogbeta}: \xf \ \text{is a} \ \beta\text{-expansion of}\ s\right\}
	\end{equation}
	satisfies 
	\begin{equation}
		\#\Sigma_\beta[s,n] \geq 2^{\gamma n},
	\end{equation}
	for $n$ large enough. This means that for $k << \gamma n$, the set $\KK_\beta[s,n,k]$ contains a very small fraction of all the $\beta$-expansions of $s$. Concretely,
	\begin{equation}
		\frac{\# \KK_\beta[s,n,\alpha n]}{\#\Sigma_\beta[s,n]} \leq 2^{k+1 - \gamma n}n^M,
	\end{equation}
	for almost all $s \in [0,1]$, and $n,k \in \N$ with $n$ large enough. Then, for all $\alpha < \gamma$,
	\begin{equation}
		\lim_{n \to \infty} \frac{\# \KK_\beta[s,n,k]}{\#\Sigma_\beta[s,n]} = 0,
	\end{equation}
	for almost all $s \in [0,1]$. This means that asymptotically, the majority of the $\beta$-expansions of a typical $s \in [0,1]$ are more complex than the binary expansions. This yields a tradeoff between the complexity and the precision of the A/D conversion algorithm.

	However, we now show that we can break this tradeoff. By Corollary \ref{cor:upper bound on the relative compressibility of beta expansions for pisot beta}, we have shown that if $\beta \in (1,2)$ is a Pisot number, we have
	\begin{equation}
		\underline\Delta_\beta(M_\beta\xf) = \bar\Delta_\beta(M_\beta\xf) = 0,
	\end{equation}
	for all $\xf \in \Omega_\beta$. We then break the tradeoff by designing an effective algorithm to compute $M_\beta$ in linear time, i.e., there exists a constant $C >0$ such that the algorithm delivers $M_\beta(x)$ when presented with any input $x \in \{0,1\}^n$ in less than $Cn$ computation steps. The property of the algorithm computing in linear time crucially relies on $\beta$ being a Pisot number.

	\begin{figure}[H]
		\centering
		\begin{tikzpicture}[
			node distance=2cm and 3cm,
			box/.style = {draw, minimum width=3cm, minimum height=3cm},
			>=Stealth
			]
		
				\node[box] (A) at (0,0) {A/D converter};
				\node[box, right=of A, align=center] (B) {Algorithm to compute\\ $M_\beta$ in linear time};
		
				\draw[->] ([xshift=-1cm]A.west)node[left] {$s$} -- (A.west);
				\draw[->] (A.east) -- (B.west) node[midway, above] {$\xf$};
				\draw[->] (B.east) -- ([xshift=1cm]B.east) node[right] {$M_\beta \xf$};
		
			\end{tikzpicture}
			\caption{Pipeline of A/D conversion with controlled algorithmic complexity. Here $s \in [0,1]$ and $\xf$ is a $\beta$-expansion of $s$. $M_\beta \xf$ displays the same algorithmic complexity as the greedy binary expansion of $s$.}
	\end{figure}
	
	\subsection{Organisation of the paper}

	Section \ref{sec:multivalued functions and algorithmic complexity} is dedicated to introduce the concept of computable multivalued functions that will prove to be pivotal in most of the proofs. Section \ref{sec:lower bound on the complexity of beta expansions} is dedicated to introducing the philosophy of all the proofs of the paper, notably by showing how to exploit the concepts introduced in Section \ref{sec:multivalued functions and algorithmic complexity} to prove Theorem \ref{thm:binary expansions are at least as compressible as as beta expansions}. Section \ref{sec:upper bound on the complexity of beta expansions} is dedicated to establishing Theorems \ref{thm:upper bound on the relative compressibility of beta expansions for almost all beta in 1 sqrt 2} and \ref{thm:upper bound on the relative compressibility of beta expansions for algebraic beta}. In Section \ref{sec:a fast algorithm to control complexity} we construct the fast algorithm to compute $M_\beta$, that leads to breaking the tradeoff between complexity and robustness for A/D conversion. Section \ref{sec:distribution of the complexities of beta expansions} is the most technical, and proceeds to establish Theorem \ref{thm:distribution of algorithmic complexity of beta expansions when beta is pisot}. 

\endgroup

\begingroup
\let\clearpage\relax

\section{Multivalued functions and algorithmic complexity}

\label{sec:multivalued functions and algorithmic complexity}

In this section, we introduce multivalued functions, and we expose the pivotal result of this paper, that is seen to be a generalization of \cite[Theorem 2]{staiger2002kolmogorov}. We start by giving the definition of a multivalued function, and of a computable multivalued function.

\begin{definition}[Multivalued function]\label{def:multivalued_function}
    Let $X,Y$ be two sets. We write $f: X \rightrightarrows Y$ to denote the fact that $f$ is a multivalued function from $X$ to $Y$, i.e., for each $x \in X$, $f(x)$ is a subset of $Y$.
\end{definition}

\begin{definition}[Computable multivalued function]\label{def:computable_multivalued_function}
    Let $f: \{0,1\}^\ast \rightrightarrows \{0,1\}^\ast$ be a multivalued function. We say that $f$ is computable if there exists a computable function $g : \{0,1\}^\ast \to \{0,1\}^\ast$ such that 
    \begin{equation}
        g(x) = \langle  f(x) \rangle. 
    \end{equation}
\end{definition}

We now state the result on the algorithmic complexity of multivalued functions, which is a reformulation of classical results of nonprobabilistic statistics \cite[Section 5.5]{li2008introduction}.

\begin{theorem}[Computable multivalued functions and complexity]\label{thm:computable multivalued functions and complexity}
    Let $f: \{0,1\}^\ast \rightrightarrows \{0,1\}^\ast$ be a computable multivalued function and $x \in \{0,1\}^\ast$. For all $y \in f(x)$,
    \begin{equation}\label{eq:main:thm:computable multivalued functions and complexity}
        K[y] \leq K[x] + \log_2 \#f(x) + 2\log \log \#f(x) + \bigOx[1]{\len x}.
    \end{equation}
\end{theorem}
\begin{proof}
    Let $f: \{0,1\}^\ast \rightrightarrows \{0,1\}^\ast$ be a computable multivalued function. Then, by Definition \ref{def:computable_multivalued_function}, there exists $g : \{0,1\}^\ast \to \{0,1\}^\ast$ computable, such that 
    \begin{equation}
        g(x) = \langle  f(x) \rangle.
    \end{equation}
    For $x \in \{0,1\}^\ast$, we let $y_1(x), y_2(x), \ldots, y_{\#f(x)}$ be the list of all the elements of $f(x)$, such that 
    \begin{equation}
        g(x) = \langle y_1(x), \ldots, y_{\#f(x)}(x) \rangle.
    \end{equation}
    We can then define a computable function $h: \{0,1\}^\ast \to \{0,1\}^\ast$ such that
    \begin{equation}
        h(\langle x, \bin(i) \rangle) = y_i(x),
    \end{equation}
    through the following effective algorithm. On input $\langle x, \bin(i) \rangle$, use the function $g$ to generate the sequence 
    \begin{equation}
        y = \langle y_1(x), \ldots, y_{\#f(x)}(x) \rangle.
    \end{equation} 
    Then, use the sequence $y$ to enumerate the elements of $f(x)$ until reaching the $i$-th element, and output this element, which is $y_i(x)$. Hence, $h$ is computable and one has
    \begin{align}
        K[y_i(x)] = K[h(\langle x, \bin(i) \rangle)] &\osetlem{\ref{lem:computable functions decrease complexity}}{\leq} K[\langle x, \bin(i) \rangle] + \bigOx[1]{\len x}\\
        &\overset{(a)}{\leq} K[x] + K[i] + 2\log K[\bin(i)] + \bigOx[1]{\len x},\\
        &\overset{(b)}{\leq} K[x] + \len \bin(i) + 2 \log \len \bin(i) + \bigOx[1]{\len x}\\
        &\leq K[x] + \log i + 2 \log \log i + \bigOx[1]{\len x},\\
        &\leq K[x] + \log \#f(x) + 2 \log \log \#f(x) + \bigOx[1]{\len x},
    \end{align}
    for all $x \in \{0,1\}^\ast$ and $1 \leq i \leq \#f(x)$, where (a) and (b) follow from \cite[Example 2.1.5]{li2008introduction} and \cite[Theorem 2.1.2]{li2008introduction}, respectively.
\end{proof}

As mentionned in previous section, we will consider algorithmic complexity relative to $\beta \in (1,2)$, defined in (\ref{eq:algorithmic complexity relatively to beta}). Accordingly, we define functions that are computable relatively to $\beta \in (1,2)$ as being the functions that are computable relatively to $\yf_\beta$, where $\yf_\beta \in \{0,1\}^\N$ is the greedy binary expansion of $\beta - 1$. We further study the notion multivalued functions that are computable relatively to $\beta \in (1,2)$.
\begin{definition}[Relatively computable multivalued function]\label{def:relatively computable multivalued function}
    Let $f: \{0,1\}^\ast \rightrightarrows \{0,1\}^\ast$ be a multivalued function. We say that $f$ is computable relatively to $\yf\in \{0,1\}^\N$ (resp. $\beta \in (1,2)$) if there exists a function $g : \{0,1\}^\ast \to \{0,1\}^\ast$ computable relatively to $\yf$ (resp. $\beta$) such that 
    \begin{equation}
        g(x) = \langle  f(x) \rangle. 
    \end{equation}
\end{definition}
We can finally derive a relation for relative algorithmic complexity of multivalued functions.
\begin{theorem}[Relatively computable multivalued functions and complexity]\label{thm:relatively computable multivalued functions and complexity}
    Let $f: \{0,1\}^\ast \rightrightarrows \{0,1\}^\ast$ be a multivalued function computable relatively to $\yf \in \{0,1\}^\N$ and $x \in \{0,1\}^\ast$. Then, 
    \begin{equation}\label{eq:main:thm:relatively functions computable multivalued  and complexity}
        K[y|\yf] \leq K[x|\yf] + \log_2 \#f(x) + 2\log \log \#f(x) + \bigOx[1]{\len x},
    \end{equation}
    for all $x \in \{0,1\}^\ast$ and $y \in f(x)$. Moreover, (\ref{eq:main:thm:relatively functions computable multivalued  and complexity}) also holds if we replace $\yf \in \{0,1\}^\N$ by $\beta\in(1,2)$.
\end{theorem}
We omit the proof, as it is essentially the same as the proof of Theorem \ref{thm:computable multivalued functions and complexity}, except that we need to relativize every equation with respect to $\yf \in \{0,1\}^\N$.

\endgroup

\begingroup
\let\clearpage\relax

\section{\texorpdfstring{Lower bound on the complexity of $\beta$-expansions}{}}

\label{sec:lower bound on the complexity of beta expansions}

This section is devoted to the proof of Theorem 1.2. Let $\beta \in (1,2]$. For $s \in [0,1]$, we denote by $\Sigma_\beta(s)$ the set of all $\beta$-expansions of $s$, i.e.,
\begin{equation}
    \Sigma_\beta(s) := \left\{ \xf \in \{0,1\}^\N : \sum_{i=1}^\infty \xf_i \beta^{-i} = s \right\}.
\end{equation}
Recall that, as explained in Sections \ref{sec:binary expansion} and \ref{sec:beta expansions}, $\# \Sigma_2(s) \leq 2$ and $\# \Sigma_\beta(s) = 2^{\aleph_0}$ if $\beta < \frac{1+\sqrt{5}}{2}$, for all $s \in [0,1]$. Further, we define $\Sigma_\beta(s,n)$ to be the set of sequences of length $n$ that are prefix of some $\beta$-expansion of $s$, i.e.,
\begin{equation}
    \Sigma_\beta(s,n) := \{ \xf_{1:n} : \xf \in \Sigma_\beta(s)\}.
\end{equation}

In order to prove Theorem \ref{thm:binary expansions are at least as compressible as as beta expansions}, we will construct a multivalued function $f_{\beta \to 2} : \{0,1\}^\ast  \rightrightarrows \{0,1\}^\ast$ that is computable relatively to $\beta$, for $\beta \in (1,2]$. This function $f_{\beta \to 2}$ will be designed so that when it is presented with an $\nlogbeta$-prefix $x$ of some $\beta$-expansion of a given real number $s \in [0,1]$, it outputs a set that contains all the $n$-prefixes of all the binary expansions of $s$. In mathematical symbols, this is expressed as
\begin{equation}
    \Sigma_2(s,n) \subseteq f_{\beta \to 2}(x), \ \forall x \in \Sigma_\beta(s,\nlogbeta), \ n \in \N.
\end{equation}
The function $f_{\beta \to 2}$ will be constructed so that $x \mapsto \# f_{\beta \to 2}(x)$ is bounded, in order to use Theorem \ref{thm:relatively computable multivalued functions and complexity}. 

\subsection{\texorpdfstring{Construction of a multivalued function to convert between bases $\beta$ and $2$}{}}
We now explain the heuristic that is at the origin of the contruction of the function $f_{\beta \to 2}$. This heuristic is inspired from the proof of \cite[Theorem 3]{staiger2002kolmogorov}. Let $\beta \in (1,2]$, $s \in [0,1]$ and $x \in \Sigma_\beta(s,\nlogbeta)$. Since $x \in \Sigma_\beta(s,\nlogbeta)$, then there exists $\xf \in \{0,1\}^\N$ such that $x\xf$ is a $\beta$-expansion of $s$, i.e.,
\begin{equation}\label{eq:1:construction of a multivalued function computable wrt beta that maps beta expansions to binary expansions of a real number}
    s = \sum_{i=1}^{\nlogbeta} x_i \beta^{-i} + \sum_{i=\nlogbeta+1}^{\infty} \xf_i \beta^{-i}.
\end{equation}
Observe that 
\begin{equation}\label{eq:2:construction of a multivalued function computable wrt beta that maps beta expansions to binary expansions of a real number}
    0 \leq \sum_{i=\nlogbeta+1}^{\infty} \xf_i \beta^{-i} \leq \sum_{i=\nlogbeta+1}^{\infty} \beta^{-i} \leq \frac{\beta^{-\nlogbeta}}{\beta-1} \leq \frac{2^{-n}}{\beta-1}\cdot
\end{equation}
Combining (\ref{eq:1:construction of a multivalued function computable wrt beta that maps beta expansions to binary expansions of a real number}) and (\ref{eq:2:construction of a multivalued function computable wrt beta that maps beta expansions to binary expansions of a real number}), we obtain
\begin{equation}\label{eq:3:construction of a multivalued function computable wrt beta that maps beta expansions to binary expansions of a real number}
    \sum_{i=1}^{\nlogbeta} x_i \beta^{-i} \leq s \leq \sum_{i=1}^{\nlogbeta} x_i \beta^{-i} + \frac{2^{-n}}{\beta-1}\cdot
\end{equation}
Therefore, from the knowledge of $x$ only, we can deduce that $s$ belongs to the interval $I(x)$ defined as 
\begin{equation}\label{eq:4:construction of a multivalued function computable wrt beta that maps beta expansions to binary expansions of a real number}
    I(x) :=  \left[\sum_{i=1}^{\nlogbeta} x_i \beta^{-i},  \sum_{i=1}^{\nlogbeta} x_i \beta^{-i}+ \frac{2^{-n}}{\beta-1} \right].
\end{equation}
Now, by combining again (\ref{eq:1:construction of a multivalued function computable wrt beta that maps beta expansions to binary expansions of a real number}) and (\ref{eq:2:construction of a multivalued function computable wrt beta that maps beta expansions to binary expansions of a real number}) for the specific case of $\beta = 2$, we get that
\begin{equation}\label{eq:5:construction of a multivalued function computable wrt beta that maps beta expansions to binary expansions of a real number}
    t - 2^{-n} \leq \sum_{i=1}^{n} y_i 2^{-i} \leq t,
\end{equation}
for all $t \in [0,1]$ and $y \in \Sigma_2(t,n)$. This means that if we know exactly the value of $t$, we can deduce that all the $n$-prefixes of the binary expansions of $t$ are exactly those sequences $y \in \{0,1\}^n$ that satisfy 
\begin{equation}\label{eq:6:construction of a multivalued function computable wrt beta that maps beta expansions to binary expansions of a real number}
    \sum_{i=1}^{n} y_i 2^{-i} \in [t-2^{-n}, t].
\end{equation}
However, if we do not know exactly $t$, but rather we know that $t \in [a,b] \subseteq [0,1]$, then we can deduce that all the $n$-prefixes $y$ of the binary expansions of $t$ satisfy
\begin{equation}\label{eq:7:construction of a multivalued function computable wrt beta that maps beta expansions to binary expansions of a real number}
    \sum_{i=1}^{n} y_i 2^{-i} \in [a-2^{-n}, b].
\end{equation}
Now, in light of (\ref{eq:7:construction of a multivalued function computable wrt beta that maps beta expansions to binary expansions of a real number}), (\ref{eq:4:construction of a multivalued function computable wrt beta that maps beta expansions to binary expansions of a real number}) can be reinterpreted as follows. From the knowledge of $x$, we can deduce that all the $n$-prefixes $y$ of the binary expansions of $s$ satisfy
\begin{equation}\label{eq:J x:construction of a multivalued function computable wrt beta that maps beta expansions to binary expansions of a real number}
    \sum_{i=1}^{n} y_i 2^{-i} \in \left[\sum_{i=1}^{\nlogbeta} x_i \beta^{-i} - 2^{-n},  \sum_{i=1}^{\nlogbeta} x_i \beta^{-i}+ \frac{2^{-n}}{\beta-1} \right] =: J(x).
\end{equation}
The above heuristic results in the following formal definition of a multivalued function $\hat f_{\beta \to 2}$, that however is still not the multivalued function $f_{\beta \to 2}$ we are looking for, for reasons we explain just next.
\begin{definition}\label{def:noncomputable multivalued function converting from base beta to base 2}
    Let $\hat f_{\beta \to 2} : \{0,1\}^\ast \rightrightarrows \{0,1\}^\ast$ be the multivalued function defined as 
\begin{equation}\label{eq:main:def:noncomputable multivalued function converting from base beta to base 2}
    \hat f_{\beta \to 2}(x) := \left\{ y \in \{0,1\}^n :  \sum_{i=1}^{n} y_i 2^{-i} \in J(x)\right\},
\end{equation}
for all $x \in \{0,1\}^{\nlogbeta}$ for some $n \in \N$, and $\hat f_{\beta \to 2}(x) = \epsilon$ for $x \in \{0,1\}^m$ where $m \neq \nlogbeta$ for all $n \in \N$.
\end{definition} 
However, the function $\hat f_{\beta \to 2}$ is clearly not computable (even relatively to $\beta$) since it makes use of comparisons between arbitrary real numbers to evaluate wether or not $\sum_{i=1}^{n} y_i 2^{-i} \in J(x)$, i.e., its output cannot be calculated from its input through an effective algorithm. The core of the problem is that the operation ``$\in$'' is not computable, as it relies on both the operations ``$\leq$'' and ``$\geq$''. This well-known issue can be fixed by using an approximate version of ``$\leq$''. Namely, following \cite[Section 4.1, p79-80]{brattka2018ComputabilityAnalysis}, there exists a computable function $\varphi : \N \times \{0,1\}^\N \to \{0,1\}$ such that
\begin{equation}
    \varphi(n, \langle \yf, \zf \rangle) = \begin{cases}
        0 \ \text{if} \ s < t\\
        0 \ \text{or} \ 1 \ \text{if} \ t \leq s \leq t + 2^{-n},\\
        1 \ \text{if} \ t > s + 2^{-n},
    \end{cases}
\end{equation}
for all $n \in \N$, $s,t \in [0,1]$, and where $\yf$ and $\zf$ are the respective greedy binary expansions of $s$ and $t$. For $n \in \N$, we define the binary relation $\leq_n$ by 
\begin{equation}
    s \leq_n t \iff \varphi(n , \langle \yf, \zf \rangle), \ \forall s,t \in [0,1],
\end{equation}
where $\yf,\zf$ are the respective greedy binary expansions of $s$ and $t$. Based on this relation, we can construct an ``approximate membership'' relation for $n \in \N$ by 
\begin{equation}
    s \in_n [a,b] \iff a \leq_n s - 2^{-n} \ \text{and} \ a\leq_n b, \ \forall s,a,b \in [0,1].
\end{equation}
Note that, in particular,
\begin{equation}
    s \in_n [a,b] \Rightarrow s \in [a - 2^{-n}, b + 2^{-n}], \ \ \forall s,a,b \in [0,1].
\end{equation}
In consequence, for an interval $I := [a,b] \subseteq [0,1]$, we define 
\begin{equation}\label{eq:definition of superset of size 2 -n more}
    I^{(n)} := [a - 2^{-n}, b + 2^{-n}], \ \text{for} \ n \in \N,
\end{equation}
and we get 
\begin{equation}\label{eq:finite precision membership up to 2 -n implies membership to a superset of size 2 -n more}
    s \in_n I \Rightarrow s \in I^{(n)}, \ \ \forall s \in [0,1].
\end{equation}
Now, the function $f_{\beta \to 2}$ can simply be defined with the usage of this approxiamte membership relation.
\begin{definition}\label{def:multivalued function converting from base beta to base 2}
    Let $f_{\beta \to 2} : \{0,1\}^\ast \rightrightarrows \{0,1\}^\ast$ be the multivalued function defined as 
\begin{equation}\label{eq:main:def:multivalued function converting from base beta to base 2}
    f_{\beta \to 2}(x) := \left\{ y \in \{0,1\}^n :  \sum_{i=1}^{n} y_i 2^{-i} \in_n J(x)\right\},
\end{equation}
for all $x \in \{0,1\}^{\nlogbeta}$ for some $n \in \N$, and $f_{\beta \to 2}(x) = \epsilon$ for $x \in \{0,1\}^m$ where $m \neq \nlogbeta$ for all $n \in \N$.
\end{definition}

\subsection{Consequences on algorithmic complexity}

In order to use Theorem \ref{thm:relatively computable multivalued functions and complexity}, there remains to establish an upper bound on the cardinality of $f_{\beta \to 2}(x)$, for all $x \in \{0,1\}^\ast$.
\begin{lemma}\label{lem:upper bound on the cardinality of the beta to 2 conversion multifunction}
    For $\beta \in (1,2]$ and $x \in \{0,1\}^\ast$,
    \begin{equation}
        \#f_{\beta \to 2}(x) \leq \frac{1}{\beta-1} + 3, \ \forall x \in \{0,1\}^\ast.
    \end{equation}
\end{lemma}
\begin{proof}
    Let $\beta \in (1,2]$ and $x \in \{0,1\}^\ast$. 
    \begin{enumerate}[label=(\alph*)]
        \item Suppose that there is no $n \in \N$ such that $\len x = \nlogbeta$. Then, $f(x) = \epsilon$, so $\#f(x) = 1 \leq \frac{1}{\beta-1} + 3$.
        \item Suppose that there exists $n \in \N$ such that $\len x = \nlogbeta$. By definition,
        \begin{align}
            f_{\beta \to 2}(x) &= \left\{ y \in \{0,1\}^n :  \sum_{i=1}^{n} y_i 2^{-i} \in_n J(x)\right\}\\
            &\subseteq \left\{ y \in \{0,1\}^n :  \sum_{i=1}^{n} y_i 2^{-i} \in J(x)^{(n)}\right\} =: A(x)
        \end{align}
        We give an upper bound on the cardinality of $A(x)$, which turns out to be an upper bound on the cardinality of $f_{\beta \to 2}(x)$. Note that for $y,z \in \{0,1\}^n$, satisfying $y \neq z$, we have
        \begin{equation}
            \left|\sum_{i=1}^{n} y_i 2^{-i} - \sum_{i=1}^{n} z_i 2^{-i}\right| \geq 2^{-n}.
        \end{equation}
        Therefore, for an interval $I = [a,b]$, there can be at most $2^n(b-a)$ sequences $y \in \{0,1\}^n$ that satisfy $\sum_{i=1}^{n} y_i 2^{-i} \in I$. 
        By (\ref{eq:J x:construction of a multivalued function computable wrt beta that maps beta expansions to binary expansions of a real number}) and (\ref{eq:definition of superset of size 2 -n more}),
        \begin{equation}
            J(x)^{(n)} = \left[\sum_{i=1}^{\nlogbeta} x_i \beta^{-i} - 2^{-n} - 2^{-n}, \sum_{i=1}^{\nlogbeta} x_i \beta^{-i}+ \frac{2^{-n}}{\beta-1} + 2^{-n}\right],
        \end{equation}
        so
        \begin{align}
            \#A(x) &\leq 2^n \left(\sum_{i=1}^{\nlogbeta} x_i \beta^{-i}+ \frac{2^{-n}}{\beta-1} + 2^{-n} - \left(\sum_{i=1}^{\nlogbeta} x_i \beta^{-i} - 2^{-n} - 2^{-n} \right)\right)\\
            &= 2^n \left(2^{-n}\frac{1}{\beta-1} + 3\cdot2^{-n}\right) = \frac{1}{\beta-1} + 3.
        \end{align}
        As $f_{\beta \to 2}(x) \subseteq A(x)$, we finally get
        \begin{equation}
            \#f_{\beta \to 2}(x) \leq \#A(x) \leq \frac{1}{\beta-1} + 3.
        \end{equation}
    \end{enumerate}
\end{proof}
\noindent Finally, we prove Theorem \ref{thm:binary expansions are at least as compressible as as beta expansions} as a corollary of Theorem \ref{thm:relatively computable multivalued functions and complexity} and Lemma \ref{lem:upper bound on the cardinality of the beta to 2 conversion multifunction}. 

\begin{corollary}\label{cor:binary expansions are at least as compressible as as beta expansions}
    Let $\beta \in (1,2)$. Then,
		\begin{equation}
			0 \leq \underline\Delta_\beta(\xf) \leq \bar\Delta_\beta(\xf),
		\end{equation}
		for all $\xf \in \Omega_\beta$.
\end{corollary}
\begin{proof}
    Let $\beta \in (1,2)$, $\xf \in \Omega_\beta$, and $s := \sum_{i=1}^\infty \xf_i \beta^{-i}$. By Lemma \ref{lem:upper bound on the cardinality of the beta to 2 conversion multifunction}, $f_{\beta \to 2} : \{0,1\}^\ast  \rightrightarrows \{0,1\}^\ast$ satisfies
    \begin{equation}
        \#f(\xf_{1:\nlogbeta}) \leq \frac{1}{\beta-1} + 3.
    \end{equation}
    Moreover, as $f_{\beta \to 2}$ is computable relatively to $\beta$, Theorem \ref{thm:relatively computable multivalued functions and complexity} yields that 
    \begin{align}
        K[y|\beta] &\leq K[\xf_{1:\nlogbeta}|\beta] + \log \left(\frac{1}{\beta-1} + 3\right) + \log\log \left(\frac{1}{\beta-1} + 3\right) + \bigOx[1]{n}\\
        &= K[\xf_{1:\nlogbeta}|\beta] + \bigOx[1]{n},
    \end{align}
    for all $y \in f_{\beta \to 2}(x) \supseteq \Sigma_2(s,n)$. In particular, if $\yf \in \{0,1\}^\N$ is the greedy binary expansion of $s$, then 
    \begin{equation}
        K[\yf_{1:n}|\beta] \leq K[\xf_{1:\nlogbeta}|\beta] + \bigOx[1]{n}.
    \end{equation}
    Finally, this yields
    \begin{equation}
        \bar{\Delta}_{\beta}(\xf) \geq \underline{\Delta}_{\beta}(\xf) = \liminf_{n\to \infty} \frac{K[\xf_{1:\nlogbeta}|\beta] - K[\yf_{1:n}|\beta]}{n} \geq \liminf_{n\to \infty}\bigOx[n^{-1}]{n} = 0.
    \end{equation}
\end{proof}

\endgroup

\begingroup
\let\clearpage\relax

\section{\texorpdfstring{Upper bounds on the complexity of $\beta$-expansions}{}}

\label{sec:upper bound on the complexity of beta expansions}

In this section, we establish the non-trivial upper bounds on $\underline \Delta_\beta$ and $\bar \Delta_\beta$, namely we prove Theorems \ref{thm:upper bound on the relative compressibility of beta expansions for almost all beta in 1 sqrt 2} and \ref{thm:upper bound on the relative compressibility of beta expansions for algebraic beta}. 

The base of the incoming proofs is the construction of a multivalued function $f_{2 \to \beta}$ that is computable relatively to $\beta$, which is built to be a kind of inverse of the function $f_{\beta \to 2}$ introduced in previous section. The heuristic guiding the construction of $f_{2 \to \beta}$ is exactly analoguous to the one we used for $f_{\beta \to 2}$, therefore we skip this part and directly give the definition of $f_{2 \to \beta}$.
\begin{definition}\label{def:multivalued function converting from base 2 to base beta}
    Let $f_{2 \to\beta} : \{0,1\}^\ast \rightrightarrows \{0,1\}^\ast$ be the multivalued function defined as 
\begin{equation}\label{eq:main:def:multivalued function converting from base 2 to base beta}
    f_{2 \to \beta}(x) := \left\{ y \in \{0,1\}^{\nlogbeta} :  \sum_{i=1}^{\nlogbeta} y_i \beta^{-i} \in_n J(x)\right\},
\end{equation}
for all $x \in \{0,1\}^\ast$, where $n := |x|$ and
\begin{equation}\label{eq:J x:def:multivalued function converting from base 2 to base beta}
    J(x) := \left[\sum_{i=1}^{n} x_i 2^{-i} - 2^{-n},  \sum_{i=1}^{n} x_i 2^{-i}+ 2^{-n} \right].
\end{equation}
\end{definition}
\noindent By the same arguments that led to the definition of $f_{\beta \to 2}$, one can establish that
\begin{equation}\label{eq:multivalued function converting from base 2 to base beta is a supset of the beta prefixes}
    \Sigma_\beta(s,\nlogbeta) \subseteq f_{2 \to \beta}(\xf_{1:n}),
\end{equation}
for all $s \in [0,1]$ and $\xf \in \Sigma_2(s)$. The main difference with previous section, is that $\#f_{2 \to \beta}(x)$ is not necessarily bounded by a constant for all $x \in \{0,1\}^\ast$. For example, it follows from \cite[Theorem 1.5]{feng2010GrowthRatebetaexpansions} that there exists $\alpha > 0$ such that 
\begin{equation}
    \# \Sigma_{\beta}(s,n) \geq 2^{\alpha n},
\end{equation}
for all $\beta \in \left(1,\frac{1+\sqrt{5}}{2}\right)$, $s \in [0,1]$ and from some $n \in \N$ onward. Therefore, it follows from (\ref{eq:multivalued function converting from base 2 to base beta is a supset of the beta prefixes}) that
\begin{equation}
    \#f_{2 \to \beta}(x) \geq 2^{\alpha n(\beta)}
\end{equation}
for all $x \in \{0,1\}^n$, $\beta \in \left(1,\frac{1+\sqrt{5}}{2}\right)$ and from some $n \in \N$ onward. Also note that a trivial bound for $\#f(x)$ is is $2^{n(\beta)}$, for all $x \in \{0,1\}^\ast$ with $n := \len x$, since $f(x)$ is a subset of $\{0,1\}^{n(\beta)}$. However, this trivial bound only allows to recover the trivial upper bound on $\underline \Delta_\beta$ and $\bar \Delta_\beta$ established in Lemma \ref{lem:trivial upper bound on the relative compressibility of beta expansions}. In the sequel, we show that we can improve this trivial upper bound in two cases: for almost all $\beta \in (1,2)$, and for algebraic $\beta \in (1,2)$.


\subsection{\texorpdfstring{Upper bound for almost all $\beta \in (1,2)$}{}}

In this part, we consider $\R$ as equipped with the Borel $\sigma$-algebra $\BB$ generated from the Euclidean topology. We denote by $\lambda$ the Lebesgue measure on $\R$. Fix $\beta \in (1,2)$, $x \in \{0,1\}^\ast$ and let $n := \len x$. We are interested at establishing an upper bound on $\#f_{2\to \beta}(x)$, with $f_{2\to \beta}(x)$ being defined by (\ref{eq:main:def:multivalued function converting from base 2 to base beta}). We can reformulate the definition of $f_{2\to \beta}(x)$ in terms of an integral over a certain domain, with respect to a certain measure. For $t \in \R$, we denote by $\delta_t$ the Dirac measure centered in $t$, formally defined as
\begin{equation}\label{eq:0:introduction to the formalism of Bernoulli convolution}
    \delta_t(A) = \begin{cases}
        1 \text{ if } t \in A,\\
        0 \text{ if } t \notin A,
    \end{cases}
\end{equation}
for all $A \in \BB$. Hence, for all $y \in \{0,1\}^{\nlogbeta}$, we have
\begin{equation}\label{eq:1:introduction to the formalism of Bernoulli convolution}
    \delta_{\sum_{i=1}^{\len y} y_i \beta^{-i}}(J(x)^{(n)}) = \begin{cases}
        1 \text{ if } \sum_{i=1}^{\len y} y_i \beta^{-i} \in J(x)^{(n)},\\
        0 \text{ if } \sum_{i=1}^{\len y} y_i \beta^{-i} \notin J(x)^{(n)},
    \end{cases}
\end{equation}
By summing over all $y \in \{0,1\}^{\nlogbeta}$, and by definition of $f_{\beta \to 2}(x)$ in (\ref{eq:main:def:multivalued function converting from base 2 to base beta}), we get
\begin{equation}\label{eq:2:introduction to the formalism of Bernoulli convolution}
    \sum_{y \in \{0,1\}^{\nlogbeta}}\delta_{\sum_{i=1}^{\len y} y_i \beta^{-i}}(J(x)^{(n)}) \geq \#f(x).
\end{equation}
This framework matches exactly that of the Bernoulli convolution, introduced in \cite[Section 6]{jessenDistributionFunctionsRiemann}. Define $C(I_\beta)$ to be the set of continuous functions $f : I_\beta \to \R$, and $\MM$ to be the set of regular Borel measures on $I_\beta$ (see Appendix \ref{sec:app:weak star limits}). Consider the sequence of regular Borel measures $(\nu_{\beta,m})_{m \in \N}$ defined by 
\begin{equation}\label{eq:3:introduction to the formalism of Bernoulli convolution}
    \nu_{\beta,m} := \frac{1}{2^m} \sum_{u \in \{0,1\}^{m}} \delta_{\sum_{i=1}^{m} u_i \beta^{-i}}, \ \forall m \in \N.
\end{equation}
Note that (\ref{eq:2:introduction to the formalism of Bernoulli convolution}) can be reformulated as
\begin{equation}\label{eq:4:introduction to the formalism of Bernoulli convolution}
    2^{\nlogbeta} \nu_{\beta,\nlogbeta}(J(x)) \geq \#f_{\beta \to 2}(x).
\end{equation}
The Bernoulli convolution $\nu_\beta$ is defined as being the weak limit in $\MM$ of the sequence $(\nu_{\beta,m})_{m \in \N}$, i.e., the unique regular Borel measure satisfies
\begin{equation}\label{eq:5:introduction to the formalism of Bernoulli convolution}
    \int_{I_\beta} f d\nu_\beta = \lim_{m \to \infty}\int_{I_\beta} f d\nu_{\beta,m},\ \forall f \in C(I_\beta).
\end{equation}
The fact that there is indeed such a regular Borel measure $\mu$ is a standard result of measure theory, see Appendix \ref{sec:app:weak star limits} for greater details. The Bernoulli convolution has been widely studied in the literature. In particular, it was shown in \cite{solomyak1995random} that for $\lambda$-almost all $\beta \in (1,2)$, $\nu_\beta$ is absolutely continuous with respect to the Lebesgue measure $\lambda$, i.e., for all $A \in \BB$, $\lambda(A) = 0 \Rightarrow \nu_\beta(A) = 0$. In particular, the Radon-Nykodym theorem \cite[Section 31, Theorem B]{halmos1950MeasureTheory} states that if a measure $\nu$ is absolutely continuous with respect to the Lebesgue measure $\lambda$, there exists a function $h : I_\beta \to \R$, called \newname{Radon-Nikodym derivative of $\nu$}, such that 
\begin{equation}
    \nu\left([a,b]\right) = \int_{a}^{b} h(x) dx,
\end{equation}
for all $a, b \in I_\beta$, $a < b$. We can exploit this to study the asymptotic behavior of $\#f_{2 \to \beta}(x)$.
\begin{theorem}\label{thm:asymptotic behavior of card f x when the Bernoulli convolution is absolutely continuous}
    Let $\beta \in (1,2)$ such that $\nu_\beta$ is absolutely continuous with respect to $\lambda$, and let $h_\beta : I_\beta \to \R$ be the Radon-Nikodym derivative of $\nu_\beta$. Then,
    \begin{equation}
        \limsup_{n\to\infty} 2^{-n\log_\beta(2/\beta)} \#f_{2 \to \beta}(\yf_{1:n}) \leq 2\left(4 + \frac{1}{\beta-1}\right) h_\beta\left(\sum_{i=1}^\infty \yf_i 2^{-i}\right),
    \end{equation}
    for all $\yf \in \{0,1\}^\N$.
\end{theorem}
\begin{proof}
    The proof follows that of \cite[Lemma 3.4]{kempton2013counting}. Let $\beta \in (1,2)$ such that $\nu_\beta$ is absolutely continuous with respect to $\lambda$, let $h_\beta : I_\beta \to \R$ be the Radon-Nikodym derivative of $\nu_\beta$, and fix $\yf \in \{0,1\}^\N$. For $n \in \N$, define $a_n$ and $b_n$ so that $J(\yf_{1:n}) = [a_n,b_n]$, according to (\ref{eq:J x:def:multivalued function converting from base 2 to base beta}). By definition,
    \begin{align}
        f_{2 \to \beta}(\yf_{1:n}) &=  \left\{ y \in \{0,1\}^{\nlogbeta} :  \sum_{i=1}^{\nlogbeta} y_i \beta^{-i} \in_{n} J(\yf_{1:n})\right\}\\
        &\subseteq \left\{ y \in \{0,1\}^{\nlogbeta} :  \sum_{i=1}^{\nlogbeta} y_i \beta^{-i} \in J(\yf_{1:n})^{(n)}\right\} =: A(\yf_{1:n}).
    \end{align}
    We study the cardinality of $A(\yf_{1:n})$, which turns out to be an upper bound to the cardinality of $f_{2 \to \beta}(\yf_{1:n})$. Fix $n \in \N$, and define, for all $m \geq \nlogbeta$, then set 
    \begin{equation}
        A_{m,n} := \{ z \in \{0,1\}^{m} : z_{1:\nlogbeta} \in A(\yf_{1:n}) \}.
    \end{equation}
    We make the following two remarks:
    \begin{enumerate}[label=(\alph*)]
        \item For all $y \in A(\yf_{1:n})$, there are $2^{m-\nlogbeta}$ elements in $A_{m,n}$ that are suffixes of $y$. This implies that
        \begin{equation}
            \# A_{m,n} \geq 2^{m - \nlogbeta} \# A(\yf_{1:n}).
        \end{equation}
        \item For every $z \in A_{m,n}$, since $z_{1:\nlogbeta} \in A(\yf_{1:n})$, we have that
        \begin{align}
            a_n - 2^{-n} &\overset{(a)}{\leq} \sum_{i=1}^{\nlogbeta} z_i \beta^{-i} \leq \sum_{i=1}^{m} z_i \beta^{-i} \leq \sum_{i=1}^{\infty} z_i \beta^{-i}\\
            & = \sum_{i=1}^{\nlogbeta} z_i \beta^{-i} + \sum_{i=\nlogbeta+1}^{\infty} z_i \beta^{-i}\\
            & \leq \sum_{i=1}^{\nlogbeta} z_i \beta^{-i} + \sum_{i=\nlogbeta+1}^{\infty} \beta^{-i}\\
            & = \sum_{i=1}^{\nlogbeta} z_i \beta^{-i} + \frac{\beta^{-\nlogbeta}}{\beta-1}\\
            & \overset{(b)}{\leq} b_n + 2^{-n} + \frac{\beta^{-\nlogbeta}}{\beta-1}\\
            & \leq b_n + 2^{-n} + \frac{2^{-n}}{\beta-1},
        \end{align}
        where (a) and (b) follow from the definition of $A(\yf_{1:n})$ and (\ref{eq:definition of superset of size 2 -n more}). Hence, for every $z \in A_{m,n}$,
        \begin{equation}
            \sum_{i=1}^{m} z_i \beta^{-i} \in \left[a_n - 2^{-n}, b_n + 2^{-n} + \frac{2^{-n}}{\beta-1}\right],
        \end{equation}
        which translates to
        \begin{equation}
            \#A_{m,n} \leq 2^m\nu_{\beta,m}\left(\left[a_n - 2^{-n}, b_n + 2^{-n} + \frac{2^{-n}}{\beta-1}\right]\right).
        \end{equation}
    \end{enumerate}
    Combining these two remarks, we get that 
    \begin{equation}
        2^{-\nlogbeta} \# A(\yf_{1:n}) \leq \nu_{\beta,m}\left(\left[a_n - 2^{-n}, b_n + 2^{-n} + \frac{2^{-n}}{\beta-1}\right]\right),
    \end{equation}
    which then translates to 
    \begin{equation}
        2^{-\nlogbeta} \# f_{2 \to \beta}(\yf_{1:n}) \leq \nu_{\beta,m}\left(\left[a_n - 2^{-n}, b_n + 2^{-n} + \frac{2^{-n}}{\beta-1}\right]\right),
    \end{equation}
    since $f_{2 \to \beta}(\yf_{1:n}) \subseteq A(\yf_{1:n})$.
    By taking the limit superior when $m \to \infty$, Lemma \ref{app:lem:sequence of weakly convergent measures can be upper bounded by the limit on compact sets} delivers
    \begin{equation}\label{eq:0:proof:thm:asymptotic behavior of card f x when the Bernoulli convolution is absolutely continuous}
        2^{-\nlogbeta} \# f_{2 \to \beta}(\yf_{1:n}) \leq \nu_{\beta}\left(\left[a_n - 2^{-n}, b_n + 2^{-n} + \frac{2^{-n}}{\beta-1}\right]\right).
    \end{equation}
    We assumed that $\nu_\beta$ is absolutely continuous and of Radon-Nikodym derivative $h_\beta$. Hence, by the Radon-Nikodym theorem, we have
    \begin{equation}\label{eq:1:proof:thm:asymptotic behavior of card f x when the Bernoulli convolution is absolutely continuous}
        \nu_\beta\left(\left[a_n - 2^{-n}, b_n + 2^{-n} + \frac{2^{-n}}{\beta-1}\right]\right) = \int_{a_n - 2^{-n}}^{b_n + 2^{-n} + \frac{2^{-n}}{\beta-1}} h_\beta(x) dx.
    \end{equation}
    By definition of $J(\yf_{1:n}) = [a_n,b_n]$, one has $b_n - a_n = 2\cdot 2^{-n}$, and $\lim_{n \to \infty} a_n = \lim_{n \to \infty} b_n = \sum_{i=1}^\infty \yf_i 2^{-i}$. Incorporating this fact in (\ref{eq:1:proof:thm:asymptotic behavior of card f x when the Bernoulli convolution is absolutely continuous}) yields
    \begin{align}
        \lim_{n \to \infty} \frac{\nu_\beta\left(\left[a_n- 2^{-n}, b_n + 2^{-n} + \frac{2^{-n}}{\beta-1}\right]\right)}{ 2^{-n}\left(4+\frac{1}{\beta-1}\right)}& = \lim_{n \to \infty} \frac{\nu_\beta\left(\left[a_n - 2^{-n}, b_n + 2^{-n} + \frac{2^{-n}}{\beta-1}\right]\right)}{b_n + 2^{-n} + \frac{2^{-n}}{\beta-1} - \left(a_n - 2^{-n}\right)}\\
        & = \lim_{n \to \infty} \frac{\int_{a_n - 2^{-n}}^{b_n + 2^{-n} + \frac{2^{-n}}{\beta-1}} h_\beta(x) dx}{b_n + 2^{-n} + \frac{2^{-n}}{\beta-1} - a_n + 2^{-n}}\\
        &= h_\beta\left( \lim_{n \to \infty} a_n \right)\\
        &= h_\beta\left( \sum_{i=1}^\infty \yf_i 2^{-i}\right).
    \end{align}
    Combining this result with (\ref{eq:0:proof:thm:asymptotic behavior of card f x when the Bernoulli convolution is absolutely continuous}) delivers
    \begin{align}
        \limsup_{n\to \infty}2^{-(n(\beta)-n)} \# f_{2 \to \beta}(\yf_{1:n}) &\leq  \limsup_{n \to \infty}\frac{\nu_\beta\left(\left[a_n - 2^{-n}, b_n + 2^{-n} + \frac{2^{-n}}{\beta-1}\right]\right)}{ 2^{-n}}\\
        &= \lim_{n \to \infty}\frac{\nu_\beta\left(\left[a_n - 2^{-n}, b_n + 2^{-n} + \frac{2^{-n}}{\beta-1}\right]\right)}{ 2^{-n}}\\
        &= \left(4+ \frac{1}{\beta-1}\right)h_\beta\left( \sum_{i=1}^\infty \yf_i \beta^{-i}\right).
    \end{align}
    We conclude the proof by noting that 
    \begin{equation}
        2^{-(n(\beta)-n)} \geq 2^{-(n\log_\beta(2) + 1 -n)} = 2 \cdot 2^{- n \log_\beta(2/\beta)}, \ \forall n \in \N. 
    \end{equation}
\end{proof}
\noindent As a direct Corollary, we get Theorem \ref{thm:upper bound on the relative compressibility of beta expansions for almost all beta in 1 sqrt 2}.
\begin{corollary}\label{cor:upper bound on the relative compressibility of beta expansions for almost all beta in 1 2}
    For almost all $\beta \in (1,2)$,
		\begin{equation}
			0 \leq \underline\Delta_\beta(\xf) \leq \bar\Delta_\beta(\xf) \leq \log_\beta\left(\frac{2}{\beta}\right),
		\end{equation}
		for all $\xf \in \Omega_\beta$.
\end{corollary}
\begin{proof}
    Let $\beta \in (1,2)$ such that $\nu_\beta$ is absolutely continuous with respect to $\lambda$, and let $h_\beta : I_\beta \to \R$ be the Radon-Nikodym derivative of $\nu_\beta$. Fix $\xf \in \Omega_\beta$, define $s := \sum_{i=1}^\infty \xf_i\beta^{-i} \in [0,1]$ and let $\yf$ be the greedy binary expansion of $s$. By Theorem \ref{thm:asymptotic behavior of card f x when the Bernoulli convolution is absolutely continuous}, we have
    \begin{equation}
        \limsup_{n\to\infty} 2^{-n\log_\beta(2/\beta)} \#f_{2 \to \beta}\yf_{1:n}) \leq 2\left(4 + \frac{1}{\beta-1}\right) h_\beta\left(\sum_{i=1}^\infty \yf_i 2^{-i}\right).
    \end{equation}
    Therefore, there exists $N > 0$ such that for all $n \geq N$, 
    \begin{equation}
        2^{-n\log_\beta(2/\beta)} \#f_{2 \to \beta}(\yf_{1:n}) \leq 2\left(4 + \frac{1}{\beta-1}\right) h_\beta\left(\sum_{i=1}^\infty \yf_i 2^{-i}\right) + 1 =: C,
    \end{equation}
    so 
    \begin{equation}
        \#f_{2 \to \beta}(\yf_{1:n}) \leq C \cdot 2^{n\log_\beta(2/\beta)},
    \end{equation}
    for large enough $n \in \N$. This results in 
    \begin{align}
        &\limsup_{n \to \infty} \frac{\log\#(f_{2 \to \beta}(\yf_{1:n}))}{n} \leq \log_\beta(2/\beta)\label{eq:log f:cor:upper bound on the relative compressibility of beta expansions for almost all beta in 1 2}\\
        &\limsup_{n \to \infty} \frac{\log\log\#(f_{2 \to \beta}(\yf_{1:n}))}{n} = 0.\label{eq:log log f:cor:upper bound on the relative compressibility of beta expansions for almost all beta in 1 2}
    \end{align}
    Finally, note that by definition of $f_{2 \to \beta}$, $\xf_{1:\nlogbeta} \in f_{2 \to \beta}(\yf_{1:n})$ for all $n \in \N$. Hence, by Theorem \ref{thm:computable multivalued functions and complexity}, one has
    \begin{align}
        K[\xf_{1:\nlogbeta}] &\leq K[\yf_{1:n}] + \log\#(f_{2 \to \beta}(\yf_{1:n})) + \log\log\#(f_{2 \to \beta}(\yf_{1:n})) + \bigOx[1]{n}.
    \end{align}
    By incorporation of (\ref{eq:log f:cor:upper bound on the relative compressibility of beta expansions for almost all beta in 1 2}) and (\ref{eq:log log f:cor:upper bound on the relative compressibility of beta expansions for almost all beta in 1 2}), we get
    \begin{equation}
        \bar\Delta_\beta(\xf) = \limsup_{n\to\infty} \frac{K[\xf_{1:\nlogbeta}]-K[\yf_{1:n}]}{n} \leq \log_\beta(2/\beta).
    \end{equation}
\end{proof}





\subsection{\texorpdfstring{Upper bound when $\beta$ is algebraic}{}}

In this section, we work with $\beta \in (1,2)$ algebraic, and we establish Theorem \ref{thm:upper bound on the relative compressibility of beta expansions for algebraic beta}. Recall that for a fixed number $\beta \in (1,2)$, we have defined the equivalence relationship $\sim_\beta$ over $\{0,1\}^\ast$ by 
\begin{equation}
    x \sim_\beta y \iff n := |x| = |y| \ \text{and} \ \sum_{i=1}^n x_i \beta^{-i} = \sum_{i=1}^n y_i \beta^{-i}, \ \forall x,y \in \{0,1\}^\ast.
\end{equation}
Moreover, we have defined the function $M_\beta : \{0,1\}^\ast \to \{0,1\}^\ast$ such that $M_\beta(x)$ is the lexicographically maximal element of the equivalence class $[x]_\beta := \{y \in \{0,1\}^\ast : x \sim_\beta y\}$, for all $x \in \{0,1\}^\ast$. The pivotal result allowing to establish Theorem \ref{thm:upper bound on the relative compressibility of beta expansions for algebraic beta} is a mere generalization of the famous ``separation lemma'' \cite[Lemma 1.51]{garsia1962ArithmeticPropertiesBernoulli}, which establishes a lower bound on the distance between the numbers of the form $\sum_{i=1}^n x_i \beta^{-i}$, $x \in \{0,1\}^n$, if $\beta \in (1,2)$ belongs to a certain class of algebraic numbers. We adapt the proof of this Lemma to get a similar result for $\beta \in (1,2)$ being any algebraic number. For an algebraic number $\beta \in (1,2)$, we denote by $L_\beta$ the leading coefficient of its minimal polynomial, and by $G_\beta$ the set of its Galois conjugates. We further define 
\begin{align}
    G_\beta^+ &:= \{ z \in G_\beta : |z| > 1\},\\
    G_\beta^1 &:= \{ z \in G_\beta : |z| = 1\},
\end{align}
and
\begin{equation}\label{eq:definition of pi beta - pi beta + and k beta}
    \Pi_\beta := \prod_{z \in G_\beta} |1-|z||, \ \ \ \ \ \ \Pi_\beta^+ := \prod_{z \in G_\beta^+} |z|, \ \ \ \ \ \ \text{and} \ \ \ \ \ \ k_\beta := \# G_\beta^1.
\end{equation}
The result is expressed as follows.
\begin{lemma}\label{lem:generalization of the garsia separation lemma}
    Let $\beta \in (1,2)$ be an algebraic number, $n \in \N$, and let $x,y \in \{0,1\}^n$ such that $x \not\sim_\beta y$. Then,
    \begin{equation}
       \left|\sum_{i=1}^n x_i \beta^{-i} - \sum_{i=1}^n y_i \beta^{-i}\right| \geq \frac{L_\beta\Pi_\beta}{\displaystyle  n^{k_\beta}\left(\beta L_\beta \Pi_\beta^+ \right)^n}\cdot
    \end{equation}
\end{lemma}
As this result is tied to algebraic considerations that deviate a lot from the message of this paper, we postpone the proof to Appendix \ref{sec:app:separation lemma}. This Lemma has two important consequences, that are essential to establish \ref{thm:upper bound on the relative compressibility of beta expansions for algebraic beta}. First, this implies that we can lower bound the distance between numbers of the form $\sum_{i=1}^n (Mx)_i \beta^{-i}$, $x \in \{0,1\}^n$, and second, we can show that $M_\beta$ is computable.
\begin{corollary}\label{cor: separation lemma for lexicographically maximal sequences}
    Let $\beta \in (1,2)$ be an algebraic number and $n \in \N$. Then, for all $x,y \in M_\beta\left(\{0,1\}^n\right)$ such that $x \neq y$,
    \begin{equation}\label{eq:main:cor: separation lemma for lexicographically maximal sequences}
        \left|\sum_{i=1}^n x_i \beta^{-i} - \sum_{i=1}^n y_i \beta^{-i}\right| \geq \frac{L_\beta\Pi_\beta}{\displaystyle  n^{k_\beta}\left(\beta L_\beta \Pi_\beta^+ \right)^n}\cdot
    \end{equation}
\end{corollary}
\begin{proof}
    Let $\beta \in (1,2)$ be an algebraic number, $n \in \N$ and $x,y \in M_\beta\left(\{0,1\}^n\right)$ such that $x \neq y$. The result is established by showing that $x \not\sim_\beta y$. Indeed, if we prove that $x \not\sim_\beta y$, then we can immediately apply Lemma \ref{lem:generalization of the garsia separation lemma} to get (\ref{eq:main:cor: separation lemma for lexicographically maximal sequences}).
    
    We now proceed to prove that $x \not\sim_\beta y$. Since $x,y \in M_\beta(\{0,1\}^n)$, there exists $u,v \in \{0,1\}^n$ such that $x = M_\beta(u)$ and $y = M_\beta(v)$. By definition of $M_\beta$, $x = M_\beta(u) \in [u]_\beta$ and $y = M_\beta(v) \in [v]_\beta$, so $x \sim_\beta u$ and $y \sim_\beta v$. We now show that $u \not\sim_\beta v$, yielding $x \not\sim_\beta y$. By means of contradiction, suppose that $u \sim_\beta v$. Then, $[u]_\beta = [v]_\beta$ and hence $x = M_\beta(u) = M_\beta(v) = y$. This is in contradiction with $x \neq y$.
\end{proof}
\begin{lemma}\label{lem:bounding the size of the output set of the multivalued function projected on lexicographically maximal sequences}
    Let $\beta \in (1,2)$ be an algebraic number and $n \in \N$. Then,
    \begin{equation}\label{eq:main:lem:bounding the size of the output set of the multivalued function projected on lexicographically maximal sequences}
        \# (M_\beta \circ f_{2 \to \beta}(\yf_{1:n})) \leq \frac{4}{L_\beta\Pi_\beta}\nlogbeta^{k_\beta}\left(L_\beta \Pi_\beta^+ \right)^{\nlogbeta}\cdot
    \end{equation}
\end{lemma}
\begin{proof}
    Let $\beta \in (1,2)$ be an algebraic number. Define $F :=  M_\beta \circ f_{2 \to \beta}$, and let $\yf \in \{0,1\}^\N$ and $n \in \N$. We will proceed to find an upper bound for $\# F(\yf_{1:n})$. 
    \begin{enumerate}[label=(\alph*)]
        \item Similary as in the proof of Theorem \ref{thm:asymptotic behavior of card f x when the Bernoulli convolution is absolutely continuous}, we define a set $A(\yf_{1:n}) \supseteq f_{2 \to \beta}(\yf_{1:n})$ by the observation that
        \begin{align}
            f_{2 \to \beta}(\yf_{1:n}) &=  \left\{ y \in \{0,1\}^{\nlogbeta} :  \sum_{i=1}^{\nlogbeta} y_i \beta^{-i} \in_{n} J(\yf_{1:n})\right\}\\
            &\subseteq \left\{ y \in \{0,1\}^{\nlogbeta} :  \sum_{i=1}^{\nlogbeta} y_i \beta^{-i} \in J(\yf_{1:n})^{(n)}\right\} =: A(\yf_{1:n}).
        \end{align}
        We define $\hat F(\yf_{1:n}) := M_\beta(A(\yf_{1:n}))$. Note that $F(\yf_{1:n}) \subseteq \hat F(\yf_{1:n})$. We hence study the cardinality of $\hat F(\yf_{1:n})$, which will deliver an upper bound on the cardinality of $F(\yf_{1:n})$.
        \item First note that $\hat F(\yf_{1:n}) \subseteq A(\yf_{1:n})$. Indeed, let $x \in \hat F(\yf_{1:n})$. Then, there exists $u \in A(\yf_{1:n})$ such that $M_\beta(u) = x$. Since $M_\beta^{-1}(\{x\}) = [x]_\beta$, then we deduce that $u \sim_\beta x$, i.e.,
        \begin{equation}\label{eq:0:proof:cor:upper bound on the relative compressibility of beta expansions for algebraic beta}
            \sum_{i=1}^{\nlogbeta} x_i \beta^{-i} = \sum_{i=1}^{\nlogbeta} u_i \beta^{-i} \overset{(a)}{\in} J(\yf_{1:n})^{(n)},
        \end{equation}
        where (a) follows from the definition of $A(\yf_{1:n})$. We conclude that, indeed, $x \in A(\yf_{1:n})$, so $\hat F(\yf_{1:n}) \subseteq A(\yf_{1:n})$.
        \item As a direct consequence, remark that by denoting $m_n$ to be the smallest distance between two different numbers of the form $\sum_{i=1}^{\nlogbeta} x_i \beta^{-i}$, $x \in \hat F(\yf_{1:n})$, we have that 
        \begin{equation}\label{eq:1:proof:cor:upper bound on the relative compressibility of beta expansions for algebraic beta}
           m_n \#\hat  F(\yf_{1:n}) \leq |J(\yf_{1:n})|^{(n)} \oseteq{\ref{eq:definition of superset of size 2 -n more}}{=} |J(\yf_{1:n})| + 2\cdot 2^{-n} \oseteq{\ref{eq:J x:def:multivalued function converting from base 2 to base beta}}{=} 4 \cdot 2^{-n}.
        \end{equation}
        \item Finally, since $A(\yf_{1:n}) \subseteq \{0,1\}^{\nlogbeta}$, then $\hat F(\yf_{1:n}) \subseteq M_\beta(\{0,1\}^{\nlogbeta})$. By Corollary \ref{cor: separation lemma for lexicographically maximal sequences}, this shows that 
        \begin{equation}\label{eq:2:proof:cor:upper bound on the relative compressibility of beta expansions for algebraic beta}
            m_n \geq \frac{L_\beta\Pi_\beta}{\displaystyle  \nlogbeta^{k_\beta}\left(\beta L_\beta \Pi_\beta^+\right)^{\nlogbeta}} = \frac{2^{-n}L_\beta\Pi_\beta}{\displaystyle \nlogbeta^{k_\beta}\left(L_\beta \Pi_\beta^+\right)^{\nlogbeta}}\cdot
        \end{equation}
    \end{enumerate}
    Combining (\ref{eq:1:proof:cor:upper bound on the relative compressibility of beta expansions for algebraic beta}) and (\ref{eq:2:proof:cor:upper bound on the relative compressibility of beta expansions for algebraic beta}), we get 
    \begin{equation}
       \#F(\yf_{1:n}) \leq \# \hat F(\yf_{1:n}) \leq \frac{\displaystyle  4\nlogbeta^{k_\beta}\left(L_\beta \Pi_\beta^+ \right)^{\nlogbeta}}{L_\beta \Pi_\beta}\cdot
    \end{equation}
\end{proof}
\noindent We further show that $M_\beta \circ f_{\beta \to 2}$ is computable relatively to $\beta$.
\begin{lemma}\label{cor:M is computable for algebraic beta}
    Let $\beta \in (1,2)$. Then, $M_\beta \circ f_{2 \to \beta}$ is computable relatively to $\beta$.
\end{lemma}
\begin{proof}
    The proof relies on showing that $M_\beta$ is computable. Indeed, we already know that $f_{2\to \beta}$ is computable relatively to $\beta$, so $M_\beta$ being computable implies that $M_\beta \circ f_{2\to \beta}$ is computable relatively to $\beta$.

    Note that the equivalence classes $[x]_\beta$, $x \in \{0,1\}^\ast$ can be seen as the following multivalued function 
    \begin{equation}\label{eq:definition of equivalence class as a multivalued function:cor:M is computable for algebraic beta}
        \begin{array}{rccc}
            [\cdot]_\beta : &\{0,1\}^\ast &\rightrightarrows &\{0,1\}^\ast\\
            & x & \mapsto & [x]_\beta.
        \end{array}
    \end{equation}
    Note that $M_\beta = \max_L \circ [\cdot]_\beta$, hence we simply have to prove that $[\cdot]_\beta$ is computable to prove that $M_\beta$ is computable.
   
    Let $\beta \in (1,2)$. If $\beta$ is not algebraic, then $[x]_\beta = {x}$ for all $x\in \{0,1\}^\ast$, so $[\cdot]_\beta$ is computable.

    Now, fix $\beta \in (1,2)$ algebraic, and $x \in \{0,1\}^\ast$. We define $n := |x|$. Let $P_\beta$ be the minimal polynomial of $\beta$ of leading coefficient $L_\beta$, of degree $d \in \N$, and denote by $\alpha_1 < \alpha_2 < \ldots < \alpha_k$ its real roots, for some $k \leq d$. Note that $\beta = \alpha_{\ell_\beta}$ for some $\ell_\beta \in \{1,\ldots, k\}$. Recall the definition of $\Pi_\beta$, $\Pi_\beta^+$ and $k_\beta$ in (\ref{eq:definition of pi beta - pi beta + and k beta}). Define two rational numbers $q,q_+$ such that $0 < q \leq L_\beta \Pi_\beta$ and $q_+ \geq \beta L_\beta \Pi_\beta^+$. We construct Algorithm \ref{alg:algorithm that computes M beta} that computes $x \mapsto \langle [x]_\beta \rangle$.
    \begin{algorithm}[ht]
        \caption{Algorithm for computing $x \mapsto \langle [x]_\beta \rangle$}
        \label{alg:algorithm that computes M beta}
        \begin{algorithmic}[1]
            \Require $x \in \{0,1\}^\ast$
            \State $n \gets \len x$.
            \State $\varepsilon \gets \frac{q}{n^{k_\beta} q_+^n}\cdot$. \label{line:definition epsilon:alg:algorithm that computes M beta}
            \State Find a rational approximation $q_i$ of each real root $\alpha_i$ of $P_\beta$ up to precision $\varepsilon/(8n)$.\label{line:approximation roots of P beta:alg:algorithm that computes M beta}
            \State $q_\beta \gets q_{\ell_\beta}$.\label{line:definition of q beta:alg:algorithm that computes M beta}
            \State Find the set $X$ of all the sequences $u \in \{0,1\}^n$ that satisfy \label{line:definition of X:alg:algorithm that computes M beta}
            \begin{equation}\label{eq:definition of X:alg:algorithm that computes M beta}
                \left| \sum_{i=1}^n x_i q_\beta^{-i} - \sum_{i=1}^n u_i q_\beta^{-i} \right| \leq \varepsilon/4.
            \end{equation}
            \State \Return $\langle X \rangle$.
        \end{algorithmic}
        \end{algorithm}

        We prove that Algorithm \ref{alg:algorithm that computes M beta} indeed computes $x \mapsto \langle [x]_\beta \rangle$. The crucial part is to prove that the set $X$ defined in line \ref{line:definition of X:alg:algorithm that computes M beta} is equal to $[x]_\beta$, for all $x \in \{0,1\}^\ast$, which relies deeply on the separation lemma \ref{lem:generalization of the garsia separation lemma}.
        \begin{enumerate}[label=(\alph*)]
            \item First, we show a result on the regularity of $\beta$-expansions. Let $\beta_1,\beta_2 \in (1,2)$ and $n \in \N$. Then, for all $x \in \{0,1\}^n$,
            \begin{align}
                \left| \sum_{i=1}^n x_i \beta_1^{-i} - \sum_{i=1}^n x_i \beta_2^{-i}\right|& \leq  \sum_{i=1}^n x_i\left| \beta_1^{-i} - \beta_2^{-i}\right| \leq  \sum_{i=1}^n x_i\left| \beta_1^{-1} - \beta_2^{-1}\right|\\
                & \leq n \left| \beta_1^{-1} - \beta_2^{-1}\right| = \frac{n}{\beta_1\beta_2} \left| \beta_2 - \beta_1\right| \leq n \left| \beta_2 - \beta_1\right|.
            \end{align}
            This regularity condition can be further extended as follows.
            \begin{align}
                &\left| \sum_{i=1}^n x_i \beta_1^{-i} - \sum_{i=1}^n u_i \beta_1^{-i} \right| \\
                &= \left| \sum_{i=1}^n x_i \beta_1^{-i} - \sum_{i=1}^n x_i \beta_2^{-i} + \sum_{i=1}^n x_i \beta_2^{-i} - \sum_{i=1}^n u_i \beta_2^{-i} + \sum_{i=1}^n u_i \beta_2^{-i} - \sum_{i=1}^n u_i \beta_1^{-i} \right|\\
                &\leq \left| \sum_{i=1}^n x_i \beta_1^{-i} - \sum_{i=1}^n x_i \beta_2^{-i}\right| +  \left|\sum_{i=1}^n x_i \beta_2^{-i} - \sum_{i=1}^n u_i \beta_2^{-i}\right| +  \left|\sum_{i=1}^n u_i \beta_2^{-i} - \sum_{i=1}^n u_i \beta_1^{-i} \right|\\
                &\leq n |\beta_2 - \beta_1| + \left|\sum_{i=1}^n x_i \beta_2^{-i} - \sum_{i=1}^n u_i \beta_2^{-i}\right| + n |\beta_2 - \beta_1|\\
                &\leq 2n |\beta_2 - \beta_1| + \left|\sum_{i=1}^n x_i \beta_2^{-i} - \sum_{i=1}^n u_i \beta_2^{-i}\right|, \label{eq:bounding the error with q beta by the error with beta}
            \end{align}
            for all $x,u \in \{0,1\}^n$.
            \item We now use the above inequation to show that $X = [x]_\beta$, for all $x \in \{0,1\}^\ast$. Recall that $q_\beta$ denotes an approximation of $\beta$ up to precision $\varepsilon/(8n)$, as defined in line \ref{line:definition of q beta:alg:algorithm that computes M beta}. Let $x \in \{0,1\}^\ast$, define $n \in \N$, and let $u \in \{0,1\}^n$.
            \begin{enumerate}[label = (\roman*)]
                \item Suppose that $u \in [x]_\beta$. Then,
                \begin{equation}
                    \left| \sum_{i=1}^n x_i q_\beta^{-i} - \sum_{i=1}^n u_i q_\beta^{-i} \right| \oseteq{\ref{eq:bounding the error with q beta by the error with beta}}{\leq} \varepsilon/4 + \left|\sum_{i=1}^n x_i \beta^{-i} - \sum_{i=1}^n u_i \beta^{-i}\right| \overset{(a)}{=} \varepsilon/4,
                \end{equation}
                where (a) follows from $u \in [x]_\beta$. Hence, $u \in X$.
                \item Let $u \notin [x]_\beta$. Then, by Lemma \ref{lem:generalization of the garsia separation lemma},
                \begin{equation}
                    \left|\sum_{i=1}^n x_i \beta^{-i} - \sum_{i=1}^n u_i \beta^{-i}\right| \geq \frac{L_\beta\Pi_\beta}{\displaystyle  n^{k_\beta}\left(\beta L_\beta \Pi_\beta^+ \right)^n} \geq \frac{q}{n^{k_\beta} q_+^n} \overset{(a)}{=} \varepsilon.
                \end{equation}
                where (a) is by definition of $\varepsilon$ in line \ref{line:definition epsilon:alg:algorithm that computes M beta}. It follows that 
                \begin{equation}
                    \left| \sum_{i=1}^n x_i q_\beta^{-i} - \sum_{i=1}^n u_i q_\beta^{-i} \right| \oseteq{\ref{eq:bounding the error with q beta by the error with beta}}{\geq} -\varepsilon/4 + \left|\sum_{i=1}^n x_i \beta^{-i} - \sum_{i=1}^n u_i \beta^{-i}\right| \geq 3\varepsilon/4 > \varepsilon/4.
                \end{equation}
                Therefore, $u \notin X$.
            \end{enumerate}
            This concludes the proof that $X = [x]_\beta$.
        \end{enumerate}
\end{proof}
\noindent As a final corollary, we establish Theorem \ref{thm:upper bound on the relative compressibility of beta expansions for algebraic beta}.
\begin{corollary}\label{cor:upper bound on the relative compressibility of beta expansions for algebraic beta}
    Let $\beta \in (1,2)$ be an algebraic number. Then,
		\begin{equation}
			0 \leq \underline\Delta_\beta(M_\beta\xf) \leq \bar\Delta_\beta(M_\beta\xf) \leq \log_\beta\left(L_\beta \Pi_\beta^+\right),
		\end{equation}
		for all $\xf \in \Omega_\beta$.
\end{corollary}
\begin{proof}
    Let $\beta \in (1,2)$ be an algebraic number and $\xf \in \Omega_\beta$. Define $s := \sum_{i=1}^\infty \xf_i \beta^{-i}$, and $\yf \in \{0,1\}^\N$ to be the greedy binary expansion of $s$. Recall that by (\ref{eq:multivalued function converting from base 2 to base beta is a supset of the beta prefixes}),
    \begin{equation}
        \Sigma_\beta(s,\nlogbeta) \subseteq f_{2 \to \beta}(\yf_{1:n}),
    \end{equation}
    which implies that 
    \begin{equation}
        M_\beta(\Sigma_\beta(s,\nlogbeta)) \subseteq M_\beta\circ f_{2 \to \beta}(\yf_{1:n}).
    \end{equation}
    Since $\xf_{1:\nlogbeta} \in \Sigma_\beta(s,\nlogbeta)$, then 
    \begin{equation}
        (M_\beta\xf)_{1:\nlogbeta} \in M_\beta(\Sigma_\beta(s,\nlogbeta)) \subseteq M_\beta\circ f_{2 \to \beta}(\yf_{1:n}).
    \end{equation}
    By Corollary \ref{cor:M is computable for algebraic beta}, $F := M_\beta \circ f_{2 \to \beta}$ is computable relatively to $\beta$. By Theorem \ref{thm:computable multivalued functions and complexity}, this yields
    \begin{align}
        K[(M_\beta\xf)_{1:\nlogbeta}|\beta]\leq K[\yf_{1:n}|\beta] + \log \# F(\yf_{1:n}) + \log\log \# F(\yf_{1:n}) + \bigOx[1]{n}.
    \end{align}
    By Lemma \ref{lem:bounding the size of the output set of the multivalued function projected on lexicographically maximal sequences}, 
    \begin{equation}
        \# (F(\yf_{1:n})) \leq \frac{4}{L_\beta\Pi_\beta}  \nlogbeta^{k_\beta}\left(L_\beta \Pi_\beta^+ \right)^{\nlogbeta}\cdot
    \end{equation}
    Note that then,
    \begin{equation}
        \limsup_{n \to \infty} \frac{\log_2\# (F(\yf_{1:n}))}{n} \leq \log_\beta\left(L_\beta \Pi_\beta^+\right),
    \end{equation}
    and 
    \begin{equation}
        \limsup_{n \to \infty} \frac{\log_2 \log_2\# (F(\yf_{1:n}))}{n} = 0.
    \end{equation}
    We conclude that
    \begin{equation}
        \bar\Delta_\beta(M_\beta\xf) = \limsup_{n\to \infty} \frac{K[(M_\beta\xf)_{1:\nlogbeta}|\beta] - K[\yf_{1:n}|\beta]}{n} \leq \log_\beta\left(L_\beta \Pi_\beta^+\right).
    \end{equation}
\end{proof}

\endgroup

\begingroup
\let\clearpage\relax

\section{A fast algorithm to control complexity}

\label{sec:a fast algorithm to control complexity}

Recall that for $\beta \in (1,2)$, we have defined an equivalence relationship as follows
\begin{equation}
    x \sim_\beta y \iff n := \len x = \len y \ \text{and} \ \sum_{i=1}^n x_i \beta^{-i} = \sum_{i=1}^n y_i \beta^{-i},
\end{equation}
for all $x \in \{0,1\}^\ast$. The subsequent equivalence classes are denoted $[x]_\beta \in \{0,1\}^\ast / \sim_\beta$. We also have proven that the function $M_\beta$ that maps any $x \in \{0,1\}^\ast$ to the lexicographically maximal element of $[x]_\beta$ is computable, by showing that Algorithm \ref{alg:algorithm that computes M beta} delivers $\langle [x]_\beta \rangle$ on input $x \in \{0,1\}^\ast$, which we can then use to find the lexicographically maximal element of $[x]_\beta$. However, we can be convinced that this algorithm might run for a very long time before delivering its input. The main bottleneck is the line \ref{line:definition of X:alg:algorithm that computes M beta} of Algorithm \ref{alg:algorithm that computes M beta}. As stated, the algorithm evaluates, for every $u \in \{0,1\}^{|x|}$, if 
\begin{equation}
    \left| \sum_{i=1}^n x_i \beta^{-i} - \sum_{i=1}^n u_i \beta^{-i}\right| \leq \varepsilon/4
\end{equation}
is satisfied. This line hence requires $2^{|x|}$ steps, making this algorithm of exponential complexity. In this section, we show that if $\beta \in (1,2)$ is a Pisot number, we can construct another algorithm to compute $M_\beta$, that computes in linear time, i.e., on input $x \in \{0,1\}^\ast$, the algorithm delivers its output in $\bigOw[|x|]{x}$ steps.
We first make the observation that $\Sigma_\beta(s,n)$ contains the equivalence classes generated by its elements, for all $s \in I_\beta$ and $n \in \N$.
\begin{lemma}\label{lem:sigma beta s n contains the equivalence classes}
    Let $\beta \in (1,2)$, $s \in I_\beta$, and $n \in \N$. Then, for all $x \in \Sigma_\beta(s,n)$, $[x]_\beta \subseteq \Sigma_\beta(s,n)$.
\end{lemma}
\begin{proof}
    Let $\beta \in (1,2)$, $s \in I_\beta$, $n \in \N$, $x \in \Sigma_\beta(s,n)$ and let $y \in [x]_\beta$. Then, $y \sim_\beta x$, so 
    \begin{equation}
        \sum_{i=1}^n y_i \beta^{-i} = \sum_{i=1}^n x_i \beta^{-i} \in \left[ s - \frac{\beta^{-n}}{\beta-1}, s\right],
    \end{equation}
    which by combination of (\ref{eq:1:construction of a multivalued function computable wrt beta that maps beta expansions to binary expansions of a real number}) and (\ref{eq:2:construction of a multivalued function computable wrt beta that maps beta expansions to binary expansions of a real number}) implies that $y \in \Sigma_\beta(s,n)$.
\end{proof}
\noindent For $s \in I_\beta$, we define
\begin{equation}\label{eq:definition set of beta expansions wrt to a binary sequence}
    \Sigma_\beta(x) := \Sigma_\beta\left(\sum_{i=1}^{\len x} x_i \beta^{-i}, \len x\right),
\end{equation}
and
\begin{equation}\label{eq:definition set of lexicographically maximal elements of equivalence classes in beta expansions of sum x beta -i}
    \hat \Sigma_\beta(x) := \left\{ \lexmax (u) : u \in \Sigma_\beta(x) / \sim_\beta\right\}.
\end{equation}
\begin{lemma}
    Let $\beta \in (1,2)$, $x \in \{0,1\}^\ast$, and $n := \len x$. Then,
    \begin{equation}
        M_\beta(x) = \underset{u \in \hat \Sigma_\beta(x)}{\arg\max} \sum_{i=1}^n u_i \beta^{-i}.
    \end{equation}
\end{lemma}
\begin{proof}
    Let $\beta \in (1,2)$, $x \in \{0,1\}^\ast$, and $n := \len x$. We first prove that $M_\beta(x) \in \hat \Sigma_\beta(x)$, and then that $M_\beta(x)$ is the element of $\hat \Sigma_\beta(x)$ that maximizes the function $u \mapsto \sum_{i=1}^n u_i \beta^{-i}$.
    \begin{enumerate}[label = (\alph*)]
        \item Let $s_x := \sum_{i=1}^n x_i \beta^{-i}$. By (\ref{eq:definition set of beta expansions wrt to a binary sequence}), $x \in \Sigma_\beta(x) = \Sigma_\beta(s_x, n)$. By Lemma \ref{lem:sigma beta s n contains the equivalence classes}, $[x]_\beta \subseteq \Sigma_\beta(s_x,n)$, which implies in particular that $M_\beta(x) \in \Sigma_\beta(s_x,n) = \Sigma_\beta(x)$, since by definition $M_\beta(x) \in [x]_\beta$. Furthermore, by definition, $M_\beta(x) = \lexmax ([x]_\beta) = \lexmax ([M_\beta(x)]_\beta)$, and we get by (\ref{eq:definition set of lexicographically maximal elements of equivalence classes in beta expansions of sum x beta -i}) that $M_\beta(x) \in \hat \Sigma_\beta(x)$.
        \item Let $u \in \hat \Sigma_\beta(x) \subseteq \Sigma_\beta(x) = \Sigma_\beta(s_x, |x|)$. Then, by (\ref{eq:3:construction of a multivalued function computable wrt beta that maps beta expansions to binary expansions of a real number}), 
        \begin{equation}
            \sum_{i=1}^n u_i \beta^{-i} \leq s_x = \sum_{i=1}^n x_i \beta^{-i} = \sum_{i=1}^n M_\beta(x)_i \beta^{-i}.
        \end{equation}
        Therefore, $M_\beta(x)$ maximizes the function $u \mapsto \sum_{i=1}^n u_i \beta^{-i}$ over $\hat \Sigma_\beta(x)$.
    \end{enumerate}
\end{proof}

We now proceed to establish an iterative constructive method to calculate the set $\hat \Sigma_\beta(x)$. For $\beta \in (1,2)$ and $x \in \{0,1\}^\ast$, define 
\begin{equation}\label{eq:definition of pi beta, the transition set between sigma hat beta n and sigma hat beta n 1}
    \Pi_\beta(x) := \left\{u \in \Sigma_\beta(x) : u_{1:n-1} \in \hat \Sigma_\beta(x_{1:n-1})\right\},
\end{equation}
where $n := \len x$.
\begin{theorem}
    Let $\beta \in (1,2)$, and $\xf \in \{0,1\}^\N$. Then,
    \begin{equation}
        \hat\Sigma_\beta(\xf_{1:n}) = \left\{ \lexmax (u) : u \in \Pi_\beta(\xf_{1:n}) / \sim_\beta\right\},
    \end{equation}
    for all $n \in \N$.
\end{theorem}
\begin{proof}
    Let $\beta \in (1,2)$, $\xf \in \{0,1\}^\N$, and $n \in \N$. We first prove that $\hat\Sigma_\beta(\xf_{1:n}) \subseteq \Pi_\beta(\xf_{1:n})$, and then that $\Pi_\beta(\xf_{1:n})/ \sim_\beta = \Sigma_\beta(\xf_{1:n})/ \sim_\beta$.
    \begin{enumerate}[label=(\alph*)]
        \item Let $x \in \hat\Sigma_\beta(\xf_{1:n})$, and suppose by contradiction that $x \notin \Pi_\beta(\xf_{1:n})$. Then, by definition (\ref{eq:definition of pi beta, the transition set between sigma hat beta n and sigma hat beta n 1}) of $\Pi_\beta(\xf_{1:n})$, either $x \notin \Sigma_\beta(\xf_{1:n})$ or $x_{1:n-1} \notin  \hat \Sigma_\beta(\xf_{1:n-1})$. Since $\hat\Sigma_\beta(\xf_{1:n}) \subseteq \Sigma_\beta(\xf_{1:n})$, we know that $x \in \Sigma_\beta(\xf_{1:n})$, and hence $x_{1:n-1} \notin  \hat \Sigma_\beta(\xf_{1:n-1})$. Then, there exists a $y \in [x_{1:n-1}]_\beta$ such that $y >_L x_{1:n-1}$, which further implies that $yx_n >_L x$. Moreover, since $y \in [x_{1:n-1}]_\beta$, we have 
        \begin{align}
            \sum_{i=1}^{n-1} y_i \beta^{-i} &= \sum_{i=1}^{n-1} x_i \beta^{-i}\\
            \Rightarrow \sum_{i=1}^{n-1} y_i \beta^{-i} + x_n \beta^{-n} &= \sum_{i=1}^{n} x_i \beta^{-i}.
        \end{align}
        Therefore, $yx_n \sim_\beta x$, so $yx_n \in [x]_\beta$ and $yx_n >_L x$. Hence, $x$ is not the lexicographically maximal element of $[x]_\beta$, i.e., $x \notin \hat\Sigma_\beta(\xf_{1:n})$, which is a contradiction. This concludes the proof that $\hat\Sigma_\beta(\xf_{1:n}) \subseteq \Pi_\beta(\xf_{1:n})$.
        \item Since $\hat \Sigma_\beta(\xf_{1:n}) \subseteq \Pi_\beta(\xf_{1:n}) \subseteq \Sigma_\beta(\xf_{1:n})$, then 
        \begin{equation}
            \hat \Sigma_\beta(\xf_{1:n}) / \sim_\beta \, \subseteq \Pi_\beta(\xf_{1:n}) / \sim_\beta \subseteq \Sigma_\beta(\xf_{1:n}) / \sim_\beta.
        \end{equation}
        Moreover, by definition, $\hat \Sigma_\beta(\xf_{1:n})$ consists of the elements of $\Sigma_\beta(\xf_{1:n})$ that are the lexicographically element of their equivalence class. In particular,
        \begin{equation}
            \hat \Sigma_\beta(\xf_{1:n}) / \sim_\beta \, = \Sigma_\beta(\xf_{1:n})/ \sim_\beta.
        \end{equation}
        Therefore,
        \begin{equation}
            \hat \Sigma_\beta(\xf_{1:n}) / \sim_\beta \, = \Pi_\beta(\xf_{1:n}) / \sim_\beta = \Sigma_\beta(\xf_{1:n}) / \sim_\beta.
        \end{equation}
    \end{enumerate}
    We can finally conclude the proof with
    \begin{align}
        \hat \Sigma_\beta(\xf_{1:n}) &= \left\{ \lexmax (u) : u \in \Sigma_\beta(\xf_{1:n}) / \sim_\beta\right\}\\
        & = \left\{ \lexmax (u) : u \in \Pi_\beta(\xf_{1:n}) / \sim_\beta\right\}.
    \end{align}
\end{proof}

The algorithm we will design is hence based on the following heuristic.
\begin{algorithm}[H]
    \caption{Heuristic for computing $M_\beta$}
    \label{alg:Heuristic for computing M beta}
    \begin{algorithmic}[1]
        \Require $x \in \{0,1\}^\ast$
        \State $n \gets \len x$.
        \State $\hat \Sigma_0 \gets \{\varepsilon\}$
        \For {$i = 1, \ldots, n$}
            \State $\Pi_i \gets \{ u \in \Sigma_\beta(x_{1:i-1}) : u_{1:i-1} \in \hat \Sigma_{i-1}\}$
            \State $\hat \Sigma_i \gets \{ \lexmax(u) : u \in \Pi_i/ \sim_\beta\}$
        \EndFor
    \end{algorithmic}
    \Return $\underset{u \in \hat \Sigma_n}{\arg\max} \sum_{i=1}^n u_i \beta^{-i}$
\end{algorithm}

When $\beta$ is an algebraic number, this heuristic can be turned into an effective algorithm, and further if $\beta$ is Pisot, this algorithm require only a linear number of steps.
\begin{theorem}
    Let $\beta \in (1,2)$ be an algebraic number. Then, $M_\beta$ is computable. Moreover, if $\beta$ is a Pisot number, then $M_\beta$ is computable in linear time.
\end{theorem}
\begin{proof}
    Let $\beta \in (1,2)$ be an algebraic number. Define 
    \begin{equation}
        \varepsilon := \frac{L_\beta\Pi_\beta}{\displaystyle  4n^{k_\beta}\left(\beta L_\beta \Pi_\beta^+ \right)^n}
    \end{equation}
    \begin{algorithm}[t]
        \caption{Fast algorithm for computing $M_\beta$}
        \label{alg:fast algorithm for computing M beta}
        \begin{algorithmic}[1]
            \Require $x \in \{0,1\}^\ast$
            \State $n \gets \len x$. 
            \State $s \gets \sum_{i=1}^nx_i \beta^{-i}$
            \State $\hat \Sigma_0 \gets \{\epsilon\}$
            \For {$i = 1, \ldots, n$} \label{line:4:alg:fast algorithm for computing M beta}
                \State $\Pi_i \gets \{\epsilon\}$
                \For {$u \in \hat \Sigma_{i-1}$} \label{line:6:alg:fast algorithm for computing M beta}
                    \State $s_u \gets \sum_{j=1}^{i-1}u_j \beta^{-j}$
                    \If {$s - \beta^{-i}/(\beta-1) - \varepsilon \leq s_x$}
                        \State $\Pi_i \gets \Pi_i \cup {u0}$
                    \EndIf
                    \If {$s_x + \beta^{-i} \leq s + \varepsilon$}
                        \State $\Pi_i \gets \Pi_i \cup {u1}$
                    \EndIf
                \EndFor
                \State $\hat \Sigma_i \gets \{\epsilon\}$
                \While {$\Pi_i \neq \varnothing$} \label{line:16:alg:fast algorithm for computing M beta}
                    \State Take any $v \in \Pi_i$\label{line:17:alg:fast algorithm for computing M beta}
                    \State Create the set $X_v$ of all the elements $u$ in $\Pi_i$ that satisfy \label{line:18:alg:fast algorithm for computing M beta}
                    \begin{equation}
                        \left| \sum_{i=1}^n v_i \beta^{-i} - \sum_{i=1}^n u_i \beta^{-i}\right| \leq \varepsilon
                    \end{equation}
                    \State Find the lexicographically maximal element $v^\ast$ of $X_v$. \label{line:19:alg:fast algorithm for computing M beta}
                    \State $\hat \Sigma_i \gets \hat \Sigma_i \cup \{v^\ast\}$ \label{line:20:alg:fast algorithm for computing M beta}
                    \State $\Pi_i \gets \Pi_i \backslash X_v$ \label{line:21:alg:fast algorithm for computing M beta}
                \EndWhile
            \EndFor
            \State \Return $\underset{u \in \hat \Sigma_n}{\arg\max} \sum_{i=1}^n u_i \beta^{-i}$
        \end{algorithmic}
    \end{algorithm}
    We establish Algorithm \ref{alg:fast algorithm for computing M beta} to computes $M_\beta$. We consider that each elementary operation $\gets$, $+$, $-$, $\times$, $/$, $\leq$, $\cup$ and $\backslash$ to use one time step. Let $n \in \N$, and $i \in \{1,\ldots, n \}$. Each pass of the For loop defined in line \ref{line:6:alg:fast algorithm for computing M beta} hence consumes $5$ steps. In total, this For loop requires $5 \# \hat \Sigma_i$ steps. Then, for each pass of the while loop defined in line \ref{line:16:alg:fast algorithm for computing M beta}, lines \ref{line:17:alg:fast algorithm for computing M beta}, \ref{line:20:alg:fast algorithm for computing M beta} and \ref{line:21:alg:fast algorithm for computing M beta} consume each one time step, line \ref{line:18:alg:fast algorithm for computing M beta} consumes $\# \Pi_i$ time steps, and line \ref{line:19:alg:fast algorithm for computing M beta} at most $\# \Pi_i$ time steps. Hence, since the while loop runs for at most $\# \Pi_i$ iterations, we get that the entire while loop consumes $3\#\Pi_i + \#\Pi_i^3$ steps. In total, each pass of the For loop defined in line \ref{line:4:alg:fast algorithm for computing M beta} requires at most $2 + 5 \#\hat \Sigma_{i-1} + 3 \#\Pi_i + \#\Pi_i^3$. Note that, by (\ref{eq:definition of pi beta, the transition set between sigma hat beta n and sigma hat beta n 1}), $\#\Pi_i \leq 2 \# \hat \Sigma_{i-1}$. Moreover, by the same arguments as in the proof of Lemma \ref{lem:bounding the size of the output set of the multivalued function projected on lexicographically maximal sequences}, one has 
    \begin{equation}
        \# \hat \Sigma_{i} \leq \frac{\beta^{-i}}{\beta-1} \frac{i^{k_\beta}(\beta L_\beta \Pi_\beta^+)^i}{L_\beta\Pi_\beta} = \frac{i^{k_\beta}(L_\beta \Pi_\beta^+)^i}{(\beta-1)L_\beta\Pi_\beta}\cdot
    \end{equation}
    Suppose that $\beta$ is a Pisot number, then $L_\beta = \Pi_\beta^+ = 1$, and $k_\beta =0$, so
    \begin{equation}
        \# \hat \Sigma_{i} \leq \frac{1}{(\beta-1)\Pi_\beta} =: C_\beta\cdot
    \end{equation}
    Therefore, the total number of steps required for the For loop defined in line \ref{line:4:alg:fast algorithm for computing M beta} satisfies
    \begin{align}
        t_n &\leq \sum_{i=1}^n \left(2 + 11\#\hat \Sigma_{i-1} + 8 \#\hat \Sigma_{i-1}^3\right)\\
        &\leq \sum_{i=1}^n \left(2 + 11 C_\beta + 8 C_\beta^3\right)\\
        &\leq \left(2 + 11 C_\beta + 8 C_\beta^3\right) n.
    \end{align}
\end{proof}

\endgroup

\begingroup
\let\clearpage\relax

\section{\texorpdfstring{Distribution of the complexities of $\beta$-expansions}{}}

\label{sec:distribution of the complexities of beta expansions}

In Section \ref{sec:upper bound on the complexity of beta expansions}, we established an upper bound on the algorithmic complexity of the lexicographically largest elements of the equivalence classes generated by the equivalence relationship $\sim_\beta$ that naturally appears when $\beta$ is algebraic. In this section, we further study the algorithmic complexity of the remaining elements of those equivalence classes. When $\beta$ is Pisot, this allows us to characterize precisely the distribution of the algorithmic complexities of all the $\beta$-expansions of a given real number $s \in [0,1]$, in function of the algorithmic complexity of its binary expansions. We start by building a multivalued function that is computable relatively to $\beta$, that generates the set of all $\beta$-expansions. Later, we will use this function to link the algorithmic complexity of a given $\beta$-expansion to the algorithmic complexity of the sequence of coin tosses that is effected in line 6 of Algorithm \ref{alg:random algorithm beta expansion}. 

\subsection{\texorpdfstring{A computable function to generate all the $\beta$-expansions from a fixed $\beta$-expansion}{}}

Recall that for a fixed number $\beta \in (1,2)$, we have defined the equivalence relationship $\sim_\beta$ over $\{0,1\}^\ast$ by 
\begin{equation}
    x \sim_\beta y \iff n := \len x = |y| \ \text{and} \ \sum_{i=1}^n x_i \beta^{-i} = \sum_{i=1}^n y_i \beta^{-i}, \ \forall x,y \in \{0,1\}^\ast.
\end{equation}
Moreover, we have defined the equivalence class $[x]_\beta := \{y \in \{0,1\}^\ast : x \sim_\beta y\}$, for all $x \in \{0,1\}^\ast$, which, seen as a multivalued function $[\cdot]_\beta : \{0,1\}^\ast \rightrightarrows \{0,1\}^\ast$, is proven to be computable in the proof of Corollary \ref{cor:M is computable for algebraic beta}. 

We now construct a multivalued function $f_{\beta, 1 \to all} : \{0,1\}^\ast \rightrightarrows \{0,1\}^\ast$ that is computable relatively to $\beta$, such that
\begin{equation}
    \langle \Sigma_\beta(s,n) \rangle \in f_{\beta, 1 \to all}\left(x\right), \ \forall x \in \Sigma_\beta(s,n), \ s \in I_\beta, \ \text{and} \ n \in \N.
\end{equation}
To build $f_{\beta, 1 \to all}$, we first follow the same intuitions that lead to the construction of $f_{\beta \to 2}$ in Definition \ref{def:multivalued function converting from base beta to base 2} and $f_{2 \to \beta}$ in Definition \ref{def:multivalued function converting from base 2 to base beta}. Namely, the knowledge of $x \in \Sigma_\beta(s,n)$ allows to infer that
\begin{equation}
    s \in I(x) := \left[\sum_{i=1}^{n} x_i \beta^{-i},  \sum_{i=1}^{n} x_i \beta^{-i}+ \frac{\beta^{-n}}{\beta-1} \right],
\end{equation}
and further deduce that 
\begin{equation}
    \sum_{i=1}^n y_i \beta^{-i} \in J(x) := \left[\sum_{i=1}^{n} x_i \beta^{-i} - \frac{\beta^{-n}}{\beta-1},  \sum_{i=1}^{n} x_i \beta^{-i}+ \frac{\beta^{-n}}{\beta-1} \right],
\end{equation}
for all $y \in \Sigma_\beta(s,n)$. This leads to the definition of the following function.
\begin{definition}\label{def:definition of g beta to be all the beta sequences that are inside of an interval}
    For $\beta \in (1,2)$, define $g_\beta : \{0,1\}^\ast \rightrightarrows \{0,1\}^\ast$ as 
\begin{equation}
    g_\beta(x) := \left\{ y \in \{0,1\}^n : n = \len x,  \sum_{i=1}^n y_i \beta^{-i} \in_n J(x)\right\},
\end{equation}
for all $x \in \{0,1\}^\ast$.
\end{definition}
\noindent By the preceding discussion, we have
\begin{equation}\label{eq:Sigma beta s n is included in g beta x}
    \Sigma_\beta(s,n) \subseteq g_\beta(x), \ \forall x \in \Sigma_\beta(s,n).
\end{equation}
Now, define $\tilde g_\beta :  \{0,1\}^\ast \to \{0,1\}^\ast / \sim_\beta$ by
\begin{equation}
    \tilde g_\beta(x) := g_\beta(x)/ \sim_\beta, \ \forall x \in \{0,1\}^\ast,
\end{equation}
and
\begin{equation}
    \tilde \NN_{g_\beta}(x) := \# \tilde g_\beta(x), \ \forall x \in \{0,1\}^\ast.
\end{equation}
Note that the function defined by
\begin{equation}
    \begin{array}{rccc}
        \varphi_\beta : & \{0,1\}^n / \sim_\beta & \to & I_\beta\\
        & [x]_\beta & \mapsto & \sum_{i=1}^n x_i \beta^{-i}
    \end{array}
\end{equation}
is injective into $I_\beta$, since by definiton of $\sim_\beta$, $[x]_\beta \neq [y]_\beta \iff \sum_{i=1}^n x_i \beta^{-i} \neq \sum_{i=1}^n x_i \beta^{-i}$, for all $x,y \in \{0,1\}^n$. We define an ordering $\prec$ on $\{0,1\}^n / \sim_\beta$ by 
\begin{equation}
    u \prec v \iff \varphi_\beta(u) < \varphi_\beta(v), \ \forall u,v \in \{0,1\}^n / \sim_\beta.
\end{equation}
For $x \in \{0,1\}^\ast$, we enumerate the elements of $\tilde g_\beta(x)$ as follows. For $i \in \{1,\ldots, \tilde \NN_{g_\beta}(x)\}$, we let $\tilde g_\beta(x)|_i \in \tilde g_\beta(x)$, such that
\begin{equation}
    \tilde g_\beta(x)|_i < \tilde g_\beta(x)|_{i+1}, \ \forall i \in \{1,\ldots, \tilde \NN_{g_\beta}(x) - 1\}.
\end{equation}
We can then express $\Sigma_\beta(s,n)$ simply in function of the elements of $\tilde g_\beta(x)$. 
\begin{lemma}\label{lem:the set of beta expansion is expressed as a union of consecutive equivalence classes}
    Let $\beta \in (1,2)$, $s \in [0,1]$, $n \in \N$ and $x \in \Sigma_\beta(s,n)$. Then, there exists $i,j \in \{1, \ldots, \tilde \NN_{g_\beta}(x)\}$, $i \leq j$, such that 
    \begin{equation}
        \Sigma_\beta(s,n) = \bigcup_{\ell=i}^j\tilde g_\beta(x)|_\ell.
    \end{equation}
\end{lemma}
\begin{proof}
    Let $\beta \in (1,2)$, $s \in [0,1]$, $n \in \N$, and $x \in \Sigma_\beta(s,n)$. The proof is in two parts: we first prove that $x \in \Sigma_\beta(s,n) \Rightarrow [x]_\beta \subseteq \Sigma_\beta(s,n)$, and further that $[x]_\beta \prec [y]_\beta \prec [z]_\beta$ together with $x,z \in \Sigma_\beta(s,n)$ implies $y \in \Sigma_\beta(s,n)$. 
    \begin{enumerate}[label = (\alph*)]
        \item Let $x \in \Sigma_\beta(s,n)$. By Lemma \ref{lem:sigma beta s n contains the equivalence classes}, $[x]_\beta \in \Sigma_\beta(s,n)$. Hence, $\Sigma_\beta(s,n)$ is expressed as 
        \begin{equation}
            \Sigma_\beta(s,n) = \bigcup_{u \in X} u,
        \end{equation}
        for some subset $X \subseteq \{0,1\}^n/ \sim_\beta$. Since by (\ref{eq:Sigma beta s n is included in g beta x}), $\Sigma_\beta(s,n) \subseteq g_\beta(x)$, we get that $u \subseteq g_\beta(x) \Rightarrow u \in g_\beta(x)/ \sim_\beta \ = \tilde g_\beta(x)$, so there exists $\ell \in \{1,\ldots, \tilde\NN_{g_\beta} (x)\}$ such that $u = \tilde g_\beta(x)|_\ell$, for all $u \in X$. Therefore,
        \begin{equation}
            \Sigma_\beta(s,n) = \bigcup_{\ell \in L} \tilde g_\beta(x)|_\ell,
        \end{equation}
        for some $L \subseteq \{1, \ldots, \tilde \NN_{g_\beta}(x)\}$.
        \item We now end the proof by showing that $L$ is made of consecutive elements. Let $i,j,k \in \{1, \ldots, \tilde \NN_{g_\beta}(x)\}$, such that $i \leq j \leq k$ and suppose that $\tilde g_\beta(x)|_i, \tilde g_\beta(x)|_k \subseteq \Sigma_\beta(s,n)$. Then, for all $y \in u_j(x)$, all $u \in \tilde g_\beta(x)|_i$ and all $v \in \tilde g_\beta(x)|_k$, we have 
        \begin{equation}
            \sum_{m=1}^n u_m \beta^{-m} \leq \sum_{m=1}^n y_m \beta^{-m} \leq \sum_{m=1}^n v_m \beta^{-m}.
        \end{equation}
        Since $\tilde g_\beta(x)|_i, \tilde g_\beta(x)|_k \subseteq \Sigma_\beta(s,n)$, we further get
        \begin{equation}
            s - \frac{\beta^{-n}}{\beta-1} \leq \sum_{m=1}^n u_m \beta^{-m} \leq \sum_{m=1}^n y_m \beta^{-m} \leq \sum_{m=1}^n v_m \beta^{-m} \leq s.
        \end{equation}
        By (\ref{eq:1:construction of a multivalued function computable wrt beta that maps beta expansions to binary expansions of a real number}) and (\ref{eq:2:construction of a multivalued function computable wrt beta that maps beta expansions to binary expansions of a real number}), this implies that $y \in \Sigma_\beta(s,n)$, so $\tilde g_\beta(x)|_j \subseteq \Sigma_\beta(s,n)$. We have proven that $i,k \in L \Rightarrow j \in L$ for all $i \leq j \leq k$, i.e., $L$ is made of consecutive elements.
    \end{enumerate}
\end{proof}
\noindent For every $M \in \N$, we define $\CC(M)$ to be the set of all subset of consecutive elements in $\{1,\ldots,M\}$. In mathematical symbols, this is expressed as
\begin{equation}
    \CC(M) := \{ \{i, \ldots, j\} : i,j \in \{1,\ldots,M\}, i \leq j\}.
\end{equation}
We are ready to state the definition of the computable function $f_{\beta, 1 \to all}$. 
\begin{definition}\label{def:definition of f beta 1 all}
    Let $\beta \in (1,2)$. Define $\iota_\beta(x)$ to be the unique natural number such that 
    \begin{equation}
        x \in \tilde g_\beta(x)|_{\iota_\beta(x)}.
    \end{equation}
    Let $f_{\beta, 1 \to all} : \{0,1\}^\ast \rightrightarrows \{0,1\}^\ast$ be defined as
    \begin{equation}\label{eq:main:def:definition of f beta 1 all}
        f_{\beta, 1 \to all}(x) := \left\{ \left\langle \bigcup_{i \in A} \tilde g_{\beta}(x)|_i \right\rangle : A \in \CC\left(\tilde \NN_{g_\beta}(x)\right), \iota_\beta(x) \in A\right\},
    \end{equation}
    for all $x \in \{0,1\}^\ast$.
\end{definition}
\begin{lemma}
    Let $\beta \in (1,2)$, $s \in I_\beta$ and $n \in \N$. Then,
    \begin{equation}
        \langle \Sigma_\beta(s,n) \rangle \in f_{\beta, 1 \to all}(x), \ \forall x \in \Sigma_\beta(s,n).
    \end{equation}
\end{lemma}
\begin{proof}
    Let $\beta \in (1,2)$, $s \in I_\beta$, $n \in \N$ and $x \in \Sigma_\beta(s,n)$. Then, by Lemma \ref{lem:the set of beta expansion is expressed as a union of consecutive equivalence classes}, there exists $A \in \CC(\tilde \NN_{g_\beta}(x))$ such that 
    \begin{equation}
        \Sigma_\beta(s,n) = \bigcup_{i \in A} \tilde g_\beta(x)|_i.
    \end{equation}
    The Lemma follows immediately by Definition \ref{def:definition of f beta 1 all} of $f_{\beta, 1 \to all}$.
\end{proof}
\begin{lemma}\label{lem:size of the multifunction that generates all beta expansions from one}
    Let $\beta \in (1,2)$ be an algebraic number and $n \in \N$. Then,
    \begin{equation}
        \#f_{\beta, 1 \to all}(x) \leq \left(\frac{1}{\beta-1} + 2\right)^2\frac{n^{2k_\beta}}{(L_\beta\Pi_\beta)^2}\left(L_\beta \Pi_\beta^+ \right)^{2n},
    \end{equation}
    for all $x \in \{0,1\}^n$.
\end{lemma}
\begin{proof}
    Let $\beta \in (1,2)$ be an algebraic number, $n \in \N$ and $x \in \{0,1\}^n$. By Lemma \ref{lem:generalization of the garsia separation lemma}, we get that 
    \begin{equation}
        |\varphi_\beta(u) - \varphi_\beta(v)| \geq \frac{L_\beta\Pi_\beta}{n^{k_\beta}\left(\beta L_\beta  \Pi_\beta^+ \right)^n} =: m_n,
    \end{equation}
    for all $u,v \in \{0,1\}^n / \sim_\beta$ with $u \neq v$. Then,
    \begin{align}
        \# \left(\varphi_\beta(\{0,1\}^n / \sim_\beta) \cap J(x)^{(n)}\right) &\leq \frac{|J(x)^{(n)}|}{m_n}\\
        & = \frac{2}{m_n}\left(\frac{\beta^{-n}}{\beta-1} + 2 \cdot 2^{-n}\right)\\
        & \leq \frac{2\beta^{-n}}{m_n}\left(\frac{1}{\beta-1} + 2\right) \cdot
    \end{align}
    By injectivity of $\varphi_\beta$,
    \begin{align}
        \# \varphi_\beta^{-1} \left(\varphi_\beta(\{0,1\}^n / \sim_\beta) \cap J(x)^{(n)}\right) &= \# \left(\varphi_\beta(\{0,1\}^n / \sim_\beta) \cap J(x)^{(n)}\right)\\ 
        &\leq \frac{2\beta^{-n}}{m_n}\left(\frac{1}{\beta-1} + 2\right) \cdot
    \end{align}
    Hence, since by Definition \ref{def:definition of g beta to be all the beta sequences that are inside of an interval}, $\tilde g_\beta(x) \subseteq \varphi_\beta^{-1} \left(\varphi_\beta(\{0,1\}^n / \sim_\beta) \cap J(x)^{(n)}\right)$ for $x \in \{0,1\}^n$, we get that 
    \begin{equation}
        \tilde \NN_{g_\beta}(x) \leq  \frac{2\beta^{-n}}{m_n}\left(\frac{1}{\beta-1} + 2\right) = \left(\frac{1}{\beta-1} + 2\right)\frac{2n^{k_\beta}}{L_\beta \Pi_\beta} \left(L_\beta \Pi_\beta^+ \right)^{n}.
    \end{equation}
    
    \noindent Now, there remains to establish a bound on $\#f_{\beta, 1 \to all}(x) = \#\{ A : A \in \CC(\tilde \NN_{g_\beta}(x)), \iota_\beta(x) \in A\}$, for $x \in \{0,1\}^n$. Let $M \in \N$, $i \in \{1,\ldots, i\}$, and define the set $\DD(M,i) := \{ A : A \in \CC(M), i \in A\}$. We will establish a general bound on $\#\DD(M,i)$, which as a corollary will give a bound on $\#\DD(\tilde \NN_{g_\beta}(x),\iota_\beta(x))$. By symmetry, we can restrict ourselves to $i \leq M/2$ without loss of generality. Define the sets $\DD_k(M,i) = \{ A \in \DD(M,i) : \#A = k\}$, for $k \in \{1, \ldots, M\}$. Note that 
    \begin{equation}
        \#\DD_k(M,i) = \min\{k, i, M-k\}.
    \end{equation}
    Then,
    \begin{align}
        \# \DD(M,i) &= \sum_{k = 1}^M \#\DD_k(M,i) = \sum_{k = 1}^M \min\{k, i, M-k\}\\
        &= \sum_{k=1}^{i-1} k + \sum_{k=i}^{M-i} i + \sum_{k=M - i +1}^{M} (M -k)\\
        & = 2 \sum_{k=1}^{i-1} + (M - 2i + 1)i = i(i-1) + (M - 2i + 1)i\\
        &= Mi - i^2.
    \end{align}
    Note that $Mi - i^2$ is maximal for $i = M/2$. Therefore,
    \begin{equation}
        \# \DD(M,i) \leq \frac14 M^2.
    \end{equation}
    We conclude that 
    \begin{align}
        \#f_{\beta, 1 \to all}(x) &= \#\DD(\tilde \NN_{g_\beta}(x),\iota_\beta(x)) \leq \frac14 \tilde \NN^2_{g_\beta}(x)\\
        & \leq \left(\frac{1}{\beta-1} + 2\right)^2\frac{n^{2k_\beta}}{((\beta-1)L_\beta\Pi_\beta)^2}\left(L_\beta \Pi_\beta^+ \right)^{2n},
    \end{align}
    for all $x \in \{0,1\}^n$.
\end{proof}

\subsection{\texorpdfstring{An algorithm to extract the coin tosses from a $\beta$-expansion}{}}

In this part, we show that every $\beta$-expansion of a given number $s \in I_\beta$ can be associated with a sequence of coin tosses generated by Algorithm \ref{alg:random algorithm beta expansion}. We will then study the algorithmic complexity of $\beta$-expansions through the study of this sequence of coin tosses. Following \cite[Section 1.2]{dajani2003random}, we first establish the equations that govern Algorithm \ref{alg:random algorithm beta expansion}, given fixed input $\beta \in (1,2)$, $s \in [0,1]$, and $n \in \N$. First remark that Algorithm \ref{alg:random algorithm beta expansion} generates two sequences $(b_i)_{i \in \{1,\ldots,n\}}$ and $(r_i)_{i \in \{0,\ldots,n\}}$, that are recursively defined according to the following pattern. By lines 1 and 9 in Algorithm \ref{alg:random algorithm beta expansion}, we get that 
\begin{equation}
    r_0 = s, \ \ r_i = \beta r_{i-1} - b_i, \ \forall i \in \{1, \ldots, n\}.
\end{equation}
Further, lines 4-6 deliver
\begin{equation}
    b_i = \begin{cases}
        0 \ \text{if} \ r_{i-1} \in [0, \beta^{-1}) =: E_\beta^{(0)}\\
        \text{Randomly} \ 0 \text{ or } 1 \ \text{if} \  r_{i-1}  \in [\beta^{-1}, ((\beta-1)\beta)^{-1}] =: S_\beta\\
        1 \ \text{if} \ r_{i-1} \in (((\beta-1)\beta)^{-1}, (\beta-1)^{-1}] =: E_\beta^{(1)}
    \end{cases} =: T_\beta(r_{i-1}).
\end{equation}
Note that the recursion generating $(r_i)_{i \in \{1, \ldots , n\}}$ can be simply expressed by 
\begin{equation}\label{eq:definition of K beta, simple recursive relation for r i}
    r_i = \beta r_{i-1} - T_\beta(r_{i-1}) =: K_\beta(r_{i-1}), \ \forall i \in \{1, \ldots, n\}.
\end{equation}
The output $\AM_\beta (s,n)$ of the algorithm is given in line 8 by 
\begin{equation}
    \AM_\beta (s,n) = \coprod_{i=1}^n b_i = \coprod_{i=1}^n T_\beta(r_{i-1}) =  \coprod_{i=0}^{n-1} T_\beta(r_i).
\end{equation}
In order to formalize better the problem, we can define $T_\beta(r)$ formally as a random variable on $\{0,1\}$ for $r \in I_\beta$, i.e., as a function $T_\beta(r) : \{0,1\} \to \{0,1\}$ satisfying
\begin{equation}
    T_\beta(r)(x) := \begin{cases}
        0, \ \text{if} \ s \in  E_\beta^{(0)},\\
        x, \ \text{if} \ s \in  S_\beta,\\
        1, \ \text{if} \ s \in  E_\beta^{(1)},
    \end{cases}
    \ \forall x \in \{0,1\}, \ r \in I_\beta.
\end{equation}
Since the generation of the sequence $(r_i)_{i \in \N}$ might rely on an arbitrarily many realisations of the above random variable, we define an extension $\tilde T_\beta$ of $T_\beta$, that will be useful to consider Algorithm \ref{alg:random algorithm beta expansion} as a whole to be a random variable. Namely, we define $\tilde T_\beta(r) : \{0,1\}^\N \to \{0,1\}$ by 
\begin{equation}\label{eq:definition of the generalized random bit generator}
    \tilde T_\beta(r)(\xf) := \begin{cases}
        0, \ \text{if} \ s \in  E_\beta^{(0)},\\
        \xf_1, \ \text{if} \ s \in  S_\beta,\\
        1, \ \text{if} \ s \in  E_\beta^{(1)},
    \end{cases}
    \ \forall \xf \in \{0,1\}^\N, \ r \in I_\beta.
\end{equation}
Thanks to the formulation above, we can see Algorithm \ref{alg:random algorithm beta expansion} as first generating an infinite sequence $\xf \in \{0,1\}^\N$ out of infinitely many Bernoulli experiments, and then reading successively the outcomes of this experiment, i.e., at each time step, the algorithm reads $\xf_1$, and then discards it, such that the sequence $\xf$ is transformed into the sequence $\xf' := \xf_2\xf_3\ldots$. This is usually denoted by $\xf' = \sigma(\xf)$, where $\sigma : \{0,1\}^\N \to \{0,1\}^\N$ is called the left-shift. Then, in the next time step, the algorithm will read $\xf'_1 = \xf_2$ to generate the next bit of the $\beta$-expansion of $s$.  Note that when $r_i \notin S_\beta$, then $\tilde T_\beta(r_i)(\xf)$ does not depend on $\xf \in \{0,1\}^\N$, as it is generated deterministically. Hence, the Algorithm \ref{alg:random algorithm beta expansion} does not need to discard $\xf_1$, and wait until the next time step $i$ for which $r_i \in S_\beta$ to discard $\xf_1$. Accordingly, we define $\sigma_r : \{0,1\}^\N \to \{0,1\}^\N$ by 
\begin{equation}
    \sigma_r(\xf) := \begin{cases}
        \xf \ \text{if} \ r \notin S_\beta,\\
        \sigma(\xf), \ \text{if} \ r \in S_\beta,
    \end{cases} \ \forall r \in I_\beta, \ \xf \in \{0,1\}^\N.
\end{equation}
We are now ready to state formally the definition of the Algorithm \ref{alg:random algorithm beta expansion} as a random variable. First, we model $K_\beta(r)$ defined in (\ref{eq:definition of K beta, simple recursive relation for r i}) as a random variable on $\{0,1\}^\N$ for $r \in I_\beta$, i.e., as a function $K_\beta(r) : \{0,1\}^\N \to I_\beta \times \{0,1\}^\N$,  defined by 
\begin{equation}
    K_\beta(r)(\xf) := (\beta r - \tilde T_\beta(r)(\xf_1), \sigma_r(\xf)),
    \ \forall \xf \in \{0,1\}^\N.
\end{equation}
Note that thanks to the formalism introduced, we can also consider the sequences $(b_i)_{i \in \{1, \ldots, n\}}$ and $(r_i)_{i \in \{0, \ldots, n\}}$ as random variables, i.e., each elements is a function from $\{0,1\}^\N$ to $\{0,1\}$ and $I_\beta$, respectively. We denote their respective elements $b_i(\xf)$ and $r_i(\xf)$, for $\xf \in \{0,1\}^\N$. In order to simplify the notations, we might consider $\tilde T_\beta$ as a function from $I_\beta \times \{0,1\}^\N$ into $\{0,1\}$ and $K_{\beta}$ as a function from $I_\beta \times \{0,1\}^\N$ into itself, by making the confusion between $\tilde T_\beta(s)(\xf)$ and $\tilde T_\beta(s,\xf)$, and the confusion between $K_{\beta}(s)(\xf)$ and $K_{\beta}(s,\xf)$. Note that this way, we get a simple equation for $r_i(\xf)$ in function of $\xf$:
\begin{equation}
    r_i(\xf) = \pi_1 K^i_\beta(s,\xf), \ \forall \xf \in \{0,1\}^\N,
\end{equation}
where $\pi_1$ is the projection onto the first coordinate. This simpler formalism allows to define the output of the Algorithm \ref{alg:random algorithm beta expansion} as a random variable as well, i.e., as a function $\AM_\beta(s,n) : \{0,1\}^\N \to \{0,1\}^n$, by 
\begin{equation}\label{eq:definition of the random beta expansion algorithm as a random variable}
    \AM_\beta(s,n)(\xf) = \coprod_{i=0}^{n-1} \tilde T_\beta\left(K_{\beta}^i(s,\xf)\right), \ \forall \xf \in \{0,1\}^\N.
\end{equation}

We repeat that Algorithm \ref{alg:random algorithm beta expansion} can generate different outputs of the same input $s \in [0,1]$. \cite[Theorem 1]{dajani2007invariant} establishes that infact, every $\beta$-expansion of $s$ can actually be generated by this Algorithm. This is formally expressed by 
\begin{equation}
    \AM_\beta(s,n) \left( \{0,1\}^\N\right) = \Sigma_\beta(s,n).
\end{equation}
We let also $\NN_\beta(s,n) :=  \# \left(\Sigma_\beta(s,n) \right)$.
We now leverage the formalism introduced above to express every $\beta$-expansion as a decomposition into a binary expansion and a sequence of coin tosses. We start by showing that for every $\beta$-expansion of $s$, we can associate a sequence of coin tosses that generated it. First, define
\begin{equation}
    \SM_\beta(s,\xf) := \{i \in \N_0 : \pi_1 K_\beta^i(s,\xf) \in S_\beta\},
\end{equation}
\begin{equation}
    \SM_\beta(s,n,\xf) := \SM_\beta(s,\xf) \cap \{0,\ldots,n-1\},
\end{equation} 
and
\begin{equation}
    h_\beta(s,n,\xf) := \# \SM_\beta(s,n,\xf) \leq n,
\end{equation}
for all $\xf \in \{0,1\}^\N$. We first note a simple technical Lemma.
\begin{lemma}\label{lem:recursive relation for determining h beta}
    Let $\beta \in (1,2)$, $s \in I_\beta$, $n \in \N$, and $\xf \in \{0,1\}^\N$. Then,
    \begin{equation}\label{eq:main1:lem:recursive relation for determining h beta}
        h_\beta(s,n+1,\xf) = h_\beta(s,n, \xf) + 1_{n \in \SM_\beta(s,\xf)},
    \end{equation}
    and therefore,
    \begin{equation}\label{eq:main2:lem:recursive relation for determining h beta}
        \bigcup_{i \in \SM_\beta(s,n,\xf)} \{h_\beta(s,i+1,\xf)\} = \{1, \ldots, h_\beta(s,n,\xf)\}.
    \end{equation}
\end{lemma}
\begin{proof}
    Let $\beta \in (1,2)$, $s \in I_\beta$, $n \in \N$, and $\xf \in \{0,1\}^\N$. Suppose that $n \notin S_\beta(s,\xf)$. Then,
    \begin{equation}
        \SM_\beta(s,n+1,\xf) = \SM_\beta(s,\xf) \cap \{0,\ldots, n\} = \SM_\beta(s,\xf) \cap \{0,\ldots, n-1\} = \SM_\beta(s,n,\xf),
    \end{equation}
    therefore $h_\beta(s,n+1,\xf) = h_\beta(s,n, \xf)$. Suppose now that $n \in S_\beta(s,\xf)$. Then,
    \begin{align}
        \SM_\beta(s,n+1,\xf) &= \SM_\beta(s,\xf) \cap \{0,\ldots, n\}\\
        & = \left(\SM_\beta(s,\xf) \cap \{0,\ldots, n-1\}\right) \cup \{n\}\\
        & = \SM_\beta(s,n,\xf) \cup \{n\}\,
    \end{align}
    therefore $h_\beta(s,n+1,\xf) = h_\beta(s,n, \xf) + 1$. This yields (\ref{eq:main1:lem:recursive relation for determining h beta}). Let $i_1 < i_2 < \ldots < i_{h_\beta(s,n, \xf)}$ be the elements of $\SM_\beta(s,n,\xf)$. For all $j \leq i_1$, $h_\beta(s, j, \xf) = h_\beta(s, 0, \xf) = 0$, so $h_\beta(s, i_1+1, \xf) = h_\beta(s, i_1, \xf) + 1 = 1$. Moreover, for all $k \in \{2, \ldots, h_\beta(s,n, \xf)\}$,
    \begin{equation}
        h_\beta(s, i_k+1, \xf) = h_\beta(s, i_k, \xf) + 1 = h_\beta(s, i_{k-1}+1, \xf) + 1,
    \end{equation}
    which by induction on $k$, yields
    \begin{equation}
        h_\beta(s, i_k+1, \xf) = k.
    \end{equation}
    Then,
    \begin{align}
        \bigcup_{i \in \SM_\beta(s,n,\xf)} \{h_\beta(s,i+1,\xf)\} &= \bigcup_{k = 1}^{h_\beta(s,n,\xf)} \{h_\beta(s,i_k+1,\xf)\}\\
        & = \bigcup_{k = 1}^{h_\beta(s,n,\xf)} \{k\} \\
        &= \{1, \ldots, h_\beta(s,n,\xf)\}.
    \end{align}
\end{proof}
\noindent $h_\beta(s,n-1,\xf)$ is exactly the number of bits of $\xf$ that Algorithm \ref{alg:random algorithm beta expansion} had to read before delivering its output. Then, every bit of $\xf$ that is localed after the rank $h_\beta(s,n,\xf)$ is not read by the algorithm to produce its output. This is formally expressed in the following Lemma and the subsequent Corollary.
\begin{lemma}\label{lem:the random beta expansion algorithm reads only the first bits of its input}
    Let $\beta \in (1,2)$, $s \in I_\beta$, $n \in \N$. Then,
    \begin{equation}
        \begin{cases}
            \SM_\beta(s,n,\xf) = \SM_\beta(s,n,\xf_{1:h_\beta(s,n,\xf)}\yf),\\
            \displaystyle \pi_1 K^n_\beta(s,\xf) = \pi_1 K^n_\beta\left(s,\xf_{1:h_\beta(s,n,\xf)}\yf\right),\\
            \displaystyle \pi_2 K^n(s,\xf) = \sigma^{h_\beta(s,n,\xf)}(\xf),
        \end{cases}
    \end{equation}
    for all $\xf, \yf \in \{0,1\}^\N$, where $\pi_1$ (resp. $\pi_2$) is the projection on the first coordinate (resp. second coordinate).
\end{lemma}
\begin{proof}
    Let $\beta \in (1,2)$, $s \in I_\beta$. We prove the result by induction. Obviously,
    \begin{equation}
        \SM_\beta(s,0,\xf) = \varnothing = \SM_\beta(s,0,\xf_{1:h_\beta(s,0,\xf)}\yf),
    \end{equation}
    \begin{equation}
        \pi_1 K^0_\beta(s,\xf) = s =  \pi_1 K^0_\beta(s,\xf_{1:h_\beta(s,0,\xf)}\yf),
    \end{equation}
    and
    \begin{equation}
        \pi_2 K^i(s,\xf) = \xf = \sigma^{0}(\xf) = \sigma^{h_\beta(s,0,\xf)}(\xf),
    \end{equation}
    for all $\xf,\yf \in \{0,1\}^\N$.
    Now, suppose that there exists $n \in \N$ such that
    \begin{equation}\tag{$H_n$}
        \begin{cases}
            \SM_\beta(s,n,\xf) = \SM_\beta(s,n,\xf_{1:h_\beta(s,n,\xf)}\yf) =: \SM_n,\\
            \displaystyle \pi_1 K^n_\beta(s,\xf) = \pi_1 K^n_\beta\left(s,\xf_{1:h_\beta(s,n,\xf)}\yf\right) =: r_n,\\
            \displaystyle \pi_2 K^n(s,\xf) = \sigma^{h_\beta(s,n,\xf)}(\xf) = \sigma^{h_n}(\xf),
        \end{cases}
    \end{equation}
    where $h_n := \# \SM_n$, for all $\xf, \yf \in \{0,1\}^\N$. We will show that $(H_{n+1})$ holds for all $\xf, \yf \in \{0,1\}^\N$. Fix $\xf, \yf \in \{0,1\}^\N$. First, note that 
    \begin{align}
        K^{n+1}_\beta(s,\xf) &= K_\beta \left( K^n_\beta(s,\xf)\right)\\
        &= K_\beta\left( \pi_1 K^n_\beta(s,\xf), \pi_2 K^n_\beta(s,\xf)\right)\\
        &\overset{(a)}{=} K_\beta\left( r_n, \ \sigma^{h_n}(\xf)\right),
    \end{align}
    and similarly
    \begin{align}
        K^{n+1}_\beta(s,\xf_{1:h_\beta(s,n+1,\xf)}\yf) &= K_\beta \left( K^n_\beta(s,\xf_{1:h_n}\yf)\right)\\
        &= K_\beta\left( \pi_1 K^n_\beta(s,\xf_{1:h_n}\yf), \pi_2 K^n_\beta(s,\xf_{1:h_n}\yf)\right)\\
        &\overset{(b)}{=} K_\beta\left(r_n, \ \sigma^{h_\beta(s,n,\xf_{1:h_n}\yf)}(\xf_{1:h_n}\yf)\right),\\
        &\overset{(c)}{=} K_\beta\left(r_n, \ \sigma^{h_n}(\xf_{1:h_n}\yf)\right),
    \end{align}
    where (a), (b) and (c) follow from $(H_n)$.
    \begin{enumerate}[label = (\alph*)]
        \item Suppose that $r_n \in E_\beta^{(0)} \cup E_\beta^{(1)}$, and define $b \in \{0,1\}$ such that $r_n \in E_\beta^{(b)}$. Then, 
            \begin{enumerate}[label=(\roman*)]
                \item Since $r_n \notin S_\beta$, then $n \notin \SM_\beta(s,\xf)$, so
                \begin{align}
                    \SM_\beta(s,n+1,\xf) &= \SM_\beta(s,\xf) \cap \{0, \ldots, n\}\\
                    & =  \SM_\beta(s,\xf) \cap \{0, \ldots, n-1\}\\
                    & = \SM_\beta(s,n,\xf) = \SM_n,
                \end{align}
                and similarly, $n \notin \SM_\beta(s,\xf_{1:h_n}\yf)$, so
                \begin{align}
                    \SM_\beta(s,n+1,\xf_{1:h_n}\yf) &= \SM_\beta(s,\xf_{1:h_n}\yf) \cap \{0, \ldots, n\}\\
                    & =  \SM_\beta(s,\xf_{1:h_n}\yf) \cap \{0, \ldots, n-1\}\\
                    & = \SM_\beta(s,n,\xf_{1:h_n}\yf) = \SM_n,
                \end{align}
                which yields
                \begin{equation}
                    \SM_\beta(s,n+1,\xf) = \SM_\beta(s,n+1,\xf_{1:h_\beta(s,n,\xf)}\yf) =: \SM_{n+1},
                \end{equation}
                and note that $\SM_{n+1} = \SM_n$, and hence $h_{n+1} = h_n$.
                \item Note that
                \begin{equation}
                    \pi_1 K_\beta^{n+1}(s,\xf) = \pi_1 K_\beta(r_n, \sigma^{h_n}(\xf)) = \beta r_n - b,
                \end{equation}
                and 
                \begin{align}
                    \pi_1 K_\beta^{n+1}(s,\xf_{1:h_n}\yf) = \pi_1 K_\beta(r_n, \sigma^{h_n}(\xf_{1:h_n}\yf)) = \beta r_n - b,
                \end{align}
                so 
                \begin{equation}
                    r_{n+1} := \pi_1 K_\beta^{n+1}(s,\xf) = \pi_1 K_\beta^{n+1}(s,\xf_{1:h_n}\yf),
                \end{equation}
                and note that $r_{n+1} = \beta r_n - b$.
                \item Since $r_n \notin S_\beta$, then 
                \begin{equation}
                    \pi_2 K_\beta^{n+1}(s,\xf) = \pi_2 K_\beta(r_n, \sigma^{h_n}(\xf)) = \sigma^{h_n}(\xf) = \sigma^{h_{n+1}}(\xf).
                \end{equation}
            \end{enumerate}
                \item Suppose that $r_n \in S_\beta$. Then,
                \begin{enumerate}[label=(\roman*)]
                    \item Since $r_n \in S_\beta$, then $n \in \SM_\beta(s,\xf)$, so
                    \begin{align}
                        \SM_\beta(s,n+1,\xf) &= \SM_\beta(s,\xf) \cap \{0, \ldots, n\}\\
                        & =  \left(\SM_\beta(s,\xf) \cap \{0, \ldots, n-1\}\right) \cup \{n\}\\
                        & = \SM_\beta(s,n,\xf) \cup \{n\} = \SM_n \cup \{n\},
                    \end{align}
                    and similarly, $n \notin \SM_\beta(s,\xf_{1:h_n}\yf)$, so
                    \begin{align}
                        \SM_\beta(s,n+1,\xf_{1:h_n}\yf) &= \SM_\beta(s,\xf_{1:h_n}\yf) \cap \{0, \ldots, n\}\\
                        & =  \left(\SM_\beta(s,\xf_{1:h_n}\yf) \cap \{0, \ldots, n-1\} \right) \cup \{n\}\\
                        & = \SM_\beta(s,n,\xf_{1:h_n}\yf) \cup \{n\} = \SM_n \cup \{n\},
                    \end{align}
                    which yields
                    \begin{equation}
                        \SM_\beta(s,n+1,\xf) = \SM_\beta(s,n+1,\xf_{1:h_n}\yf) =: \SM_{n+1},
                    \end{equation}
                    and note that $\SM_{n+1} = \SM_n \cup \{n\}$, and hence $h_{n+1} = h_n + 1$.
                    \item Note that
                    \begin{align}
                        \pi_1 K_\beta^{n+1}(s,\xf) &= \pi_1 K_\beta(r_n, \sigma^{h_n}(\xf))\\
                        & = \beta r_n - \left[\sigma^{h_n}(\xf)\right]_1\\
                        & = \beta r_n - \xf_{h_n+1},
                    \end{align}
                    and 
                    \begin{align}
                        \pi_1 K_\beta^{n+1}(s,\xf_{1:h_{n+1}}\yf) &= \pi_1 K_\beta(r_n, \sigma^{h_n}(\xf_{1:h_{n+1}}\yf))\\
                        & = \beta r_n - \left[\sigma^{h_n}(\xf_{1:h_{n+1}}\yf)\right]_1\\
                        & = \beta r_n - \left[\sigma^{h_n}(\xf_{1:h_n+1}\yf)\right]_1\\
                        & = \beta r_n - \xf_{h_n+1},
                    \end{align}
                    so 
                    \begin{equation}
                        r_{n+1} := \pi_1 K_\beta^{n+1}(s,\xf) = \pi_1 K_\beta^{n+1}(s,\xf_{1:h_{n+1}}\yf),
                    \end{equation}
                    and note that $r_{n+1} = \beta r_n - \xf_{h_n+1}$.
                    \item Since $r_n \in S_\beta$, then 
                    \begin{align}
                        \pi_2 K_\beta^{n+1}(s,\xf) &= \pi_2 K_\beta(r_n, \sigma^{h_n}(\xf)) = \sigma\left(\sigma^{h_{n}}(\xf)\right)\\
                        & = \sigma^{h_{n}+1}(\xf) = \sigma^{h_{n+1}}(\xf).
                    \end{align}
                \end{enumerate}
    \end{enumerate}
    We hence have proven that $(H_n) \Rightarrow (H_{n+1})$. We have completed the induction, thereby finishing the proof.
\end{proof}

\begin{corollary}Let $\beta \in (1,2)$, $s \in I_\beta$, $n \in \N$. Then,
    \begin{equation}\label{eq:dajani algorithm only depends on the bits that were read by the algorithm}
        \AM_\beta(s,n)(\xf_{1:h_\beta(s,n,\xf)}\{0,1\}^\N) = \{\AM_\beta(s,n)(\xf)\}
    \end{equation}
    for all $\xf \in \{0,1\}^\N$.
\end{corollary}
\begin{proof}
    Let $\beta \in (1,2)$, $s \in I_\beta$, $n \in \N$, $\xf\in \{0,1\}^\N$ and $\zf \in \xf_{1:h_\beta(s,i,\xf)} \{0,1\}^\N$, Note that there exists $\yf \in \{0,1\}^\N$ such that $\zf := \xf_{1:h_\beta(s,i,\xf)}\yf$. Fix $i \in \{0, \ldots, n-1\}$. By Lemma \ref{lem:the random beta expansion algorithm reads only the first bits of its input}, we have 
    \begin{equation}
        \begin{cases}
            h_\beta(s,i,\xf) = h_\beta(s,i,\zf) =: h_i,\\
            \displaystyle \pi_1 K^n_\beta(s,\xf) = \pi_1 K^n_\beta\left(s,\zf\right) =: r_i,\\
            \displaystyle \pi_2 K^i(s,\xf) = \sigma^{h_\beta(s,i,\xf)}(\xf) = \sigma^{h_i}(\xf),\\
            \displaystyle \pi_2 K^i(s,\zf) = \sigma^{h_\beta(s,i,\zf)}(\zf) = \sigma^{h_i}(\zf),
        \end{cases}
    \end{equation}
    Then,
    \allowdisplaybreaks
    \begin{align}
        \tilde T_\beta \left( K_\beta^i(s,\xf) \right) &= \tilde T_\beta \left(  \pi_1 K_\beta^i(s,\xf), \pi_2 K_\beta^i(s,\xf)\right) = \tilde T_\beta \left(  r_i, \sigma^{h_i}(\xf)\right)\\
        &= \nonumber \begin{cases}
            0 \text{ if } r_i \in E_\beta^{(0)}\\
            \left[\sigma^{h_i}(\xf)\right]_1 \text{ if } r_i \in S_\beta\\
            1 \text{ if } r_i \in E_\beta^{(1)}
        \end{cases} = \begin{cases}
            0 \text{ if } r_i \in E_\beta^{(0)}\\
            \xf_{h_i + 1} \text{ if } r_i \in S_\beta\\
            1 \text{ if } r_i \in E_\beta^{(1)}
        \end{cases} \overset{(a)}{=} \begin{cases}
            0 \text{ if } r_i \in E_\beta^{(0)}\\
            \xf_{h_{i+1}}\text{ if } r_i \in S_\beta\\
            1 \text{ if } r_i \in E_\beta^{(1)}
        \end{cases}
    \end{align}
    where (a) follows from (\ref{eq:main1:lem:recursive relation for determining h beta}). Note that $h_{i+1} = h_\beta(s,i+1,\xf) \leq h_\beta(s,n,\xf)$, so $\xf_{h_{i+1}} = \zf_{h_{i+1}}$. Therefore,
    \begin{align}
        \tilde T_\beta \left( K_\beta^i(s,\xf) \right) = \begin{cases}
        0 \text{ if } r_i \in E_\beta^{(0)}\\
        \zf_{h_{i+1}}\text{ if } r_i \in S_\beta\\
        1 \text{ if } r_i \in E_\beta^{(1)}
    \end{cases} = \tilde T_\beta \left(  r_i, \sigma^{h_i}(\zf)\right) = \tilde T_\beta \left( K_\beta^i(s,\zf) \right).
    \end{align}
    We conclude the proof by noting that 
    \begin{equation}
        \AM_\beta(s,n)(\xf) = \coprod_{i=0}^{n-1} \tilde T_\beta \left( K_\beta^i(s,\xf) \right) = \coprod_{i=0}^{n-1} \tilde T_\beta \left( K_\beta^i(s,\zf) \right) = \AM_\beta(s,n)(\zf),
    \end{equation}
    for all $\zf \in \xf_{1:h_\beta(s,n,\xf)}\{0,1\}^\N$.
\end{proof}
We now show that for every prefix $x$ of some $\beta$-expansion of $s$, is generated by Algorithm \ref{alg:random algorithm beta expansion} through a unique sequence of coin tosses denoted $w_\beta(s,x)$.
\begin{lemma}\label{lem:the set of preimages by the dajani algorithm is a cylinder}
    Let $\beta \in (1,2)$, $s \in I_\beta$ and $n \in \N$. Then, for every $x \in \Sigma_\beta(s,n)$, there exists a unique $w_\beta(s, x) \in \{0,1\}^{\leq n}$, such that 
    \begin{equation}
        \AM_\beta(s,n)^{-1}(x) = w_\beta(s, x) \{0,1\}^\N.
    \end{equation}
    Moreover, 
    \begin{equation}
        \xf_{1:h_\beta(s,n,\xf)} = w_\beta(s,x), \ \forall \xf \in \AM_\beta(s,n)^{-1}(x).
    \end{equation}
\end{lemma}
\begin{proof}
    Let $\beta \in (1,2)$, $s \in I_\beta$, $n \in \N$ and $x \in \Sigma_\beta(s,n)$. We first prove that $\AM_\beta(s,n)^{-1}(x) = A_x \{0,1\}^\N$, where $A_x \subseteq \{0,1\}^{\leq n}$, and then that $A_x$ consists of a unique element. 
    \begin{enumerate}[label = (\alph*)]
        \item Let $\xf \in \AM_\beta(s,n)^{-1}(x)$. Then, by (\ref{eq:dajani algorithm only depends on the bits that were read by the algorithm}), every $\yf \in \xf_{1:h_\beta(s,n,\xf)} \{0,1\}^\N$ satisfies $\yf \in \AM_\beta(s,n)^{-1}(x)$. It follows that
        \begin{equation}\label{eq:1:proof:lem:the set of preimages by the dajani algorithm is a cylinder}
            \AM_\beta(s,n)^{-1}(x) = \bigcup_{\xf \in \AM_\beta(s,n)^{-1}(x)} \xf_{1:h_\beta(s,n,\xf)} \{0,1\}^\N.
        \end{equation}
        By defining $A_x := \bigcup_{\xf \in \AM_\beta(s,n)^{-1}(x)} \{\xf_{1:h_\beta(s,n,\xf)}\}$, we then get 
        \begin{equation}
            \AM_\beta(s,n)^{-1}(x) = A_x \{0,1\}^\N.
        \end{equation}
        Since $h_\beta(s,n,\xf) \leq n$ for all $\xf \in \{0,1\}^\N$, we get that $A_x \subseteq \{0,1\}^{\leq n}$.
        \item Let $u, v \in A_x$, and assume that $\len u \leq \len v$. We will show by contradiction that $u \sqsubset v$. By definition of $A_x$, there exists $\uf,\vf \in \AM_\beta(s,n)^{-1}(x)$ such that $\uf_{1:h_\beta(s,n,\uf)} = u$ and $\vf_{1:h_\beta(s,n,\vf)} = v$. By sake of contradiction, assume that $u \not\sqsubset v$, i.e. there exists $i \in \{1,\ldots, \len u\}$, such that $u_{1:i-1} = v_{1:i-1}$, and $u_i \neq v_i$. Denote by $n_i$ the $i$-th element of $\SM_\beta(s,\uf)$ in increasing order. Then,
        \begin{equation}
            h_\beta(s,n_i,\uf) = \#\SM_\beta(s,n_i,\uf) = \# \left(\SM_\beta(s,\uf) \cap \{0, \ldots, n_i-1\} \right) = i-1,
        \end{equation}
        since by definition there are exactly $i$ elements in $\SM_\beta(s,\uf)$ that are not larger than $n_i$. Further note that
        \begin{equation}
            \uf_{1:h_\beta(s,n_i,\uf)} = \uf_{1:i-1} = u_{1:i-1} = v_{1:i-1} = \vf_{1:i-1} = \vf_{1:h_\beta(s,n_i,\uf)}.
        \end{equation}
        Therefore, $\vf$ is such that there exists $\yf \in \{0,1\}^\N$ satisfying $\vf = \uf_{1:h_\beta(s,n_i,\uf)} \yf$, so we can apply Lemma \ref{lem:the random beta expansion algorithm reads only the first bits of its input}. In particular,
        \begin{align}
            i-1 \overset{(a)}{=} h_\beta(s,n_i,\uf) &= h_\beta(s,n_i,\vf)\\
            r_i := \pi_1 K_\beta^{n_i}(s,\uf) &= \pi_1 K_\beta^{n_i}(s,\vf) \in S_\beta\\
            \pi_2 K_\beta^{n_i}(s,\uf) &= \sigma^{i-1}(\uf)\\
            \pi_2 K_\beta^{n_i}(s,\vf) &= \sigma^{i-1}(\vf).
        \end{align}
        where (a) follows by $n_i \in \SM_\beta(s,n_i,\uf)$. This further implies that 
        \begin{equation}
            \left[ \AM_\beta(s,n)(\uf)\right]_{n_i} = \tilde T_\beta\left( K_\beta^{n_i}(s,\uf)\right) = \tilde T_\beta \left( r_i, \sigma^{i-1}(\uf) \right) = \left[\sigma^{i-1}(\uf) \right]_i = \uf_i,
        \end{equation}
        and 
        \begin{equation}
            \left[ \AM_\beta(s,n)(\vf)\right]_{n_i} = \tilde T_\beta\left( K_\beta^{n_i}(s,\vf)\right) = \tilde T_\beta \left( r_i, \sigma^{i-1}(\vf) \right) = \left[\sigma^{i-1}(\vf) \right]_i = \vf_i.
        \end{equation}
        Since $\uf_i \neq \vf_i$, we get $\AM_\beta(s,n)(\uf) \neq \AM_\beta(s,n)(\vf)$, which is a contraction, since $\uf, \vf \in \AM_\beta(s,n)^{-1}(x)$. Therefore, $u \sqsubset v$.
        \item Since $A_x$ is a finite set, there exists an element of minimal length $u^\ast \in A_x$. Then, for every $v \in A_x$, $\len u^\ast \leq \len v$, so by the previous result, $u^\ast \sqsubset v$, which further implies $v \{0,1\}^\N \subseteq u^\ast \{0,1\}^\N$. This yields
        \begin{equation}
            \AM_\beta(s,n)^{-1}(x) = A_x \{0,1\}^\N = \bigcup_{v \in A_x} v \{0,1\}^\N = u^\ast \{0,1\}^\N.
        \end{equation}
        We define $w_\beta(s,x) := u^\ast$. Note that by definition of $A_x$, there exists $\xf^\ast \in \AM_\beta(s,n)^{-1}(x)$ such that $\xf^\ast_{1:h_\beta(s,n,\xf^\ast)} = w_\beta(s,x)$. Moreover, for every $\xf \in \AM_\beta(s,n)^{-1}(x)$, $\xf_{1:h_\beta(s,n,\xf)} \in A_x$. By the preceding point ,we get that $\xf^\ast_{1:h_\beta(s,n,\xf^\ast)} = w_\beta(s,x) \sqsubset \xf_{1:h_\beta(s,n,\xf)}$, i.e., there exists $\yf \in \{0,1\}^\N$ such that $\xf = \xf^\ast_{1:h_\beta(s,n,\xf^\ast)} \yf$. By Lemma \ref{lem:the random beta expansion algorithm reads only the first bits of its input},
        \begin{equation}
            h_\beta(s,n,\xf) = h_\beta(s,n,\xf^\ast) = |w_\beta(s,x)|.
        \end{equation}
        Finally, since $\AM_\beta(s,n)^{-1}(x) = w_\beta(s,x) \{0,1\}^\N$, then 
        \begin{equation}
            \xf_{1:h_\beta(s,n,\xf)} = \xf_{1:|w_\beta(s,x)|} = w_\beta(s,x).
        \end{equation}
        This concludes the proof.
    \end{enumerate}
\end{proof}

\noindent We now show how to infer this sequence of coin tosses from $x$ and $s$. We define
\begin{equation}
    \tilde \SM_\beta(s,x) = \bigcap_{\xf \in \AM_\beta(s,\len x)^{-1}(x)} \SM_\beta(s,\len x,\xf),
\end{equation}
and 
\begin{equation}
    \tilde h_\beta(s,x) := \# \tilde \SM_\beta(s,x).
\end{equation}
\begin{lemma}\label{lem:every preimage of a beta expansion throught the dajani algorithm yields the occurences in the switch set}
    Let $\beta \in (1,2)$, $s \in I_\beta$, and $x \in \{0,1\}^\ast$. Then, for all $\xf,\yf \in \AM_\beta(s,n)^{-1}(x)$,
    \begin{equation}
        \SM_\beta(s,\len x,\xf) = \SM_\beta(s,\len x,\yf).
    \end{equation}
    Consequently,
    \begin{equation}
        \tilde \SM_\beta(s,x) = \SM_\beta(s,\len x,\xf),
    \end{equation}
    for all $\xf \in  \AM_\beta(s,\len x)^{-1}(x)$.
\end{lemma}
\begin{proof}
    Let $\beta \in (1,2)$, $s \in I_\beta$, $x \in \{0,1\}^\ast$ and $\xf,\yf \in \AM_\beta(s,n)^{-1}(x)$. Then, by Lemma \ref{lem:the set of preimages by the dajani algorithm is a cylinder}, $\xf_{1:h_\beta(s,|x|,\xf)} = \yf_{1:h_\beta(s,|x|,\yf)} = w_\beta(s,x)$. Hence, there exists $\zf \in \{0,1\}^\N$ such that $\yf = w_\beta(s,x) \zf = \xf_{1:h_\beta(s,|x|,\xf)} \zf$. Then, by Lemma \ref{lem:the random beta expansion algorithm reads only the first bits of its input},
    \begin{equation}
        \SM_\beta(s,\len x,\xf) = \SM_\beta(s,\len x,\yf).
    \end{equation}
\end{proof}
The set $\tilde \SM_\beta(s,x)$ can be exploited to infer $w_\beta(s,x)$.
\begin{lemma}\label{lem:the noise can be expressed from the time steps for which the trajectory is in switch set}
    Let $\beta \in (1,2)$, $s \in I_\beta$ and $x \in \{0,1\}^\ast$. 
    \begin{equation}
        w_\beta(s,x) = \coprod \left\{ x_{i+1} : i \in \tilde \SM_\beta(s,x)\right\}.
    \end{equation}
\end{lemma}

\begin{proof}
    Let $\beta \in (1,2)$, $s \in I_\beta$, $x \in \{0,1\}^\ast$ and $\xf \in \AM_\beta(s,n)^{-1}(x)$. Then, 
    \begin{align}
        \left\{ x_{i+1} : i \in \tilde \SM_\beta(s,x)\right\} &= \left\{ x_{i+1} : i \in  \SM_\beta(s,|x|,\xf)\right\}\\
        &= \left\{\left[\AM_\beta(s,n)(\xf)\right]_{i+1} : i \in  \SM_\beta(s,|x|,\xf) \right\}\\
        &\oseteq{\ref{eq:definition of the random beta expansion algorithm as a random variable}}{=} \left\{\tilde T_\beta\left(K_\beta^i(s,\xf) \right) : i \in  \SM_\beta(s,|x|,\xf) \right\}\\
        &=\left\{\tilde T_\beta\left(\pi_1K_\beta^i(s,\xf), \pi_2K_\beta^i(s,\xf) \right) : i \in  \SM_\beta(s,|x|,\xf) \right\}\\
        &\overset{(a)}{=} \left\{\left[\pi_2K_\beta^i(s,\xf)\right]_1 : i \in  \SM_\beta(s,|x|,\xf) \right\}\\
        &\osetlem{\ref{lem:the random beta expansion algorithm reads only the first bits of its input}}{=} \left\{\left[\sigma^{h_\beta(s,i,\xf)}(\xf)\right]_1 : i \in  \SM_\beta(s,|x|,\xf) \right\}\\
        &= \left\{\xf_{h_\beta(s,i,\xf)+1} : i \in  \SM_\beta(s,|x|,\xf) \right\}\\
        &\oseteq{\ref{eq:main1:lem:recursive relation for determining h beta}}{=} \left\{\xf_{h_\beta(s,i+1,\xf)} : i \in  \SM_\beta(s,|x|,\xf) \right\}\\
        &\oseteq{\ref{eq:main2:lem:recursive relation for determining h beta}}{=} \left\{\xf_j : j \in \{1,\ldots, h_\beta(s,|x|,\xf)\} \right\},
    \end{align}
    where (a) follows from $i \in \SM_\beta(s,|x|,\xf) \Rightarrow \pi_1K_\beta^i(s,\xf) \in S_\beta$, together with (\ref{eq:definition of the generalized random bit generator}).Hence,
    \begin{equation}
        \coprod \left\{ x_{i+1} : i \in \tilde \SM_\beta(s,|x|,\xf)\right\} = \coprod \left\{\xf_j : j \in \{1,\ldots, h_\beta(s,|x|,\xf)\} \right\} = \xf_{1:h_\beta(s,|x|,\xf)}.
    \end{equation}
    By Lemma \ref{lem:the set of preimages by the dajani algorithm is a cylinder}, $\xf_{1:h_\beta(s,|x|,\xf)} = w_\beta(s,x)$, which concludes the proof.
    
    

\end{proof}

We now show how to derive $\tilde \SM_\beta(s,x)$ from $\Sigma_\beta(s,\len x)$ and $x$.
\begin{lemma}\label{lem:characterization of the set of indices in the switch set by divergences in the tree of beta expansions}
    Let $\beta \in (1,2)$, $s \in I_\beta$ and $x \in \{0,1\}^\ast$. Then,
    \begin{equation}
        \tilde \SM_\beta(s,x) = \left\{ i \in \{1,\ldots, \len x\}: \exists y \in \Sigma_\beta(s,\len x) \ s.t. \ y_{1:i-1} = x_{1:i-1}, y_i \neq x_i \right\}.
    \end{equation}
\end{lemma}
\begin{proof}
    Let $\beta \in (1,2)$, $s \in I_\beta$ and $x \in \{0,1\}^\ast$. Let 
    \begin{equation}
        A(s,x) := \left\{ i \in \{0,\ldots, \len x -1\}: \exists y \in \Sigma_\beta(s,\len x) \ s.t. \ y_{1:i} = x_{1:i}, y_{i+1} \neq x_{i+1} \right\}.
    \end{equation}
    We prove that $\tilde \SM_\beta(s,x) = A(s,x)$ by double inclusion.
    \begin{enumerate}[label = (\alph*)]
        \item We show that $\tilde \SM_\beta(s,x) \subseteq  A(s,x)$. Let $j \in \tilde \SM_\beta(s,x)$ and $\xf \in \AM_\beta(s,\len x)^{-1}(x)$. By Lemma \ref{lem:every preimage of a beta expansion throught the dajani algorithm yields the occurences in the switch set}, $\SM_\beta(s,\len x, \xf) = \tilde \SM_\beta(s,x)$, so $j \in \SM_\beta(s,\len x, \xf)$. Let $\yf \in \{0,1\}^\N$, such that $\yf_{1:h_\beta(s,j,\xf)} = \xf_{1:h_\beta(s,j,\xf)}$ and $\yf_{1:h_\beta(s,j,\xf)+1} \neq \xf_{1:h_\beta(s,j,\xf)+1}$, and define $y := \AM_\beta(s,\len x)(\yf)$. Note that $y \in \Sigma_\beta(s,|x|)$. We will prove that such a $y$ satisfies the condition in the definition of $A(s,x)$. By Lemma \ref{lem:the random beta expansion algorithm reads only the first bits of its input}, 
        \begin{align}
            \SM_\beta(s,i, \xf) &= \SM_\beta(s,i, \yf) =: \SM_i\\
            h_\beta(s,i,\xf) &= h_\beta(s,i,\yf) =: h_i\\
            \pi_1 K_\beta^i(s,\xf) &= \pi_1 K_\beta^i(s,\yf) =: r_i\\
            \pi_2 K_\beta^i(s,\xf) &= \sigma^{h_\beta(s,i,\xf)}(\xf) = \sigma^{h_i}(\xf)\\
            \pi_2 K_\beta^i(s,\yf) &= \sigma^{h_\beta(s,i,\yf)}(\yf) = \sigma^{h_i}(\yf),
        \end{align}
        for all $i \leq j$. Then,we have 
        \begin{align}
            x_{i+1} = \left[\AM_\beta(s,\len x)(\xf)\right]_{i+1} &= \tilde T_\beta \left(K_\beta^i(s,\xf) \right)\\
            &= \tilde T_\beta \left(r_i, \sigma^{h_i}(\xf)\right)\\
            &= \begin{cases}
                b \ \text{if} \ r_i \in E_\beta^{(b)}, \ b \in \{0,1\},\\
                \left[\sigma^{h_i}(\xf)\right]_1 \ \text{if} \ r_i \in S_\beta
            \end{cases}\\
            &= \label{eq:0:proof:lem:characterization of the set of indices in the switch set by divergences in the tree of beta expansions} \begin{cases}
                b \ \text{if} \ r_i \in E_\beta^{(b)}, \ b \in \{0,1\},\\
                \xf_{h_i+1} \ \text{if} \ r_i \in S_\beta,
            \end{cases}
        \end{align}
        for all $i \leq j$. By similar calculations, we obtain 
        \begin{equation}
            y_{i+1} = \begin{cases}
                b \ \text{if} \ r_i \in E_\beta^{(b)}, \ b \in \{0,1\},\\
                \yf_{h_i+1} \ \text{if} \ r_i \in S_\beta,
            \end{cases}
        \end{equation}
        for all $i \leq j$. Then, for $i \in \{0, \ldots, j-1\}$ such that $i \notin \SM_j$, we get 
        \begin{equation}\label{eq:1:proof:lem:characterization of the set of indices in the switch set by divergences in the tree of beta expansions}
            x_{i+1} = y_{i+1}.
        \end{equation}
        Moreover, for $i \in \{0, \ldots, j\}$ such that $i \in \SM_j$, we have
        \begin{equation}\label{eq:2:proof:lem:characterization of the set of indices in the switch set by divergences in the tree of beta expansions}
            x_{i+1} = \xf_{h_i+1} = \xf_{h_{i+1}}
        \end{equation}
        and 
        \begin{equation}\label{eq:3:proof:lem:characterization of the set of indices in the switch set by divergences in the tree of beta expansions}
            y_{i+1} = \yf_{h_i+1} = \yf_{h_{i+1}}.
        \end{equation}
        Since by assumption, $\xf_{h_{i+1}} = \yf_{h_{i+1}}$ for all $i < j$, and $\xf_{h_{j}+1} \neq \yf_{h_{j}+1}$, we get by combining (\ref{eq:1:proof:lem:characterization of the set of indices in the switch set by divergences in the tree of beta expansions}), (\ref{eq:2:proof:lem:characterization of the set of indices in the switch set by divergences in the tree of beta expansions}) and (\ref{eq:3:proof:lem:characterization of the set of indices in the switch set by divergences in the tree of beta expansions}) that
        \begin{equation}
            x_{1:j} = y_{1:j}, \ \text{and} \ x_{j+1} \neq y_{j+1}.
        \end{equation}
        Then, $x \in A(s,x)$. This concludes the proof that $\tilde \SM_\beta(s,x) \subseteq  A(s,x)$.
        \item We show that $\tilde \SM_\beta(s,x) \supseteq  A(s,x)$. Let $j \in A(s,x)$. Then, there exists $y \in \Sigma_\beta(s,\len x)$ such that 
        \begin{equation}
            x_{1:j} = y_{1:j}, \ \text{and} \ x_{j+1} \neq y_{j+1}.
        \end{equation}
        We prove that $j \in \tilde \SM_\beta(s,x)$ which is equivalent to prove that $j \in \SM_\beta(s,\len x, \xf)$ for some $\xf \in \{0,1\}^\N$. Let $\xf \in \AM_\beta(s,\len x)^{-1}(x)$ and $\yf \in \AM_\beta(s,\len x)^{-1}(y)$. In particular, $\xf \in \AM_\beta(s,j)^{-1}(x_{1:j})$ and $\yf \in \AM_\beta(s,j)^{-1}(y_{1:j}) = \AM_\beta(s,j)^{-1}(x_{1:j})$, so $w_\beta(s,x_{1:j})$ is a prefix of both $\xf$ and $\yf$, and by Lemma \ref{lem:the set of preimages by the dajani algorithm is a cylinder}, we have that 
        \begin{equation}
            \xf_{1:h_\beta(s,j,\xf)} = \yf_{1:h_\beta(s,j,\yf)} = w_\beta(s,x).
        \end{equation}
        By Lemma \ref{lem:the random beta expansion algorithm reads only the first bits of its input}, we have that
        \begin{align}
            h_\beta(s,j,\xf) &= h_\beta(s,j,\yf) =: h_j\\
            \pi_1 K_\beta^j(s,\xf) &= \pi_1 K_\beta^j(s,\yf) =: r_j\\
            \pi_2 K_\beta^j(s,\xf) &= \sigma^{h_\beta(s,j,\xf)}(\xf) = \sigma^{h_j}(\xf)\\
            \pi_2 K_\beta^j(s,\yf) &= \sigma^{h_\beta(s,j,\yf)}(\yf) = \sigma^{h_j}(\yf),
        \end{align}
        Then, following the calculations that lead to (\ref{eq:0:proof:lem:characterization of the set of indices in the switch set by divergences in the tree of beta expansions}),
        \begin{align}
            x_{j+1} = \begin{cases}
                b \ \text{if} \ r_j \in E_\beta^{(b)}, \ b \in \{0,1\},\\
                \xf_{h_j+1} \ \text{if} \ r_j \in S_\beta,
            \end{cases}
        \end{align}
        and 
        \begin{align}
            y_{j+1} = \begin{cases}
                b \ \text{if} \ r_j \in E_\beta^{(b)}, \ b \in \{0,1\},\\
                \yf_{h_j+1} \ \text{if} \ r_j \in S_\beta,
            \end{cases}
        \end{align}
    \end{enumerate}
    Suppose that $r_j \notin S_\beta$. Then, $x_{j+1} = y_{j+1}$, which is a contradiction. Then, $r_j \in S_\beta$, so $j \in \SM_\beta(s,\len x, \xf)$. This concludes the proof.
\end{proof}

We finally define a computable function that allows to calculate $w_\beta(s,x)$, given $x$ and the set of all $\beta$-expansions of $s$.
\begin{definition}
    Let $f_{1,all \to tosses} : \{0,1\}^\ast \to \{0,1\}^\ast$ to be defined such that 
    \begin{equation}
        f_{1,all \to tosses}\left(\langle x, \langle \Sigma_\beta(s,\len x) \rangle \rangle\right) = w_\beta(s,x),
    \end{equation}
    for every $\beta \in (1,2)$, $s \in I_\beta$ and $x \in \{0,1\}^\ast$.
\end{definition}
\begin{lemma}\label{lem:f 1 all to noise is computable}
    $f_{1,all \to tosses}$ is computable.
\end{lemma}
\begin{proof}
    We construct an algorithm that computes $f_{1,all \to tosses}$ as follows.
    \begin{algorithm}[ht]
        \caption{Algorithm for computing $f_{1,all \to tosses}$}
        \label{alg:algorithm that computes f 1 all noise}
        \begin{algorithmic}[1]
            \Require $x \in \{0,1\}^\ast$, $\langle \Sigma_\beta(s, \len x) \rangle$
            \State $n \gets \len x$. \label{line:1:alg:algorithm that computes f 1 all noise}
            \State Use $\langle \Sigma_\beta(s, n) \rangle$ to compute $k = \NN_\beta(s,n)$.\label{line:2:alg:algorithm that computes f 1 all noise}
            \State Use $\langle \Sigma_\beta(s, n) \rangle$ to generate all the elements $y \in \Sigma_\beta(s, n)$.\label{line:3:alg:algorithm that computes f 1 all noise}
            \State $S \gets \{\Sigma_\beta(s, n)\}$\label{line:4:alg:algorithm that computes f 1 all noise}
            \For{$i = 1, \ldots ,n$}\label{line:5:alg:algorithm that computes f 1 all noise}
                \If {there is $y \in \Sigma_\beta(s,n)$ such that $y_{1:i} = x_{1:i}$ and $y_{i+1} \neq x_{i+1}$}\label{line:7:alg:algorithm that computes f 1 all noise}
                    \State $S \gets S \cup \{i\}$.\label{line:8:alg:algorithm that computes f 1 all noise}
                \EndIf
            \EndFor
            \State \Return $\coprod S$.\label{line:12:alg:algorithm that computes f 1 all noise}
        \end{algorithmic}
    \end{algorithm}

    \noindent Note that the set $S$ constructed by the algorithm exactly satisfies 
    \begin{equation}
        S = \left\{ i \in \{1,\ldots, \len x\}: \exists y \in \Sigma_\beta(s,\len x) \ s.t. \ y_{1:i-1} = x_{1:i-1}, y_i \neq x_i \right\}.
    \end{equation}
    By combining Lemma \ref{lem:the noise can be expressed from the time steps for which the trajectory is in switch set} and Lemma \ref{lem:characterization of the set of indices in the switch set by divergences in the tree of beta expansions}, we get that indeed
    \begin{equation}
        w_\beta(s,x) = \coprod S.
    \end{equation}
\end{proof}

\subsection{\texorpdfstring{Distribution of algorithmic complexity of $\beta$-expansions}{}}
We are now ready to move on to the proofs on algorithmic complexity. We combine the different functions defined previously to build a new multivalued function.
\begin{definition}\label{def:f beta to 2 + noise}
    Let $\beta \in (1,2)$, $f_{\beta \to tosses} : \{0,1\}^\ast \rightrightarrows \{0,1\}^\ast$ be defined as 
    \begin{equation}
        f_{\beta \to tosses}(x) = \left\{ f_{\beta,1+all\to tosses}(\langle x, y\rangle) : y \in f_{\beta, 1 \to all}(x)\right\},
    \end{equation}
    and $f_{\beta \to 2+tosses} : \{0,1\}^\ast \rightrightarrows \{0,1\}^\ast$ be defined by 
    \begin{equation}\label{eq:main:def:f beta to 2 + noise}
        f_{\beta \to 2+tosses}(x) = \left\{ \left\langle y,z \right\rangle : (y,z) \in f_{\beta \to 2}(x) \times f_{\beta \to tosses}(x) \right\},
    \end{equation}
    for all $x \in \{0,1\}^\ast$.
\end{definition}
\begin{lemma}
    Let $\beta \in (1,2)$, $s \in [0,1]$, $n \in \N$ and $\xf \in \Sigma_\beta(s)$. Then,
    \begin{equation}
        \langle \yf_{1:n}, w_\beta(s,\xf_{1:\nlogbeta})\rangle \in f_{\beta \to 2+tosses}(\xf_{1:\nlogbeta}),
    \end{equation}
    for all $n \in \N$, where $\yf$ is the greedy binary expansion of $s$.
\end{lemma}
\begin{proof}
    Let $\beta \in (1,2)$, $s \in [0,1]$, $n \in \N$, $\xf \in \Sigma_\beta(s)$ and define $\yf$ to be the greedy binary expansion of $s$. Then,
    \begin{equation}
        \yf_{1:n} \in f_{\beta \to 2}(\xf_{1:\nlogbeta}),
    \end{equation}
    and 
    \begin{equation}
        \langle \Sigma_\beta(s, \nlogbeta) \rangle \in f_{\beta, 1 \to all}(\xf_{1:\nlogbeta}),
    \end{equation}
    so 
    \begin{equation}
        w_\beta(s, \xf_{1:\nlogbeta}) = f_{\beta,1+all\to tosses}(\langle x, \langle \Sigma_\beta(s,\nlogbeta)\rangle\rangle) \in f_{\beta \to tosses}(\xf_{1:\nlogbeta}).
    \end{equation}
    Then,
    \begin{equation}
        (\yf_{1:n}, w_\beta(s, \xf_{1:\nlogbeta})) \in f_{\beta \to 2}(\xf_{1:\nlogbeta}) \times f_{\beta \to tosses}(\xf_{1:\nlogbeta}),
    \end{equation}
    and the Lemma follows from (\ref{eq:main:def:f beta to 2 + noise}).
\end{proof}
\begin{lemma}\label{lem:size of the multifunction converting a beta expansion into its noise and the binary expansion}
    Let $\beta \in (1,2)$ be an algebraic number. Then, for all $n \in \N$, and all $x \in \{0,1\}^n$,
    \begin{equation}
        \# f_{\beta \to 2+tosses}(x) \leq \left(\frac{1}{\beta-1} + 3\right) \left(\frac{1}{\beta-1} + 2\right)^2\frac{n^{2k_\beta}}{(L_\beta\Pi_\beta)^2}\left(L_\beta \Pi_\beta^+ \right)^{2n}.
    \end{equation}
\end{lemma}
\begin{proof}
    Let $\beta \in (1,2)$ be an algebraic number, and $x \in \{0,1\}^\ast$. First, by Lemma \ref{lem:upper bound on the cardinality of the beta to 2 conversion multifunction},
    \begin{equation}
        \#f_{\beta \to 2}(x) \leq \frac{1}{\beta-1} + 3.
    \end{equation}
    Second, by Lemma \ref{lem:size of the multifunction that generates all beta expansions from one},
    \begin{equation}
        \#f_{\beta, 1 \to all}(x) \leq \left(\frac{1}{\beta-1} + 2\right)^2\frac{n^{2k_\beta}}{(L_\beta\Pi_\beta)^2}\left(L_\beta \Pi_\beta^+ \right)^{2n},
    \end{equation}
    and therefore 
    \begin{equation}
        \#f_{\beta \to tosses}(x) \leq \left(\frac{1}{\beta-1} + 2\right)^2\frac{n^{2k_\beta}}{(L_\beta\Pi_\beta)^2}\left(L_\beta \Pi_\beta^+ \right)^{2n}.
    \end{equation}
    The proof is concluded by noting that
    \begin{align}
        \# f_{\beta \to 2+tosses}(x) &= \#f_{\beta \to 2}(x)  \times \#f_{\beta \to tosses}(x).
    \end{align}
\end{proof}
We now establish the subsequent results on algorithmic complexity. The equations are very involved, and we first make the following remark. Algebraic numbers are computable, i.e., for every algebraic number $\beta$, there exists a computable function $\varphi: \N \to \Q$ such that $\limi{n} \varphi(n) = \beta$. A consequence of this is that 
\begin{equation}
    K[x|\beta] = K[x] + \bigOx[1]{|x|}, \ \forall x \in \{0,1\}^\ast.
\end{equation}
In what follows, we focus on $\beta$ being algebraic, and we hence replace $K[\cdot |\beta]$ by $K[\cdot]$.
\begin{theorem}\label{thm:lower bound of kolmogorov complexity for beta expansions when beta is algebraic}
    Let $\beta \in (1,2)$ be an algebraic number, $s \in [0,1]$ and $n \in \N$. Denote by $\yf \in \{0,1\}^\N$ the greedy binary expansion of $s$. Then,
    \begin{align}\label{eq:algebraic:thm:lower bound of kolmogorov complexity for beta expansions when beta is algebraic}
        K[\langle \yf_{1:n}, w_\beta(s,\xf_{1:\nlogbeta})\rangle] &\leq K[\xf_{1:\nlogbeta}] + 2n \log_\beta\left(L_\beta \Pi_\beta^+\right)
        + \bigOx[\log n]{n},
    \end{align}
for all $\xf \in \Sigma_\beta(s)$. Moreover, if $\beta$ is a Pisot number,
\begin{equation}\label{eq:pisot:thm:lower bound of kolmogorov complexity for beta expansions when beta is algebraic}
    K[\langle \yf_{1:n}, w_\beta(s,\xf_{1:\nlogbeta})\rangle] \leq K[\xf_{1:\nlogbeta}] + \bigOx[1]{n},
\end{equation}
for all $\xf \in \Sigma_\beta(s)$. 
\end{theorem}
\begin{proof}
    Let $\beta \in (1,2)$ be an algebraic number, $s \in [0,1]$, $n \in \N$ and $\xf \in \Sigma_\beta(s)$. Then, by combination of 
    \begin{equation}
        \langle \yf_{1:n}, w_\beta(s,\xf_{1:\nlogbeta})\rangle \in f_{\beta \to 2+tosses}(\xf_{1:\nlogbeta}).
    \end{equation}
    Moreover, $f_{\beta \to 2+tosses}$ is computable. Hence, by Theorem \ref{thm:relatively computable multivalued functions and complexity},
    \begin{align}
        K[\langle \yf_{1:n}, w_\beta(s,\xf_{1:\nlogbeta})\rangle] &\leq K[\xf_{1:\nlogbeta}] + \log \#f_{\beta \to 2+tosses}(\xf_{1:\nlogbeta})\\
        & + \log \log \#f_{\beta \to 2+tosses}(\xf_{1:\nlogbeta}) + \bigOx[1]{n}.
    \end{align}
    Note that by Lemma \ref{lem:size of the multifunction converting a beta expansion into its noise and the binary expansion}, 
    \begin{equation}
        \log \#f_{\beta \to 2+tosses}(\xf_{1:\nlogbeta}) \leq 2n \log_\beta \left(L_\beta \Pi^+_\beta \right) + 2k_\beta \log_\beta n + \bigOx[1]{n}
    \end{equation}
    and 
    \begin{equation}
        \log \log \#f_{\beta \to 2+tosses}(\xf_{1:\nlogbeta}) \leq \log n + \log \log_\beta n + \bigOx[1]{n},
    \end{equation}
    hence (\ref{eq:algebraic:thm:lower bound of kolmogorov complexity for beta expansions when beta is algebraic}) follows. If specifically, $\beta$ is a Pisot number, then $L_\beta = \Pi^+_\beta = 1$ and $k_\beta = 0$, so 
    \begin{equation}
        \log \#f_{\beta \to 2+tosses}(\xf_{1:\nlogbeta}) \leq \bigOx[1]{n}
    \end{equation}
    and hence
    \begin{equation}
        \log \log \#f_{\beta \to 2+tosses}(\xf_{1:\nlogbeta})\leq \bigOx[1]{n}.
    \end{equation}
    (\ref{eq:pisot:thm:lower bound of kolmogorov complexity for beta expansions when beta is algebraic}) follows.
\end{proof}
\noindent In what follows, we use  
\begin{equation}
    K[x|y] := K[x|y0^\infty], \ \forall x,y \in \{0,1\}^\ast.
\end{equation}
\begin{corollary}\label{cor:lower bound of kolmogorov complexity for beta expansions when beta is algebraic}
    Let $\beta \in (1,2)$ be an algebraic number, $s \in [0,1]$ and $n \in \N$. Denote by $\yf \in \{0,1\}^\N$ the greedy binary expansion of $s$. Then,
    \begin{align}
        K[\xf_{1:\nlogbeta}] &\geq K[\yf_{1:n}] + K[w_\beta(s,\xf_{1:\nlogbeta})|\yf_{1:n}] - 2n \log_\beta\left(L_\beta \Pi_\beta^+\right) \nonumber\\ 
        &+ \bigOx[\log n]{n},\label{eq:main:cor:lower bound of kolmogorov complexity for beta expansions when beta is algebraic}
    \end{align}
    for all $\xf \in \Sigma_\beta(s)$.
\end{corollary}
\begin{proof}
    Let $\beta \in (1,2)$ be an algebraic number, $s \in [0,1]$ and $n \in \N$.
    By \cite[Theorem 2.8.2]{li2008introduction}, 
    \begin{equation}
        K[\langle y,x\rangle] \geq K[y] + K[x|y] + \bigOw[\log |\langle y,x\rangle|]{|\langle y,x \rangle|},
    \end{equation}
    for all $x,y \in \{0,1\}^\ast$. Therefore, in particular,
    \begin{equation}\label{eq:0:proof:cor:lower bound of kolmogorov complexity for beta expansions when beta is algebraic}
        K[\langle y,x\rangle] \geq K[y] + K[x|y] + \bigOx[\log n]{n},
    \end{equation}
    for all $x \in \{0,1\}^n$ and $y \in \{0,1\}^{\leq \nlogbeta}$. (\ref{eq:main:cor:lower bound of kolmogorov complexity for beta expansions when beta is algebraic}) follows immediately by combining (\ref{eq:0:proof:cor:lower bound of kolmogorov complexity for beta expansions when beta is algebraic}) and Theorem \ref{thm:lower bound of kolmogorov complexity for beta expansions when beta is algebraic}.
\end{proof}
\noindent We define a distribution of the Kolmogorov complexity of $\beta$-expansions. Let 
\begin{equation}
    \KK_\beta[s,n,k] := \left\{ x \in \Sigma_\beta(s,\nlogbeta) : K[x] \leq K[\yf_{1:n}] + k - 2n \log_\beta\left(L_\beta \Pi_\beta^+\right)\right\}.
\end{equation}

\begin{corollary}
    Let $\beta \in (1,2)$ be an algebraic number, $s \in [0,1]$ and $n,k \in \N$. Then, there exists $M \in \N$ such that
    \begin{equation}
        \# \KK_\beta[s,n,k] \leq 2^{k+1} n^M.
    \end{equation}
\end{corollary}
\begin{proof}
    Let $\beta \in (1,2)$ be an algebraic number, $s \in [0,1]$ and $n,k \in \N$. By Corollary \ref{cor:lower bound of kolmogorov complexity for beta expansions when beta is algebraic}, there exists $M \in \N$ such that 
    \begin{align}
        K[\xf_{1:\nlogbeta}] &\geq K[\yf_{1:n}] + K[w_\beta(s,\xf_{1:\nlogbeta})|\yf_{1:n}] - 2n \log_\beta\left(L_\beta \Pi_\beta^+\right)\\ 
        &- M\log n,
    \end{align}
    for all $\xf \in \Sigma_\beta(s)$. Define the set 
    \begin{equation}
        \WW_\beta[s,n,k] := \left\{ x \in \Sigma_\beta(s,\nlogbeta) : K[w_\beta(s,x)|\yf_{1:n}] \leq  k - M\log n\right\}.
    \end{equation}
    Then, by Corollary \ref{cor:lower bound of kolmogorov complexity for beta expansions when beta is algebraic}, we have 
    \begin{equation}
        \KK_\beta[s,n,k] \subseteq \WW_\beta[s,n,k].
    \end{equation}
    For $z \in \{0,1\}^\ast$, let $z^\ast$ be the canonical sequence for $z$ with respect to $\yf_{1:n}$, i.e., the lexicographically maximal sequence satisfying $|z^\ast| = K[z|\yf_{1:n}]$. Note that $z \mapsto z^\ast$ injective. Note that, by injectivity of $x\mapsto w_\beta(s,x)$, we have
    \begin{align}
        \#\WW_\beta[s,n,k] &= \#\left\{ x \in \Sigma_\beta(s,\nlogbeta) : K[w_\beta(s,x)|\yf_{1:n}] \leq  k + M\log n\right\}\\
        &\overset{(a)}{=} \#\left\{ w_\beta(s,x) : x \in \Sigma_\beta(s,\nlogbeta),\right.\\
        & \hspace{0.45in} \left. |K[w_\beta(s,x)|\yf_{1:n}] \leq  k + M\log n\right\}\\
        &\leq \#\left\{ z \in \{0,1\}^\ast : K[z|\yf_{1:n}] \leq k + M \log n\right\}\\
        &\overset{(b)}{\leq} \{ z^\ast \in  : z \in \{0,1\}^\ast,  |z^\ast| \leq k + M \log n\}\\
        &= \# \{0,1\}^{\leq k + M \log n} = \sum_{i=0}^{k + M\log n} 2^i\\
        & = 2^{k + M \log n+1} - 1 \leq 2^{k+1} n^M,
    \end{align}
    where (a) is by injectivity of $x \mapsto w_\beta(s,x)$, (b) follows by injectivity of $z \mapsto z^\ast$.
\end{proof}
\endgroup







\appendix 


\section{Appendix on weak-star limits}

\label{sec:app:weak star limits}

In this part, we collect useful statements that help the understanding of weak limit measures. Let $X$ be a locally compact Hausdorff topological space (see \cite[p31, p104]{lee2011IntroductionTopologicalManifolds}). We denote by $\BB$ the Borel $\sigma$-algebra on $X$, i.e., the $\sigma$-algebra generated by the open subsets of $X$ (see \cite[Section 15]{halmos1950MeasureTheory}). We recall the definition of a regular Borel measure.
\begin{definition}\cite[Section 52]{halmos1950MeasureTheory}
    A regular Borel measure $\mu : \BB \to [0, \infty]$ is a measure that is finite on compact sets, i.e., $\mu(C) < \infty$ for every compact set $C$, and that satisfies
        \begin{equation}
            \mu(A) = \inf \{\mu(U) : A \subseteq U, U \text{ is open}\} =  \sup \{\mu(U) : U \subseteq A, U \text{ is compact}\}.
        \end{equation}
\end{definition}

Regular Borel measures are a central element in the following Lemmata, that are key ingredients of the proof of Theorem \ref{thm:asymptotic behavior of card f x when the Bernoulli convolution is absolutely continuous}. Define the set of compactly supported continuous functions from $X$ to $\R$ as \cite[Section 55]{halmos1950MeasureTheory}
\begin{equation}
    \CC_c(X) := \left\{ f: X \to \R : f \text{ is continuous, and } \overline{\{x : f(x) \neq 0\}}\text{ is compact} \right\}.
\end{equation}
Note that if $X$ is compact, $\CC_c(X)$ is simply the set of continuous functions. 
\begin{lemma}
    Let $(\mu_n : \BB \to [0,\infty])_{n \in \N}$ be a sequence of regular Borel measures converging weakly, i.e., $\left(\int_X f d\mu_n\right)_{n \in \N}$ is converging in $\R$, for every $f \in \CC_c(X)$. Then, there exists a regular Borel measure $\mu : \BB \to [0,\infty]$ such that
    \begin{equation}
        \int_X f d\mu = \lim_{n\to \infty}\int_X f d\mu_n, \ \forall f \in \CC_c(X).
    \end{equation}
\end{lemma}
\begin{proof}
    A map $\Lambda : \CC_c(X) \to \R$ is said to be positive if $\Lambda(f) \geq 0$ when $f \geq 0$ and linear if $\Lambda(\alpha f + g) = \alpha \Lambda(f) + \Lambda(g)$ for all $\alpha \in \R$, $f,g \in \CC_c(X)$. For every measure $\mu : \BB \to [0,\infty]$, the functional defined by $\Lambda_\mu : f \mapsto \int_X f d\mu$ is obviously positive and linear. Define the functional $\Lambda : \CC_c(X) \to \R$ as 
    \begin{equation}
        \Lambda (f) := \lim_{n\to \infty}\int_X f d\mu_n, \ \forall f \in \CC_c(X).
    \end{equation}
    $\Lambda$ is well-defined since $(\mu_n)_{n \in \N}$ converges weakly. Moreover, $\Lambda$ is obviously positive and linear, by positivity and linearity of the limit. Further, there is a standard result, known as Riesz representation theorem, that states that for every positive linear functional $\Lambda : \CC_c(X) \to \R$, there exists a unique regular Borel measure $\mu$ such that $\Lambda = \Lambda_\mu$ (see  \cite[Section 56, Theorems D and E]{halmos1950MeasureTheory}). This concludes the proof.
\end{proof}

\begin{lemma}\label{app:lem:sequence of weakly convergent measures can be upper bounded by the limit on compact sets}
    Suppose that $X$ is compact. Let $(\mu_n)_{n \in \N}$ be a sequence of regular Borel measures converging weakly to a Borel regular measure $\nu$, i.e.,
    \begin{equation}
        \int_Xfd\mu = \lim_{n \to \infty} \int_X fd\mu_n, \ \forall f \in \CC_c(X).
    \end{equation}
    Then, for every compact set $C$, we have
    \begin{equation}
        \limsup_{n \to \infty} \mu_n(C) \leq \mu(C).
    \end{equation}
\end{lemma}
\begin{proof} 
    Let $(\mu_n)_{n \in \N}$ be a sequence of regular Borel measures converging weakly to the Borel regular measure $\mu$ and $C$ be a closed set. Let $U$ be an open set, so that $C \subseteq U$. Since $U$ is open, then $X \backslash U$ is closed and disjoint from $C$. Since $X$ is compact and Hausdorff, then it is normal (see \cite[p111]{lee2011IntroductionTopologicalManifolds} and \cite[Exercise 4.78]{lee2011IntroductionTopologicalManifolds}), so we can apply Urysohn's Lemma \cite[Lemma 4.82]{lee2011IntroductionTopologicalManifolds}: there exists a continuous function $f_U : X \to [0,1]$ such that $f(C) = \{1\}$ and $f(X\backslash U) = \{0\}$. Note that since 
    \begin{equation}
        \chi_C(x) \leq f_U(x) \leq \chi_U(x), \ \forall x \in X,
    \end{equation}
    where $\chi_A$ denotes the characteristic function of $A \in \BB$, then
    \begin{equation}
        \mu_n(C) = \int_X \chi_C d\mu_n \leq \int_Xf_Ud\mu_n,
    \end{equation}
    and therefore
    \begin{equation}
        \limsup_{n \to \infty}\mu_n(C) \leq \limsup_{n \to \infty}\int_Xf_Ud\mu_n = \int_Xf_Ud\mu \leq \int_X\chi_U d\mu = \mu(U).
    \end{equation}
    Since the left-hand side of the inequality does not depend on $U$, we get
    \begin{equation}
        \limsup_{n \to \infty}\mu_n(C) \leq \inf_{\substack{U \text{ open},\\ U \supseteq C}} \mu(U) = \mu(C),
    \end{equation}
    where the last equality follows by Borel regularity of $\mu$.
\end{proof}

Since $\R$ is a locally compact Hausdorff topological space, and every closed interval $[a,b]$, $a,b \in \R$ is a compact subset of $\R$, we can use the above Lemma to prove Theorem \ref{thm:asymptotic behavior of card f x when the Bernoulli convolution is absolutely continuous}. 


\section{Proof of the separation lemma}

\label{sec:app:separation lemma}

In this section, we prove Lemma \ref{lem:generalization of the garsia separation lemma}. We first introduce a succession classical algebraic results, that ultimately yield the Lemma. We assume that the reader is familiar to elementary notions of algebra, such as the definition of a field. We refer to the standard books of algebra, such as \cite{lang2002Algebra}.
\begin{lemma}\label{lem:from polynomials in q, one of the roots of the polynomial not being a root of another polynomial implies that all the roots are not the rroots of the other}
    Let $P,Q \in \Q[X]$, with $P$ irreducible in $\Q[X]$. Suppose that one of the roots of $P$ is not a root of $Q$. Then, none of the roots of $P$ is a root of $Q$.
\end{lemma}
\begin{proof}
    We show the contraposition of the statement. Let $P,Q \in \Q[X]$, with $P$ irreducible in $\Q[X]$. We will show that if a root of $P$ is also a root of $Q$, then all the roots of $P$ must be roots of $Q$. First, note that since $P$ is irreducible in $\Q[X]$, its dividors in $\Q[X]$ are either polynomials of the form $\alpha P$ with $\alpha \in \Q$, or constant polynomials in $\Q[X]$. Then, the greatest common dividors of $P$ and $Q$, by being dividors of $P$, must be also be either of the form $\alpha P$ or constant polynomials. In particular, this means that either all the greatest common dividors of $P$ and $Q$ in $\Q[X]$ are of degree $\deg P$, or they all are of degree $0$.

    Now, assume that one root $\alpha \in \C$ of $P$ is also a root of $Q$, then we know that the polynomial $X - \alpha$ divides both $P$ and $Q$ in $\C[X]$, hence the greatest common dividors of $P$ and $Q$ in $\C[X]$ are at least of degree $1$. There is a standard result in algebra that states that for any two fields $\K,\Lb$ satisfying $\K \subseteq \Lb$, the greatest common dividors of two polynomials in $\K[X]$ are also greastest common dividors in $\Lb[X]$ \cite[Section III-8, Proposition 3]{mazet1996capes}. In particular, this means that if the greatest common dividors of $P$ and $Q$ in $\Lb[X]$ are of degree $d$, then the greatest common dividors of $P$ and $Q$ in $\K[X]$ also are of degree $d$. Since $\Q$ and $\C$ are fields such that $\Q \subseteq \C$, then by the discussion above the degree of the greatest common dividors of $P$ and $Q$ in $\Q[X]$ are of degree at least $1$, which implies that they must be of the form $\alpha P$, $\alpha \in \Q$, i.e., there exists a polynomial $A \in \Q[X]$ such that $Q = AP$. Consequently, for any root $r$ of $P$, we have 
    \begin{equation}
        Q(r) = A(r) P(r) = 0,
    \end{equation}
    i.e., $r$ is a root of $Q$.
\end{proof}

The next statement emerges from the theory of polynomial resultants. We introduce this theory along the lines of \cite[Section IV.8]{lang2002Algebra}. Let $P,Q \in \C[X]$ of respective degrees $n$ and $m$. We denote by $p_0,\ldots,p_n \in \C$ the coefficients of $P$, and by $q_0,\ldots,q_m \in \C$ the coefficients of $Q$. The resultant $Res(P,Q)$ is defined to be the determinant of the matrix 

\begin{equation}
    \RR(P,Q) = 
    \begin{tikzpicture}[baseline=(current bounding box.center)]
        \matrix (m) [matrix of math nodes,left delimiter={(},right delimiter={)}] {
        p_n & p_{n-1} & \ldots & p_0 & 0 & 0 & 0 \\
        0 & p_n & p_{n-1} & \ldots & p_0 & 0 & 0 \\
        & \vdots & & & & &\\
        0& 0 & 0 & p_n & p_{n-1} & \ldots & p_0 \\
        q_m & \ldots & q_0 & 0 & 0 & 0 & 0 \\
        0 & q_m & \ldots & q_0 & 0 & 0 & 0 \\
        & \vdots & & & & &\\
        0& 0 & 0 & 0 & q_m & \ldots & q_0 \\
    };
    \draw[<->] ([xshift=-1.5em, yshift=-1em]m-8-1.south) -- node[below] {$m$} +(6.7,0);

    \draw[<->] ([xshift=1.5em, yshift=1em]m-1-7.east) -- node[right] {$n$} +(0,-5.7);
\end{tikzpicture}
\end{equation}
Further denote $\alpha_1,\alpha_2,\ldots,\alpha_{n} \in \C$ be the roots of $P$, counted with multiplicities. Then, following \cite[Section IV, Proposition 8.3]{lang2002Algebra}, we have
\begin{equation}\label{eq:polynomial resultant is equal to the product of one of the polynomials on the roots of the other}
    Res(P,Q) = p_0^m \prod_{i=1}^n Q(\alpha_i).
\end{equation}
The next lemma bounds the resultant of two polynomials with integer coefficients from below.
\begin{lemma}\label{lem:the polynomial result of integer polynomials is bigger than one}
    Let $P,Q \in \Z[X]$, with $P$ irreducible in $\Z[X]$. Suppose that one of the roots of $P$ is not a root of $Q$. Then,
    \begin{equation}
        |Res(P,Q)| \geq 1.
    \end{equation}
\end{lemma}
\begin{proof}
    Let $P,Q \in \Z[X]$, with $P$ irreducible in $\Z[X]$. Then, in particular, $P,Q \in \Q[X]$ and $P$ is irreducible in $\Q[X]$. Suppose that one of the roots of $P$ is not a root of $Q$. Then, by Lemma \ref{lem:from polynomials in q, one of the roots of the polynomial not being a root of another polynomial implies that all the roots are not the rroots of the other}, none of the roots of $P$ are roots of $Q$. By (\ref{eq:polynomial resultant is equal to the product of one of the polynomials on the roots of the other}), it follows that
    \begin{equation}\label{eq:0:proof:lem:the polynomial result of integer polynomials is bigger than one}
        Res(P,Q) = p_0^n \prod_{i=1}^n Q(\alpha_i) \neq 0.
    \end{equation}
    Moreover, since $Res(P,Q)$ is defined to be the determinant of the matrix $\RR(P,Q)$, and that this matrix has only coefficients in $\Z$, then $Res(P,Q) \in \Z$. Combined with (\ref{eq:0:proof:lem:the polynomial result of integer polynomials is bigger than one}), we get that $Res(P,Q) \in \Z \backslash \{0\}$, and hence that $|Res(P,Q)| \geq 1$.
\end{proof}

We are now ready to prove Lemma \ref{lem:generalization of the garsia separation lemma}. For an algebraic number $\beta \in (1,2)$, recall that $L_\beta$ denotes the leading coefficient of its minimal polynomial $P_\beta$. Fix $n \in \N$. For all $x,y \in \{0,1\}^n$, we can express the quantity $\left|\sum_{i=1}^n x_i \beta^{-i} - \sum_{i=1}^n y_i \beta^{-i}\right|$ as $|\beta^{-n}A_{xy}(\beta)|$, where $A_{xy}$ is a poynomial of degree $n-1$, and of coefficients $a_0 := x_n - y_n$, $a_1 := x_{n-1} - y_{n-1}, \ldots, a_{n-1} := x_{1} - y_{1}$. There are two cases.
\begin{enumerate}[label=(\alph*)]
    \item $x \sim_\beta y$, i.e., $A_{xy}(\beta) = 0$.
    \item $x \not\sim_\beta y$, so $A_{xy}(\beta) \neq 0$, i.e., $\beta$ is not a root of $A_{xy}$. Lemma \ref{lem:the polynomial result of integer polynomials is bigger than one} yields
\begin{equation}
    |Res(P_\beta,A_{xy})| \geq 1,
\end{equation}
which, combined with (\ref{eq:polynomial resultant is equal to the product of one of the polynomials on the roots of the other}), delivers
\begin{equation}\label{eq:0:proof:lem:generalization of the garsia separation lemma}
    \left|L_\beta^{n-1} A_{xy}(\beta) \prod_{z \in G_\beta} A_{xy}(z)\right| \geq 1 \iff |A_{xy}(\beta)| \geq \frac{1}{L_\beta^{n-1}} \prod_{z \in G_\beta} \frac{1}{\left|A_{xy}(z)\right|} \cdot
\end{equation}
We now give an upper bound on $|A_{xy}(z)|$, for $z \in \C$. First note that
\begin{equation}
    |A_{xy}(z)| = \left|\sum_{i=0}^{n-1} (x_i - y_i) z^i \right| \leq \sum_{i=0}^{n-1} |x_i-y_i| |z|^i \leq \sum_{i=0}^{n-1} |z|^i, \ \forall z \in \C.
\end{equation}
Then, we can split the study in three subcases.
\begin{enumerate}[label=(\roman*)]
    \item Suppose that $|z| < 1$. Then,
    \begin{equation}
        |A_{xy}(z)| \leq \sum_{i=0}^{n-1} |z|^i = \frac{1 - |z|^n}{1 - |z|} \leq \frac{1}{1 - |z|}.
    \end{equation}
    \item Suppose that $|z| = 1$. Then,
    \begin{equation}
        |A_{xy}(z)| \leq \sum_{i=0}^{n-1} |z|^i = n.
    \end{equation}
    \item Suppose that $|z| > 1$. Then,
    \begin{equation}
        |A_{xy}(z)| \leq \sum_{i=0}^{n-1} |z|^i = \frac{|z|^n-1}{|z|-1} \leq \frac{|z|^n}{|z|-1}.
    \end{equation}
\end{enumerate}
By incorporation of these results in (\ref{eq:0:proof:lem:generalization of the garsia separation lemma}), we get
\begin{equation}
    |A_{xy}(\beta)| \geq \frac{1}{L_\beta^{n-1}}  \frac{\prod_{z \in G_\beta} |1-|z||}{n^{k_\beta}\prod_{z \in G^+_\beta} |z|^n} = \frac{L_\beta\Pi_\beta}{n^{k_\beta} (L_\beta \Pi_\beta^+)^n}.
\end{equation}
Finally, as $\left|\sum_{i=1}^n x_i \beta^{-i} - \sum_{i=1}^n y_i \beta^{-i}\right| = |\beta^{-n}A_{xy}(\beta)|$, we get
\begin{equation}
    \left|\sum_{i=1}^n x_i \beta^{-i} - \sum_{i=1}^n y_i \beta^{-i}\right| \geq \frac{L_\beta\Pi_\beta}{n^{k_\beta} (L_\beta \beta \Pi_\beta^+)^n},
\end{equation}
thereby proving Lemma \ref{lem:generalization of the garsia separation lemma}.
\end{enumerate}


\bibliographystyle{IEEEtranS}

\end{document}